\documentclass{article}
\usepackage[preprint]{neurips_2025}

\usepackage[utf8]{inputenc} 
\usepackage[T1]{fontenc}    
\usepackage{hyperref}       
\usepackage{url}            
\usepackage{booktabs}       
\usepackage{amsfonts}       
\usepackage{nicefrac}       
\usepackage{microtype}      
\usepackage{xcolor}         
\usepackage{amssymb}
\usepackage{mathtools}

\usepackage{algorithm}
\usepackage[noend]{algpseudocode}

\usepackage{amsmath}
\usepackage{amsthm}
\usepackage{thmtools}
\usepackage{thm-restate}
\usepackage{subcaption}

\theoremstyle{plain}
\newtheorem{theorem}{Theorem}[section]

\theoremstyle{definition}
\newtheorem{definition}{Definition}[section]
\newtheorem{assumption}[theorem]{Assumption}

\newtheorem{remark}[theorem]{Remark}

\newtheorem{lemma}[theorem]{Remark}
\newtheorem{corollary}[theorem]{Remark}

\begin{document}

\title{Analysis of Shuffling Beyond Pure Local Differential Privacy}



\author{%
  Shun Takagi \\
  LY Corporation \\
  \texttt{shutakag@lycorp.co.jp} \\
  \And
  Seng Pei Liew \\
  LY Corporation \\
  \texttt{sliew@lycorp.co.jp} \\
}
\maketitle

\begin{abstract}
Shuffling is a powerful way to amplify privacy of a local randomizer in private distributed data analysis.
Most existing analyses of how shuffling amplifies privacy are based on the pure local differential privacy (DP) parameter $\varepsilon_0$.
This paper raises the question of whether $\varepsilon_0$ adequately captures the privacy amplification.
For example, since the Gaussian mechanism does not satisfy pure local DP for any finite $\varepsilon_0$, does it follow that shuffling yields weak amplification?
To solve this problem, we revisit the privacy blanket bound of Balle et al.\ (the \emph{blanket divergence}) and develop a direct asymptotic analysis that bypasses $\varepsilon_0$.
Our key finding is that, asymptotically, the blanket divergence depends on the local mechanism only through a single scalar parameter $\chi$ and that this dependence is monotonic.
Therefore, this parameter serves as a proxy for shuffling efficiency, which we call the \emph{shuffle index}.
By applying this analysis to both upper and lower bounds of the shuffled mechanism's privacy profile, we obtain a band for its privacy guarantee through shuffle indices. 
Furthermore, we derive a simple structural, necessary and sufficient condition on the local randomizer under which this band collapses asymptotically.
$k$-RR families with $k\ge3$ satisfy this condition, while for generalized Gaussian mechanisms the condition may not hold but the resulting band remains tight. 
Finally, we complement the asymptotic theory with an FFT-based algorithm for computing the blanket divergence at finite $n$, which offers rigorously controlled relative error and near-linear running time in $n$, providing a practical numerical analysis for shuffle DP.

\end{abstract}

\section{Introduction}

Local privacy~\cite{duchiLocalPrivacyStatistical2013}, in which each user applies a local randomizer to their data before sending it, removes the need for a trusted curator and thus enables distributed private data analysis.
However, the price of eliminating trust is often a substantial loss in accuracy.
Shuffling, viewed as an anonymization layer that breaks the link between users and their messages, offers a compelling middle ground.
Since it only anonymizes messages across users, it preserves the utility of the analytics tasks already supported in the local model, such as frequency estimation, histograms/heavy hitters, and mean estimation.
Moreover, a growing line of work shows that this simple step can dramatically amplify privacy, allowing protocols to approach trusted-curator privacy-utility trade-offs~\cite{feldmanHidingClonesSimple2022,ballePrivacyBlanketShuffle2019, erlingssonAmplificationShufflingLocal2019}.
Given its practical utility, shuffling has become a key primitive for private distributed data collection and analytics~\cite{liewNetworkShufflingPrivacy2022,luoRM2AnswerCounting2025}.

From a practical data-management perspective, we need precise analysis of how and to what extent shuffling amplifies privacy to choose the best local randomizer for a given task.
However, existing tools are local randomizer-agnostic and largely tailored to the pure local differential privacy (DP) setting~\cite{ballePrivacyBlanketShuffle2019,erlingssonAmplificationShufflingLocal2019} (i.e., $(\varepsilon_0,0)$-local DP) and this $\varepsilon_0$-centric viewpoint has two important limitations.

First, the parameter $\varepsilon_0$ summarizes the local privacy
guarantee, but it is a rather crude descriptor from the viewpoint of
shuffling: it largely ignores the structural properties of the local
mechanism that actually govern its privacy amplification.
In particular, existing general bounds typically do not distinguish
between, say, a Laplace mechanism with a given noise scale, a
$k$-randomized response ($k$-RR) mechanism with a given $k$, or other
natural families of local mechanisms.
As a consequence, generic upper bounds expressed solely in terms of
$\varepsilon_0$ can be quite loose.
Indeed, several works have shown that by exploiting the specific
structure of a given mechanism one can obtain significantly tighter
bounds on the effect of shuffling
(e.g., for $k$-RR one can show that privacy amplification becomes
stronger as $k$ increases, even when $\varepsilon_0$ is kept
fixed~\cite{feldmanHidingClonesSimple2022,biswasTightDifferentialPrivacy2024}).
Yet, we currently lack a general understanding of what intrinsic
properties of a local mechanism govern its ``shuffle efficiency.''

Second, the restriction to pure local DP mechanisms excludes a
wide range of natural and practically relevant local randomizers.
Pure local DP is a convenient sufficient condition for privacy
amplification by shuffling, but it is far from necessary.
In practice, one often encounters approximate-DP local mechanisms or
mechanisms that do not satisfy any DP condition.
Surprisingly, even for the Gaussian mechanism, arguably one of the most prominent mechanisms in central DP, it is not known how to precisely
characterize privacy amplification by shuffling.
Balle et al.~\cite{ballePrivacyBlanketShuffle2019} already pointed out
that the Gaussian case is technically challenging, and subsequent work
by Koskela et al.~\cite{koskelaNumericalAccountingShuffle2022} observed that numerical accounting for Gaussian shuffling is nontrivial.
As a result, to date, existing results for shuffling of Gaussian mechanisms are essentially limited to lower bounds~\cite{liewShuffleGaussianMechanism2022,chuaHowPrivateAre2024}.
While there are attempts to generalize shuffling analysis to approximate local DP (i.e., $(\varepsilon_0,\delta_0>0)$-local DP) in a mechanism-agnostic fashion~\cite{feldmanHidingClonesSimple2022, cheuDistributedDifferentialPrivacy2019}, such general bounds tend to be overly pessimistic; they often show little visible amplification for moderate values of $n$, because they are driven by worst-case mechanisms that do not reflect the structure of the specific mechanism.

\paragraph{The privacy blanket and blanket divergence.}
A central tool in the analysis of shuffle DP is the \emph{privacy blanket} introduced by Balle et al.~\cite{ballePrivacyBlanketShuffle2019}.
Given a local randomizer $\mathcal R$, the blanket distribution is a data-independent distribution $\mathcal R_{\mathrm{BG}}$ defined using $\mathcal R$.
Balle et al.\ showed that the privacy profile $\delta(\varepsilon)$~\cite{ballePrivacyAmplificationSubsampling2018a} of the shuffled mechanism can be upper-bounded by a divergence computed under the blanket distribution; we refer to this quantity as the \emph{blanket divergence}.
This blanket-divergence bound is currently the best-known general upper bound on privacy amplification by shuffling~\cite{suDecompositionBasedOptimalBounds2025}.

However, the blanket divergence itself is still poorly understood, largely due to its analytical and numerical intractability.
In particular, Balle et al.~\cite{ballePrivacyBlanketShuffle2019} mainly
derive generic upper bounds via Hoeffding- and Bennett-type concentration
inequalities based on $\varepsilon_0$.
Consequently, these bounds can be quite loose and fail to capture how the blanket divergence depends on the structure of the underlying local randomizer.

This is partly because Balle et al.\ explicitly aim to obtain finite-$n$ upper bounds.
More generally, carrying out a precise finite-$n$ analysis of shuffling through the blanket
divergence is intrinsically challenging: DP requires a worst-case guarantee over all neighboring datasets (i.e., establishing a uniform upper bound over all neighboring input pairs).
Identifying the pairs that maximize the blanket divergence at finite $n$ is highly nontrivial in general.
In contrast, in the $\varepsilon_0$-local DP setting, the single parameter $\varepsilon_0$ already summarizes the worst case over all neighboring input pairs, so concentration-based bounds can be expressed directly in terms of $\varepsilon_0$ without explicitly resolving the supremum over all neighboring input pairs.

\paragraph{Our goal and approach.}
In light of the above limitations, our goal is to circumvent the apparent finite-$n$ bottlenecks of generic concentration-based analyses by taking a different route, namely an asymptotic but direct analysis of the blanket divergence. 
While this sacrifices exactness at any fixed $n$, it is tight and allows us to address the conceptual and practical limitations discussed above.

Concretely, the relevant quantity in blanket divergence can be written as a sum of $n$ i.i.d.\ random variables, so that its behavior is driven by a central limit theorem (CLT)-type phenomenon.
We leverage asymptotic expansions of the CLT~\cite{petrovSumsIndependentRandom1975} to obtain a sharp leading-order characterization of the blanket divergence.
Formally, in the regime\footnote{This corresponds to the moderate deviation regime required to apply the asymptotic expansion of the CLT.}
$\varepsilon_n = \omega(n^{-1/2})$ and
$\varepsilon_n = O(\sqrt{\log n / n})$, the blanket divergence admits the asymptotic expansion
\[
  \mathcal D^{\mathrm{blanket}}
  =
  \varphi\bigl(\chi\,\varepsilon_n\sqrt{n}\bigr)\,
  \left(
    \frac{1}{\chi^3}\,
    \frac{1}{\varepsilon_n^2 n^{3/2}}
  \right)\bigl(1+o(1)\bigr),
\]
where $\varphi$ denotes the PDF of the standard normal distribution and $\chi$ is a mechanism-dependent constant easily computed from the structure of the local randomizer.
Thus $\chi$ fully governs how well a given local randomizer interacts with
shuffling.
Crucially, this dependence is monotonic: a larger $\chi$ results in a smaller divergence, thereby implying stronger privacy guarantees.
This means that $\chi$ serves as a single-number index of privacy after shuffling; hence, we call it the \emph{shuffle index}, which substitutes for $\varepsilon_0$ as a more refined descriptor of the local randomizer's shuffle efficiency.
In other words, we can choose the best local randomizer for shuffling by choosing the one with the largest shuffle index $\chi$.

Moreover, we obtain an asymptotic characterization of the worst-case blanket divergence and hence the shuffled mechanism's $(\varepsilon\approx\frac{1}{\chi}\sqrt{\frac{\log n}{n}},\delta\approx\tfrac{\alpha}{n})$ guarantee of DP for any $\alpha>0$.
Because we can express not only upper bounds but also lower bounds in terms of the blanket divergence, we can quantitatively assess the tightness of the upper bound.
In particular, we derive a simple structural necessary and sufficient condition on the local randomizer under which the blanket divergence yields an asymptotically optimal characterization (i.e., matches the upper bound and the lower bound) of shuffling.
This condition is satisfied, for example, by the $k$-RR family with $k\geq 3$. 
Moreover, we show that the Laplace and Gaussian mechanisms do not satisfy the condition but the upper and lower bounds narrow in higher privacy regimes.

On the practical side, asymptotic expansions alone are not sufficient for numerical accounting at finite $n$.
We therefore complement our asymptotic analysis with a finite-$n$ algorithm for computing the blanket divergence.
Our algorithm is based on an FFT approximation of the distribution, combined with explicit control of the truncation, discretization, and aliasing errors.
Leveraging the asymptotic expansion, we show that the parameters can be tuned by a parameter $\eta$ so that the relative error is $O(\eta)$, while the running time scales as $\widetilde O(n/\eta)$.

\subsection*{Related Work and Our Contributions}
\label{sec:related-work}

A line of work seeks to go beyond the $\varepsilon_0$-centric view by exploiting additional structure of the local randomizer, typically in the form of a small number of scalar parameters.
Feldman et al.~\cite{feldmanStrongerPrivacyAmplification2023a} introduced a stronger clone paradigm and showed that, for mechanisms such as $k$-RR, one can parameterize the effect of shuffling using $\varepsilon_0$ together with two additional scalar quantities, yielding substantially tighter bounds than $\varepsilon_0$-only analyses, although even in the $k$-RR case exact optimality is not established.
Building on this, Wang et al.~\cite{wangPrivacyAmplificationShuffling2024b, wangShuffleModelDifferential2025a} refine the stronger clone-based analysis, using two scalar parameters, and provide sufficient conditions on the local randomizer under which their bounds become optimal. 
A subtle point is that the proof of Lemma~4.5 in~\cite{wangPrivacyAmplificationShuffling2024b} implicitly assumes the existence of a nontrivial post-processing map (see the errata of \cite{feldmanStrongerPrivacyAmplification2023a}).
Outside specific mechanism classes such as $k$-RR where no such map is needed (see Theorem 7.1 of \cite{feldmanStrongerPrivacyAmplification2023a}), the results should therefore be interpreted with some care.

Several works have proposed numerical methods for evaluating shuffle DP guarantees.
Koskela et al.~\cite{koskelaNumericalAccountingShuffle2022} use the clone paradigm together with FFT to approximate the privacy loss of shuffled mechanisms.
Their approach, however, is restricted to pure local DP mechanisms and inherits the looseness of the clone paradigm, as shown by Su et al.~\cite{suDecompositionBasedOptimalBounds2025}.
Su et al.~\cite{suDecompositionBasedOptimalBounds2025} propose a numerical analysis of the blanket divergence based on FFT, similar in spirit to ours.
However, they still assume pure local DP and do not provide rigorous control of the relative error; in addition, the running time scales as $O(n^2)$ in the worst case when targeting a constant additive error.
Biswas et al.~\cite{biswasTightDifferentialPrivacy2024} give an exact combinatorial method for computing the privacy guarantees of shuffled $k$-RR, but their analysis focuses on fixed neighboring database pairs rather than general $(\varepsilon,\delta)$-DP guarantees.

\paragraph{Our contributions.}
Compared with the above works, our contributions are as follows:
\begin{itemize}
  \item We are the first to give a unified shuffle DP analysis that does not assume pure local DP and applies to arbitrary local randomizers under mild regularity assumptions.

  \item We theoretically analyze the blanket divergence including the exact leading constant factor and summarize its leading behavior by a single shuffle index $\chi$, yielding mechanism-aware bounds and a necessary-and-sufficient optimality condition.

  \item We develop an FFT-based blanket-divergence accountant with rigorous relative-error $O(\eta)$ guarantees and near-linear $\widetilde O(n/\eta)$ running time.

  \item We empirically validate our asymptotic and FFT-based analysis, and in a distribution-estimation task we show that generalized Gaussian mechanisms achieve favorable privacy-utility trade-offs compared with pure local DP mechanisms.
\end{itemize}

\section{Preliminaries}

\paragraph{Notations}
$\Phi$ and $\varphi$ denote the CDF and PDF of the standard normal distribution.
$\Gamma(\cdot)$ denotes the Gamma function,
$\Gamma(t) = \int_0^\infty u^{t-1}\mathrm{e}^{-u}\,du$ for $t>0$.
$W(\cdot)$ denotes the principal branch of the Lambert
$W$-function.
Throughout the paper, whenever an expression can be written in terms of $\Phi$, $\Gamma$, $W$, and elementary functions, we will also refer to it as being available in closed form.
$\mathbb{E}[X]$ and $\mathrm{Var}(X)$ denote the expectation and variance of a random variable $X$.
$\mathbb{R}$ and $\mathbb{N}$ denote the set of real numbers and natural numbers, respectively.
$\mathbf{1}\{\cdot\}$ denotes the indicator function, which equals $1$ if the condition inside the brackets is satisfied, and $0$ otherwise. We use standard asymptotic notations $O(\cdot)$, $o(\cdot)$, $\Omega(\cdot)$, $\omega(\cdot)$, and $\Theta(\cdot)$.

\subsection{Differential Privacy (DP)}

We first recall the notion of DP and express it in a form based on the hockey-stick divergence.

\begin{definition}[Hockey-stick divergence]
Let $\mathcal Y$ be an output space equipped with a reference measure
(counting measure in the discrete case, Lebesgue measure in the continuous
case).
Let $P$ and $Q$ be probability distributions on $\mathcal Y$ that admit
densities $p$ and $q$ with respect to this reference measure, and fix a
parameter $\alpha \ge 1$.
The \emph{hockey-stick divergence} of $P$ from $Q$ of order $\alpha$ is 
\[
  \mathcal D_{\alpha}(P\|Q)
  :=
  \int_{\mathcal Y} \bigl[p(y) - \alpha\,q(y)\bigr]_+\,dy,
\]
where $[u]_+ := \max\{u,0\}$.
In the discrete case, the integral should be read as a sum over
$y\in\mathcal Y$.
\end{definition}

\begin{definition}[Differential privacy via hockey-stick divergence~\cite{dworkAlgorithmicFoundationsDifferential2013a, bartheDifferentialPrivacyComposition2013}]
Let $\mathcal X$ be an input domain and $\mathcal Z$ an output space.
A randomized mechanism is a map
$
  \mathcal M : \mathcal X^n \to \mathcal Z
$
that takes a dataset $x_{1:n}\in\mathcal X^n$ as input and outputs a
random element of $\mathcal Z$.
For each dataset $x_{1:n}$ we write $\mathcal M(x_{1:n})$ for the
corresponding output distribution on $\mathcal Z$, and we denote its
density by the same symbol when convenient.
Two datasets $x_{1:n},x'_{1:n}\in\mathcal X^n$ are said to be neighboring ($x_{1:n}\simeq x'_{1:n}$) if they differ in exactly one element. For a mechanism $\mathcal M$, $\varepsilon\ge0$, and $\delta\in [0,1]$, $\mathcal M$ is $(\varepsilon,\delta)$-DP if and only if
\[
  \delta_\mathcal{M}(\varepsilon) = \sup_{x_{1:n}\simeq x'_{1:n}}
  \mathcal D_{e^\varepsilon}\bigl(\mathcal M(x_{1:n}) \,\big\|\, \mathcal M(x'_{1:n})\bigr)\le\delta.
\]
The function $\delta_\mathcal{M}(\varepsilon)$ is referred to as the privacy profile of $\mathcal M$~\cite{ballePrivacyAmplificationSubsampling2018a}.
Sometimes, we say that $\mathcal M$ satisfies $(\varepsilon,\delta)$-DP for $x_{1:n}\simeq x'_{1:n}$ if $\mathcal D_{e^\varepsilon}\bigl(\mathcal M(x_{1:n}) \,\big\|\, \mathcal M(x'_{1:n})\bigr)\le\delta$~\cite{sommerPrivacyLossClasses2019}.
\end{definition}

In the shuffle model, the mechanism $\mathcal M$ is typically constructed by applying a local randomizer independently to each user's input.

\begin{definition}[Local randomizer and blanket distribution]
Let $\mathcal X$ be an input space and $\mathcal Y$ an output space.
A \emph{local randomizer} is a randomized mechanism
$
  \mathcal R : \mathcal X \to \mathcal Y.
$
For each input $x\in\mathcal X$ we write $\mathcal R_x$ for the
distribution of the random output $\mathcal R(x)$.
We assume that there is a reference measure on $\mathcal Y$ (counting
measure in the discrete case, Lebesgue measure in the continuous case)
such that each $\mathcal R_x$ admits a density, and we write
$
  \mathcal R_x(y)
$
for this probability mass function or density, depending on the context. The \emph{blanket distribution} $\mathcal R_{\mathrm{BG}}$ associated
with $\mathcal R$ is defined for $y\in\mathcal Y$ by,
\[
  \underline{\mathcal R}(y)
  := \inf_{x\in\mathcal X} \mathcal R_x(y),\qquad
  \gamma
  := \int_{\mathcal Y} \underline{\mathcal R}(y)\,dy \in (0,1],
  \qquad\mathcal R_{\mathrm{BG}}(y)
  := \frac{1}{\gamma}\,\underline{\mathcal R}(y).
\]
where, in the purely discrete case, the integral should be read as a sum
over $y\in\mathcal Y$.
The scalar $\gamma$ is called the \emph{blanket mass} of $\mathcal R$.
\end{definition}

Following Balle et al.~\cite{ballePrivacyBlanketShuffle2019}, we work in the
single-message \emph{randomize-then-shuffle} model.%
That is, we assume black-box access to an ideal functionality \emph{shuffler}\footnote{In contrast to the multi-message shuffling model, which typically
requires access to a stronger ideal functionality $\mathcal S$ capable of
splitting or decorrelating multiple messages from the same user (for example
by hiding timing and grouping information), the single-message shuffling model
only relies on a comparatively simple ``anonymization'' ideal functionality
that permutes user identities. This makes the single-message model
significantly easier to implement in practice, and therefore in this work we
restrict attention to the single-message shuffling setting.}.

\begin{definition}[Single-message shuffling and shuffled mechanism]

\label{def:single-message-shuffle}
Let $\mathcal R$ be a local randomizer. Given a dataset $x_{1:n}\in\mathcal X^n$, each user $i\in[n]$ applies the
local randomizer to obtain a single message
$
  Y_i := \mathcal R(x_i)\in\mathcal Y.
$
The \emph{shuffler} is a randomized map
$
  \mathcal S:\mathcal Y^n\to\mathcal Y^n
$
which samples a permutation $\pi$ uniformly at random from the symmetric
group on $[n]$ and outputs
$
  \mathcal S(y_1,\dots,y_n)
  := (y_{\pi(1)},\dots,y_{\pi(n)}).
$
A shuffled mechanism is $\mathcal{M}=\mathcal{S} \circ \mathcal{R}^n$, which shuffles the outputs of $\mathcal{R}$.
\end{definition}

Shuffled mechanisms are the main object of study in this paper and the following upper bound is known.

\begin{lemma}[Lemma 5.3 of Balle et al.~\cite{ballePrivacyBlanketShuffle2019}]
\label{lem:upper-blanket-divergence}
Fix $\varepsilon \geq 0$, let $x_{1:n} \simeq x_{1:n}^\prime$ be neighboring inputs where only the first element differs (i.e., $x_1 \neq x^\prime_1$), and $\mathcal{R}$ has blanket mass $\gamma$. 
Then, $\mathcal{S} \circ \mathcal{R}^n$ satisfies $(\varepsilon,\delta(x_{1:n}, x_{1:n}^\prime))$-DP for $x_{1:n}$ and $x_{1:n}^\prime$, where

\begin{align}
\nonumber\delta(x_{1:n}, x_{1:n}^\prime)
&\le
\mathcal{D}_{\mathrm{e}^\varepsilon,n,\mathcal{R}_\mathrm{BG},\gamma}^{\mathrm{blanket}}
  (\mathcal{R}_{x_1} \| \mathcal{R}_{x_1^\prime})
:=\nonumber
\frac{1}{n\gamma}
\mathbb{E}_{\substack{
  Y_{1:M}\stackrel{\mathrm{i.i.d.}}{\sim} \mathcal{R}_\mathrm{BG} \\
  M\sim \mathrm{Bin}(n,\gamma)
}}
\left[ \left( \sum_{i=1}^M l_\varepsilon(Y_i)\right)_+ \right],
\end{align}
where $(\cdot)_+$ denotes $\max\{0,\cdot\}$.

Here, $\mathrm{Bin}$ is a binomial distribution and $l_\varepsilon(y) := \frac{\mathcal{R}_{x_1}(y) - e^{\varepsilon} \mathcal{R}_{x_1^\prime}(y)}{\mathcal{R}_\mathrm{BG}(y)}$
is called the \textit{privacy amplification random variable} when $y=Y\sim \mathcal{R}_\mathrm{BG}$.
\end{lemma}

We term this upper bound \textit{blanket divergence} $\mathcal{D}_{\mathrm{e}^\varepsilon, n,\mathcal{R}_\text{BG},\gamma}^{\text{blanket}}$.
We omit the parameters and simply write $\mathcal{D}^{\text{blanket}}$ when they are clear from the context.
Moreover, for the same mechanism, the following lower bound, also expressed in terms of the blanket divergence, is known.

\begin{theorem}[Theorem 14 of Su et al.~\cite{suDecompositionBasedOptimalBounds2025}]
\label{theorem:lower-blanket-divergence}
For any $x,a_1,a_1^\prime\in\mathcal{X},\varepsilon,n$, and local randomizer $\mathcal{R}$, it holds that:
$$
\sup_{x_{1:n} \simeq x_{1:n}^\prime}\mathcal{D}_{\mathrm{e}^\varepsilon}(\mathcal{S} \circ \mathcal{R}^n(x_{1:n}) \| \mathcal{S} \circ \mathcal{R}^n(x_{1:n}^\prime)) \geq \mathcal{D}_{\mathrm{e}^\varepsilon,n,\mathcal{R}_x,1}^{\text{blanket}}(\mathcal{R}_{a_1} \| \mathcal{R}_{a_1^\prime}).
$$
\end{theorem}

To provide a unified treatment of the blanket divergences corresponding to the upper and lower bounds, we generalize the definition of the privacy amplification random variable as follows.
\begin{definition}[Privacy amplification random variable]
Fix the pair $x_1\neq x_1'\in\mathcal X$.
Let
$
  \mathcal R_{\mathrm{ref}}\in
  \bigl\{\mathcal R_{\mathrm{BG}}\bigr\}
  \,\cup\,
  \bigl\{\mathcal R_x : x\in\mathcal X\bigr\}
$.
We define the \emph{privacy amplification random variable} with respect to
$\mathcal R_{\mathrm{ref}}$ by
\[
  l_{\varepsilon}(Y;x_1,x_1^\prime,\mathcal{R}_{\mathrm{ref}})
  :=
  \frac{\mathcal R_{x_1}(Y) - e^{\varepsilon}\mathcal R_{x_1'}(Y)}
       {\mathcal R_{\mathrm{ref}}(Y)}.
\]
Throughout the paper, whenever we write
$l_{\varepsilon}(Y;x_1,x_1',\mathcal R_{\mathrm{ref}})$
without further comment, the random variable $Y$ is understood to be drawn
from the reference distribution $\mathcal R_{\mathrm{ref}}$. When the dependence on $(x_1,x_1',\mathcal R_{\mathrm{ref}})$ is clear from
context, we simply write $l_{\varepsilon}(Y)$ or $l_{\varepsilon}$. Two choices of the reference distribution $\mathcal R_{\mathrm{ref}}$ will be
particularly important in this work:
\begin{itemize}
  \item If $\mathcal R_{\mathrm{ref}}=\mathcal R_{\mathrm{BG}}$ is the blanket
  distribution, then $l_{\varepsilon}(Y)$ is
  the privacy amplification random variable underlying the
  blanket-divergence \emph{upper bound} of
  Balle et al.~\cite{ballePrivacyBlanketShuffle2019}.

  \item If $\mathcal R_{\mathrm{ref}}=\mathcal R_x$ for some $x\in\mathcal X$,
  then $l_{\varepsilon}(Y)$ is the privacy amplification random variable
  appearing in the \emph{lower bound} of
  Su et al.~\cite{suDecompositionBasedOptimalBounds2025}.
\end{itemize}
\end{definition}

\subsection{Local Randomizers Beyond Pure Local DP}

Most existing analyses of shuffling focus on
local randomizers that satisfy \emph{pure} local DP
(i.e., $(\varepsilon_0,0)$-DP at the user level).
One of the main goals of this paper is to remove this restriction and
develop a more general analysis that applies to arbitrary local
randomizers satisfying the following mild regularity assumptions, without requiring
pure local DP at the local level.

\begin{assumption}[Regularity conditions for the local randomizer]
\label{assump:regularity}
Consider the local randomizer $\mathcal R:\mathcal X\to\mathcal Y$.
We assume that there exists $\rho_0>0$ such that, for every pair
$x_1\neq x_1'\in\mathcal X$ and every reference distribution
$
  \mathcal R_{\mathrm{ref}}\in
  \bigl\{\mathcal R_{\mathrm{BG}}\bigr\}
  \,\cup\,
  \bigl\{\mathcal R_x : x\in\mathcal X\bigr\},
$
the following conditions hold uniformly over all $|\varepsilon|\le\rho_0$.

\begin{enumerate}
  \item[\textup{(1)}]
  \textbf{Uniform moment bounds.}
  For every integer $k\ge 1$,
  $
    \mathbb E_{Y\sim \mathcal{R}_{\mathrm{ref}}}
    \bigl[\,|l_{\varepsilon}(Y;x_1,x_1^\prime,\mathcal R_{\mathrm{ref}})|^k\,\bigr]
    <\infty.
  $

\item[\textup{(2)}]
\textbf{Non-degenerate variance.}
The variance
$
  \sigma^2 := \mathrm{Var}_{Y\sim \mathcal{R}_{\mathrm{ref}}}\bigl(l_{0}(Y;x_1,x_1^\prime,\mathcal R_{\mathrm{ref}})\bigr)
$
is strictly positive, and the following conditions hold:
$
  \int \frac{\mathcal R_{x_1}(y)^2}{\mathcal R_{\mathrm{ref}}(y)}\,dy < \infty
$ and
$
  \int \frac{\mathcal R_{x_1'}(y)^2}{\mathcal R_{\mathrm{ref}}(y)}\,dy < \infty.
$
  \item[\textup{(3)}]
  \textbf{Structural condition.}
  Either of the following holds.
  \begin{enumerate}
\item[\textup{(Cont)}]

$l_{\varepsilon}(Y;x_1,x_1',\mathcal{R}_{\mathrm{ref}})$ has a nontrivial absolutely continuous component with respect to Lebesgue
measure whose support contains a fixed nondegenerate bounded interval
$I\subset\mathbb{R}$ independent of
$\varepsilon$.
    \item[\textup{(Bound)}]
    There exists a constant $C<\infty$ such that
    $
      |l_{\varepsilon}(Y;x_1,x_1^\prime,\mathcal R_{\mathrm{ref}})|
      \le C
      \quad\text{almost surely}.
    $
  \end{enumerate}
\end{enumerate}

We denote by $\mathfrak{R}$ the class of local randomizers $\mathcal R$ that satisfy Assumption~\ref{assump:regularity}.
\end{assumption}

Assumption~\ref{assump:regularity} is mild: any nontrivial local randomizer that satisfies pure local DP automatically satisfies these conditions, and even when pure local DP does not hold, if the privacy amplification random variable $l_0(Y;x_1,x_1',\mathcal R_{\mathrm{ref}})$ is non-degenerate (i.e., not almost surely constant), then one can enforce the bounded case~\textup{(Bound)} in~(3) by truncating the output space to a large but bounded interval (i.e., by slightly modifying the local randomizer).

Accordingly, we formulate our results in a mechanism-agnostic way and do
not commit to a specific local randomizer.
Nevertheless, it is useful to keep in mind a concrete running example satisfying Assumption~\ref{assump:regularity} without pure local DP and truncation.
In particular, we will often refer to the \emph{generalized Gaussian
mechanism}, which forms a flexible family of location mechanisms (parameterized by the shape $\beta$) and simultaneously encompasses the Laplace mechanism
($\beta=1$) and the Gaussian mechanism ($\beta=2$).
The latter case is especially noteworthy, as the Gaussian mechanism has
been explicitly highlighted in prior work
(e.g., Balle et al.~\cite{ballePrivacyBlanketShuffle2019} and
Koskela et al.~\cite{koskelaNumericalAccountingShuffle2022}) as technically challenging to analyze in the shuffle
model.

\begin{definition}[Generalized Gaussian local randomizer ($\beta$-Gaussian mechanism)]
\label{def:generalized-gaussian-mechanism}
Let $\mathcal{X} := [0,1]$ be the input domain and fix
parameters $c>0$ and $\beta\in [1,2]$. The \emph{generalized Gaussian
local randomizer} $\mathcal{R}:\mathcal{X}\to\mathcal{Y}$ is defined by
$
  \mathcal{Y} := \mathbb{R},
  \mathcal{R}(x) = Y,\ \text{where } Y\sim\phi_x^{\beta}
$
and $\phi_x^{\beta}$ denotes the generalized Gaussian density
\[
  \phi_{x}^{\beta}(y)
  :=
  \frac{\beta}{2c\,\Gamma(1/\beta)}
  \exp\!\left(
    -\Bigl|\frac{y-x}{c}\Bigr|^{\beta}
  \right),
  \qquad y\in\mathbb{R}.
\]
In other words, for each input $x\in\mathcal{X}$, the output $\mathcal{R}(x)$ is obtained by adding generalized Gaussian noise (with parameters $c,\beta$) to the input $x$.
The variance of this distribution is given by $\mathrm{Var}(Y) = c^{2}\,\frac{\Gamma(3/\beta)}{\Gamma(1/\beta)}$; by choosing \(c\) appropriately, we can set
$
  \mathrm{Var}(Y)=\sigma_0^2.
$
We call $\mathcal{R}$ the $\beta$-Gaussian mechanism.
\end{definition}

\section{Analysis of Shuffling via the Central Limit Theorem (CLT)}
\label{sec:clt-analysis}

Here, we first review the blanket divergence and its interpretation.
Then, we analyze the blanket divergence to obtain an asymptotic expansion.
Next, we analyze the privacy profile of the shuffled mechanism based on this expansion.
Then, we interpret and explore the shuffle indices of various local randomizers.
Finally, we derive a necessary and sufficient condition under which the blanket divergence yields an asymptotically optimal characterization.

\subsection{Review of Blanket Divergence}
\label{sec:review_blanket_divergence}
Given a shuffled mechanism $\mathcal M=\mathcal S\circ \mathcal R^n$, the direct computation of the privacy profile requires the two output distributions of $\mathcal M(x_{1:n})$ and $\mathcal M(x'_{1:n})$ where $x_{1:n}$ and $x'_{1:n}$ are the worst-case neighboring datasets.
Given a dataset $x_{1:n}$, the output distribution of the shuffled mechanism $\mathcal M=\mathcal S\circ \mathcal R^n$ can be written as the permutation-mixture
\[
\Pr[\mathcal M(x_{1:n}) = y_{1:n}]
=
\frac{1}{n!}\sum_{\pi\in S_n}\;
\prod_{i=1}^n \mathcal R_{x_i}\bigl(y_{\pi(i)}\bigr),
\qquad y_{1:n}\in\mathcal Y^n.
\]
That is, evaluating the shuffled density naively involves a mixture with $n!$ components.
Moreover, we need to apply a nonlinear transformation to this and integrate the result, together with a maximization over all neighboring datasets to identify the worst case.
Altogether, this direct computation quickly becomes intractable.
We therefore turn to the privacy blanket approach proposed by \citet{ballePrivacyBlanketShuffle2019}.

Here, we briefly recall it.
The local randomizer admits the mixture decomposition
$
\mathcal R_x
=
\gamma\,\mathcal R_{\mathrm{BG}}
+
(1-\gamma)\,\mathcal R_x^{\mathrm{LO}},
$
where $\mathcal R_x^{\mathrm{LO}}$ is the corresponding leftover distribution.
Based on this decomposition, we consider the extended mechanism
$\widetilde{\mathcal M}$ depicted in Fig.~\ref{fig:blanket_divergence_interepretation}.
We can interpret a blanket divergence as an upper bound on the hockey-stick divergence of $\widetilde{\mathcal M}$.

In $\widetilde{\mathcal M}$, each user ($\mathrm{id}> 1$) draws a message from $\mathcal R_{\mathrm{BG}}$ with probability $\gamma$, and from $\mathcal R_{x_\mathrm{id}}^{\mathrm{LO}}$ with probability $1-\gamma$.
Messages sampled from the leftover part are revealed together with
their user indices, whereas messages sampled from the blanket part are released only after being shuffled.
User~$1$ outputs $\mathcal R_{x_1}$ (resp.\ $\mathcal R_{x_1'}$) and $Y_1$ is always placed in the shuffled blanket multiset.
The hockey-stick divergence of $\widetilde{\mathcal M}$ is upper-bounded by applying its joint convexity to the mixture over the number of blanket messages, which yields the blanket divergence (see Section~\ref{sec:interepretation-blanket-divergence} for details).
Importantly, the shuffled mechanism $\mathcal M=\mathcal S\circ\mathcal R^n$ can be obtained by shuffling all outputs of $\widetilde{\mathcal M}$.
Therefore, by the data processing inequality, the hockey-stick divergence of $\mathcal M$ is also upper-bounded by the blanket divergence.

Next, in Fig.~\ref{fig:blanket_divergence_interepretation}, we consider a fixed $x_2\in\mathcal X$ and restrict attention to the neighboring-dataset instance
$
(x_1,x_2,\dots,x_2)\ \simeq\ (x_1',x_2,\dots,x_2).
$
In this setting, if we take $\gamma=1$ and choose the reference distribution as $\mathcal R_{\mathrm{BG}}=\mathcal R_{x_2}$, then the resulting construction of $\widetilde{\mathcal M}$ coincides in distribution with the  shuffled mechanism output for these two datasets (i.e., $\mathcal{M}((x_1,x_2,\dots,x_2))$ and $\mathcal{M}((x_1',x_2,\dots,x_2))$). 
The inequalities introduced earlier, namely the convexity step in $\gamma$ and the data processing inequality, do not appear here, so that the blanket divergence matches the hockey-stick divergence of $\widetilde{\mathcal{M}}$, and thus that of the shuffled mechanism $\mathcal{M}$.
Since this neighboring instance is feasible in the supremum that defines the privacy profile, it follows that this blanket divergence provides a lower bound.

\begin{remark}[Relationship to the clone paradigm]
\label{remark:clone_paradigm}
The blanket divergence viewpoint with $\gamma=1$ is closely related to the clone paradigm~\cite{feldmanHidingClonesSimple2022,feldmanStrongerPrivacyAmplification2023a}.
The clone paradigm reduces shuffle privacy analysis to an auxiliary mechanism; shuffling of $n$ categorical distributions $(\mathcal{R}^{\mathrm{Cat}}_{b}, \mathcal{R}^{\mathrm{Cat}}_{\mathrm{ref}}, \dots, \mathcal{R}^{\mathrm{Cat}}_{\mathrm{ref}})$, where $b\in\{0,1\}$ changes the target user.
In this sense, this reduced mechanism can be interpreted as a particular $\gamma=1$ instance of $\widetilde{\mathcal{M}}$, with an appropriate choice of $\mathcal R^{\mathrm{Cat}}_{b}$ and $\mathcal R^{\mathrm{Cat}}_{\mathrm{ref}}$.
Note that \citet{suDecompositionBasedOptimalBounds2025} showed that the blanket-based reduction can be at least as tight as the clone reductions.
Motivated by this and by the mechanism-aware nature of the blanket divergence, we focus primarily on the blanket-divergence upper bound of
\citet{ballePrivacyBlanketShuffle2019} in what follows.
\end{remark}

\begin{figure}[t]
  \centering
  \begin{minipage}[t]{0.4\linewidth}
    \vspace{0pt} 
    \centering
    \includegraphics[width=\linewidth]{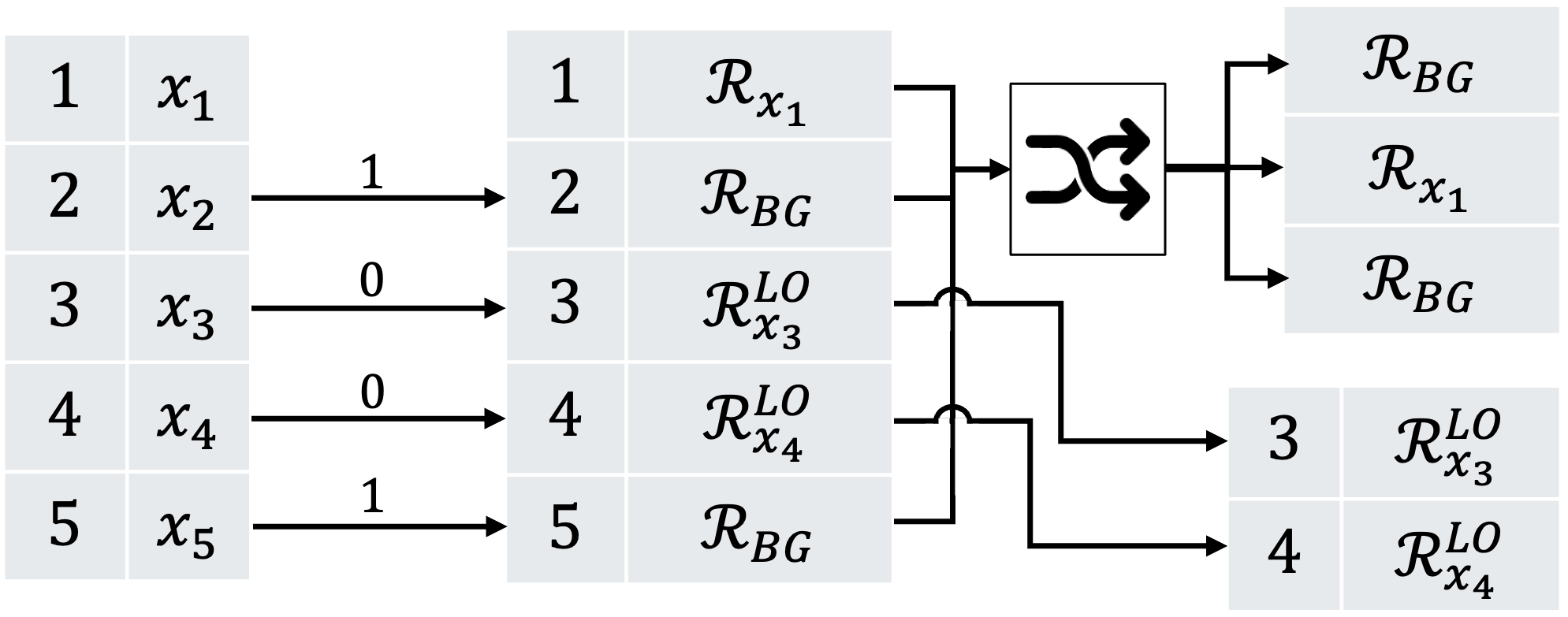}
    \captionof{figure}{$\widetilde{\mathcal{M}}$, which yields the blanket divergence.}
    \label{fig:blanket_divergence_interepretation}


    \includegraphics[width=\linewidth]{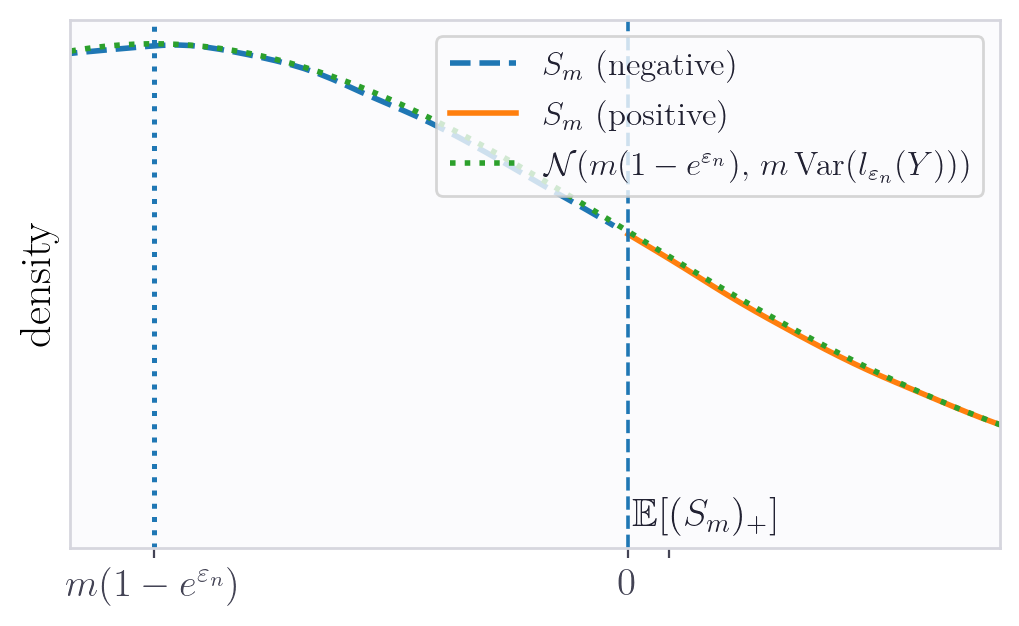}
    \captionof{figure}{The intuition behind the CLT of the blanket divergence with $m=100$.}
    \label{fig:clt_image}
  \end{minipage}
  \hfill
  \begin{minipage}[t]{0.5\linewidth}
    \vspace{0pt} 
    \centering
    \includegraphics[width=\linewidth]{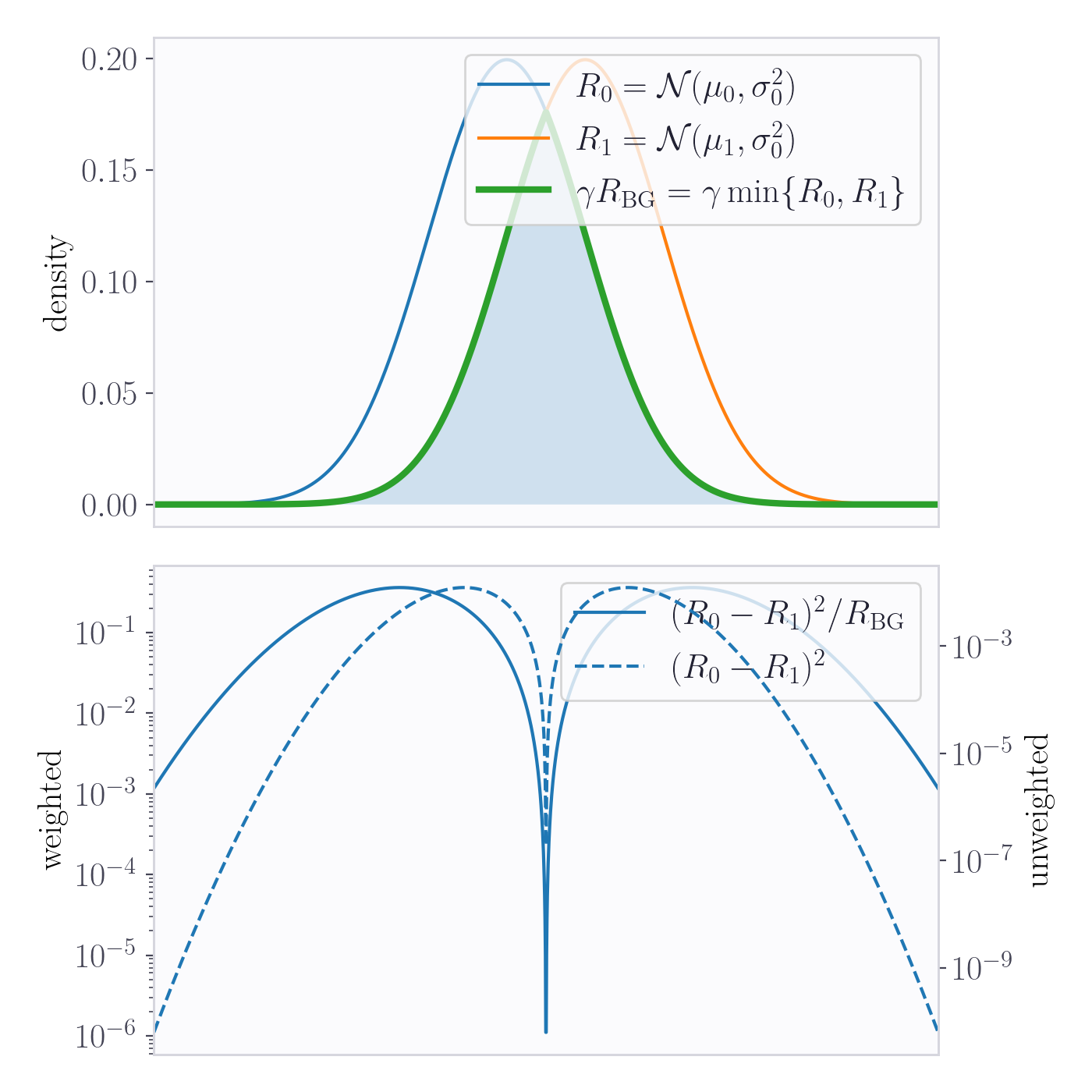}
    \captionof{figure}{The intuition of $\sigma$ with the Gaussian mechanism where $\sigma_0=2$. $x$ axis represents the output of the local randomizer and is common for both images.}
    \label{fig:sigma_image}
  \end{minipage}

\end{figure}

\subsection{Analysis of Blanket Divergence}
\label{sec:analysis_blanket_divergence}
As shown in Lemma~\ref{lem:upper-blanket-divergence}, the blanket divergence no longer involves an explicit permutation. 
This is because, conditioned on the number of blanket messages $M=m$ in $\widetilde{\mathcal{M}}$, the $m-1$ blanket outputs are i.i.d.\ draws from the blanket distribution $\mathcal R_{\mathrm{BG}}$; moreover, thanks to the symmetry attained by shuffling, the analysis of the hockey-stick divergence reduces to that of a sum of $m$ i.i.d.\ random variables.
Exploiting this tractability via the CLT, we can derive the following asymptotic expansion.

\begin{restatable}[Asymptotic expansion of the blanket divergence]{lemma}{blanketasymptoticsgeneral}
\label{lem:blanket-asymptotics-general}
Given $\mathcal{R}\in \mathfrak{R}$, fix a pair $x_1\neq x_1'\in\mathcal X$ and let
$
  \mathcal R_{\mathrm{ref}}
  \in
  \bigl\{\mathcal R_{\mathrm{BG}}\bigr\}
  \,\cup\,
  \bigl\{\mathcal R_x : x\in\mathcal X\bigr\}.
$
For $\gamma\in(0,1]$, define
$
  \sigma^2
  := \mathrm{Var}_{Y\sim\mathcal{R}_{\mathrm{ref}}}\bigl(l_{\varepsilon=0}(Y;x_1, x_1^\prime,\mathcal{R}_{\mathrm{ref}})\bigr)
$
and
$
  \chi := \sqrt{\gamma}/{\sigma}.
$
Assume that as $n\to\infty$,
$
  \varepsilon_n = \omega(\sqrt{1/n})
$
and
$
  \varepsilon_n = O\!\left(\sqrt{\log n / n}\right).
$
Then, as $n\to\infty$,
\begin{align}
  \label{eq:blanket-asymptotics-general}
  \mathcal D^{\mathrm{blanket}}_{e^{\varepsilon_n},n,\mathcal R_{\mathrm{ref}},\gamma}(\mathcal{R}_{x_1} \| \mathcal{R}_{x_1^\prime})
  &\;=\;
  \varphi\bigl(\chi\varepsilon_n\sqrt{n}\bigr)\,
  \left(
    \frac{1}{\chi^3}\,
    \frac{1}{\varepsilon_n^2 n^{3/2}}
  \right)\bigl(1+o(1)\bigr).
\end{align}
\end{restatable}
Here, we present an informal derivation based on the CLT to build intuition of how the result arises. 
The rigorous proof with the asymptotic expansion of CLT~\cite{slastnikovLimitTheoremsModerate1979,angstWeakCramerCondition2017,petrovSumsIndependentRandom1975} is in Appendix~\ref{app:proof-blanket-asymptotics-general}.

Recall that the blanket divergence can be written as
\[
\mathcal D^{\mathrm{blanket}}
=
\frac{1}{n\gamma}\;
\mathbb E_{M}\!\left[
\mathbb E\!\left[\Bigl(\sum_{i=1}^{M} l_{\varepsilon_n}(Y_i)\Bigr)_+ \,\Big|\, M\right]
\right],
\qquad
M\sim\mathrm{Bin}(n,\gamma),\;\; Y_i\stackrel{\mathrm{i.i.d.}}{\sim}\mathcal R_{\mathrm{ref}}.
\]
Let us condition on $M=m$ and define
$
S_m := \sum_{i=1}^{m} l_{\varepsilon_n}(Y_i).
$
The inner expectation becomes $\mathbb E[(S_m)_+]$, whose behavior is shown in Fig.~\ref{fig:clt_image}.
Since $S_m$ is a sum of i.i.d.\ random variables, the CLT suggests
\[
S_m \approx N_m := \mathcal N(\mu_m,s_m^2),
\qquad
\mu_m := m\,\mathbb E[l_{\varepsilon_n}(Y)]=m(1-e^{\varepsilon_n}),
\quad
s_m^2 := m\,\mathrm{Var}(l_{\varepsilon_n}(Y)),
\]
for large $m$. For a Gaussian $N_m\sim\mathcal N(\mu_m,s^2_m)$ one has the identity
$
\mathbb E[(N_m)_+]
=
s_m\,\varphi\!\left(\frac{\mu_m}{s_m}\right)
+
\mu_m\,\Phi\!\left(\frac{\mu_m}{s_m}\right).
$
Using Mills' ratio (in the regime $|\mu_m/s_m|\to\infty$), we obtain
\[
\mathbb E[(N_m)_+]
\;\approx\;
\frac{s_m^3}{\mu_m^2}\,
\varphi\!\left(\frac{\mu_m}{s_m}\right)
\;=\;
\frac{\mathrm{Var}(l_{\varepsilon_n}(Y))^{3/2}}{(e^{\varepsilon_n}-1)^2}\,
\frac{1}{\sqrt m}\,
\varphi\!\left(
\frac{(e^{\varepsilon_n}-1)\sqrt m}{\sqrt{\mathrm{Var}(l_{\varepsilon_n}(Y))}}
\right),
\]
$M\sim\mathrm{Bin}(n,\gamma)$ concentrates around its mean $m=n\gamma$, so we plug in $m=n\gamma$ in the leading term. Moreover, for small $\varepsilon_n$, $\mathrm{Var}(l_{\varepsilon_n}(Y))\approx \mathrm{Var}(l_{0}(Y))=: \sigma^2$.
With $\chi:=\sqrt{\gamma}/\sigma$, we obtain
\[
\mathbb E_M[(N_{M})_+]
\approx
\frac{\sigma^3}{(e^{\varepsilon_n}-1)^2}\;
\frac{1}{\sqrt{n\gamma}}\;
\varphi\!\bigl(\chi(e^{\varepsilon_n}-1)\sqrt n\bigr).
\]
Finally, recalling the prefactor $1/(n\gamma)$, and using $\sigma^3/\gamma^{3/2}=1/\chi^3$, we arrive at
\[
\mathcal D^{\mathrm{blanket}}
\approx
\frac{1}{n\gamma}\,\mathbb E_M[(N_{M})_+]
\approx
\varphi\!\bigl(\chi(e^{\varepsilon_n}-1)\sqrt n\bigr)\,
\left(
\frac{1}{\chi^3}\,
\frac{1}{(e^{\varepsilon_n}-1)^2\,n^{3/2}}
\right),
\]
which matches the leading term in~\eqref{eq:blanket-asymptotics-general} using $\mathrm{e}^{\varepsilon_n}\approx 1+\varepsilon_n$.
The rigorous proof in Appendix~\ref{app:proof-blanket-asymptotics-general} shows that this normal
approximation is indeed valid under Assumption~\ref{assump:regularity} (i.e., $\mathcal{R}\in\mathfrak{R}$) in the regime $\varepsilon_n=\omega(\sqrt{1/n})$ and $\varepsilon_n=O(\sqrt{\log n/n})$.

Our key conceptual finding from Lemma~\ref{lem:blanket-asymptotics-general} is that the leading term depends on the local randomizer only through the single parameter
$
  \chi = \sqrt{\gamma}/{\sigma}.
$
That is, at the level of the leading asymptotics, the dependence of the blanket divergence on the local randomizer $\mathcal R$ is fully encoded by the one-dimensional index $\chi$. 
Furthermore, the leading term is monotone in $\chi$: a larger $\chi$ yields a smaller leading term in the blanket divergence, and hence corresponds to stronger privacy amplification by shuffling. 
This motivates interpreting $\chi$ as a measure of the shuffle efficiency of the local randomizer. 
Owing to its central role, we refer to $\chi$ as the \emph{shuffle index}.

\subsection{Main Result}
\label{sec:main_result}
Up to this point, we have analyzed the blanket divergence for a fixed pair $(x_1,x_1',\mathcal{R}_\mathrm{ref})$.  
In contrast, ($\varepsilon_n,\delta_n$)-DP requires a guarantee that holds uniformly over all neighboring datasets $x_{1:n}\simeq x'_{1:n}$ (i.e., upper bounding the privacy profile $\delta_n\geq \delta_\mathcal{M}(\varepsilon_n)$).
This section aims to relate the result of the blanket divergence to the privacy profile of the shuffled mechanism.

Recall from Section~\ref{sec:review_blanket_divergence} that the privacy profile can be sandwiched between two blanket divergences. 
By combining this property with the monotonicity of the leading term with respect to the shuffle index $\chi$, evaluating the privacy profile reduces to identifying the worst-case value of $\chi$ over all pairs $(x_1,x_1')$. 
To derive the upper bound, under the reference distribution $\mathcal R_{\mathrm{ref}} = \mathcal R_{\mathrm{BG}}$, we must choose the pair $(x_1,x_1')$ that minimizes the shuffle index. 
Conversely, to establish the lower bound, under $\mathcal R_{\mathrm{ref}} = \mathcal R_x$ for some $x\in\mathcal X$ (with $\gamma=1$), we are free to select any pair $(x_1,x_1')$; naturally, to maximize this lower bound, we should again choose the pair that minimizes the shuffle index.
Formally, we use the following specific shuffle indices for the upper and lower bounds, respectively.
\begin{definition}
\label{def:shuffle_index}
Let $\gamma\in(0,1]$ be the blanket mass of $\mathcal{R}\in\mathfrak{R}$.
The \emph{lower shuffle index} is defined by
$
  \chi_{\mathrm{lo}}(\mathcal R)
  :=
  \sqrt{\gamma}/\sup_{x_1\simeq x_1'\in\mathcal X}
  \sqrt{
    \mathrm{Var}_{Y\sim \mathcal{R}_{\mathrm{BG}}}
    \!\left[
      l_0\!\left(Y; x_1,x_1',\mathcal{R}_{\mathrm{BG}}\right)
    \right]
  }.
$
The \emph{upper shuffle index} is defined by
$
  \chi_{\mathrm{up}}(\mathcal R)
  :=
  1/\sup_{x_1\simeq x_1'\in\mathcal X}
  \ \sup_{x\in\mathcal{X}}
  \sqrt{
    \mathrm{Var}_{Y\sim \mathcal{R}_{x}}
    \!\left[
      l_0\!\left(Y; x_1,x_1',\mathcal{R}_{x}\right)
    \right]
  }.
$
\end{definition}

With these specific shuffle indices in hand, we can formally establish the upper and lower bounds for the privacy profile of the shuffled mechanism, as stated in the following corollary.

\begin{corollary}[Privacy profile bounds]
\label{lem:shuffle-index-profile-seq-no-t}
Let $\mathcal R\in \mathfrak R$.
Assume that as $n\to\infty$,
$
\varepsilon_n=\omega(\sqrt{1/n})
$
and
$
\varepsilon_n=O\!\left(\sqrt{\log n/n}\right).
$
Then, there exist sequences
$e_n^{\mathrm{up}},e_n^{\mathrm{lo}}\to 0$ as $n\to\infty$ such that
\[
  f_{n,\varepsilon_n}(\chi_{\mathrm{up}}(\mathcal R))\,\bigl(1+e_n^{\mathrm{up}}\bigr)
  \le
  \delta_{\mathcal S\circ \mathcal R^n}(\varepsilon_n)
  \le
  f_{n,\varepsilon_n}(\chi_{\mathrm{lo}}(\mathcal R))\,\bigl(1+e_n^{\mathrm{lo}}\bigr),
\]
for all sufficiently large $n$ where
$
  f_{n,\varepsilon_n}(\chi)
  :=
  \varphi\bigl(\chi\,\varepsilon_n\sqrt{n}\bigr)\,
  \left(
    \frac{1}{\chi^3}\,
    \frac{1}{\varepsilon_n^2 n^{3/2}}
  \right).
$
\end{corollary}

This corollary is a direct consequence of applying the asymptotic expansion established in Lemma~\ref{lem:blanket-asymptotics-general} to the lower and upper bounds (Lemma~\ref{lem:upper-blanket-divergence} and Theorem~\ref{theorem:lower-blanket-divergence}).

In DP, it is common to set $\delta_n$ to be $\alpha/n$ for some fixed constant $\alpha>0$.
Our next result characterizes the corresponding $\varepsilon_n$.

\begin{restatable}[Shuffling delta band]{theorem}{shufflingdeltaband}
\label{thm:moderate-deviation-shuffle}
Let $\mathcal{R}\in\mathfrak{R}$ with blanket mass $\gamma$.
Fix $\alpha>0$. For $n\ge2$ and $\chi>0$, set\footnote{%
The explicit formula for $\varepsilon_n(\alpha,\chi)$ in
Theorem~\ref{thm:moderate-deviation-shuffle} is obtained by solving the
leading-order approximation $\eqref{eq:blanket-asymptotics-general}=\alpha/n$ in closed form via the Lambert $W$-function.
For our numerical experiments and practical use, however, we instead
compute $\varepsilon_n$ by numerically solving the more accurate equation based on the full asymptotic expression Eq.~\eqref{eq:A-n-expansion}, using a simple bisection method.
}
\begin{equation}
  \label{eq:asymptoptic_epsilon}
  \varepsilon_n(\alpha,\chi)
  := \log\left(
      1
      +\sqrt{\frac{2}{\chi^2 n}\,
              W\!\left(
                \frac{\sqrt{n}}{2\alpha\chi\sqrt{2\pi}}
              \right)}
    \right)
\end{equation}
Then, there exists a sequence
$(\varepsilon^\ast_n\in [\varepsilon_n\bigl(\alpha,\chi_{\mathrm{up}}(\mathcal{R})\bigr),
  \varepsilon_n\bigl(\alpha,\chi_{\mathrm{lo}}(\mathcal{R})\bigr)])_{n\ge1}$ such that
$
  \delta_{\mathcal{S}\circ\mathcal{R}^n}(\varepsilon^\ast_n)=\frac{\alpha}{n}(1+o(1))
$ for all sufficiently large $n$.
\end{restatable}

See Appendix~\ref{app:proof-moderate-deviations-shuffle} for the full proof. 
This is obtained by inverting the relationship between $\varepsilon_n$ and $\delta_n$ from the bounds established in Corollary~\ref{lem:shuffle-index-profile-seq-no-t}. 
Specifically, by equating the leading term of the privacy profile $f_{n,\varepsilon_n}(\chi)$ to the target $\delta_n = \alpha/n$, we can solve for $\varepsilon_n$ to derive the explicit formula $\varepsilon_n(\alpha, \chi)$ for both the upper and lower bounds.
Informally, Theorem~\ref{thm:moderate-deviation-shuffle} shows that for any fixed target constant $\alpha>0$
we construct a sequence $(\varepsilon_n^\ast)_{n\ge2}$ whose values lie
between two explicit curves $\varepsilon_n(\alpha,\chi_{\mathrm{up}}(\mathcal{R}))$
and $\varepsilon_n(\alpha,\chi_{\mathrm{lo}}(\mathcal{R}))$, and for which $\delta_n
 = \frac{\alpha}{n}(1+o(1))$.\footnote{%
  More generally, one can obtain any prescribed
  polynomial rate $\delta_n = \Theta(n^{-\beta})$ with $\beta>1$ in the leading
  term.  In this work we focus on the regime $\delta_n = \alpha/n(1+o(1))$ because in the DP literature, setting $\delta_n$ to be $\alpha/n$ is a common practice.
 }
In other words, the band precisely captures $\varepsilon^\ast_n$ that realizes $\delta_n\approx\alpha/n$.

Another important implication of Theorem~\ref{thm:moderate-deviation-shuffle} is that the asymptotic expansion of $\eqref{eq:asymptoptic_epsilon}= \frac{1}{\chi}\sqrt{\frac{\log n}{n}}\left(1+o(1)\right)$ as $n\to\infty$ reveals a simple relationship.
The amplified privacy parameter $\varepsilon_n$ is inversely proportional to the shuffle index $\chi$. 
Consequently, the width of the band is asymptotically captured by the ratio
\[
  \frac{\varepsilon_n(\alpha,\chi_{\mathrm{lo}}(\mathcal{R}))}
       {\varepsilon_n(\alpha,\chi_{\mathrm{up}}(\mathcal{R}))}
  \;=\;
  \frac{\chi_{\mathrm{lo}}(\mathcal{R})}{\chi_{\mathrm{up}}(\mathcal{R})}\,
  \left(1+o(1)\right).
\]
Hence, if $\chi_{\mathrm{lo}}(\mathcal{R})/\chi_{\mathrm{up}}(\mathcal{R})$ is close to 1, the band is narrow, indicating that the shuffle index provides a tight characterization of the amplified privacy parameter.

\subsection{Shuffle Index}
\label{sec:shuffle_index}

In this subsection, we explore the shuffle index $\chi$: its interpretation, examples, and the analysis of the tightness of the blanket-divergence bounds through the shuffle index.

\paragraph{Interpretation of Shuffle Index.}

Technically, by the CLT as shown in Section~\ref{sec:analysis_blanket_divergence}, the leading asymptotics of the blanket divergence with $(x_1,x_1^\prime,\mathcal R_{\mathrm{ref}})$ can be summarized by the two parameters $\gamma$ and $\sigma$. 
Here we defined
$
\sigma^2
:=
\mathrm{Var}_{Y\sim\mathcal R_{\mathrm{ref}}}\left(\frac{\mathcal R_{x_1}(Y)-\mathcal R_{x_1'}(Y)}{\mathcal R_{\mathrm{ref}}(Y)}\right)
=
\int_{\mathcal Y}
\frac{\bigl(\mathcal R_{x_1}(y)-\mathcal R_{x_1'}(y)\bigr)^2}{\mathcal R_{\mathrm{ref}}(y)}\,dy
$
so that $\sigma$ is a $\chi^2$-type distance between $\mathcal R_{x_1}$ and $\mathcal R_{x_1'}$ with respect to the reference distribution $\mathcal R_{\mathrm{ref}}$.
This differs from the usual $\chi^2$ divergence, whose reference distribution is $\mathcal R_{x_1}$ or $\mathcal R_{x_1'}$; here the reference distribution is the third distribution $\mathcal R_{\mathrm{ref}}$.
In Fig.~\ref{fig:sigma_image}, we take $\mathcal R$ to be Gaussian and plot the integrand of $\sigma^2$, namely $(\mathcal R_{x_1}-\mathcal R_{x_1'})^2/\mathcal R_{\mathrm{BG}}$.
Compared to the ordinary squared distance $(\mathcal R_{x_1}-\mathcal R_{x_1'})^2$, we see that regions where the blanket is thin are emphasized.
This distinction represents the shuffled setting: we can interpret $\sigma$ as distinguishability of $\mathcal{R}_{x_1}$ and $\mathcal{R}_{x_1'}$ from the blanket $\mathcal R_{\mathrm{BG}}$ view.

To reiterate, the blanket divergence can be interpreted as the natural upper bound of the hockey-stick divergence between $\mathcal S(\mathcal R_{x_1},\mathcal R_{\mathrm{ref}},\ldots,\mathcal R_{\mathrm{ref}})$ and $\mathcal S(\mathcal R_{x_1'},\mathcal R_{\mathrm{ref}},\ldots,\mathcal R_{\mathrm{ref}})$.
$\gamma$ represents the fraction of messages that fall into the blanket, and hence the effective message size is $M\approx n\gamma$.
On the other hand, since $1/\sigma$ is the reciprocal of the distinguishability measure in the above setting, it quantifies the per-message indistinguishability.
Because the contribution of $\gamma$ to indistinguishability appears through the $\sqrt{n}$ scaling of fluctuations, the index aggregates these effects as $\sqrt{\gamma}/\sigma$.

\paragraph{Examples of Shuffle Indices}

\begin{table}[t]
\caption{
  Shuffle index asymptotics of $\beta$-Gauss with $\mathrm{Var}(\mathcal{R}(x)) = \sigma_0^2$ where $c_\beta = \frac{\Gamma(1/\beta)}{\beta\sqrt{\Gamma(3/\beta)\Gamma(2-1/\beta)}}$ and $k$-RR.
}
\centering
\small
\setlength{\tabcolsep}{6pt}
\renewcommand{\arraystretch}{1.25}

\makebox[\linewidth][c]{%
  \begin{minipage}[t]{0.5\linewidth}
    \centering
    \begin{tabular}{@{}lcc@{}}
      \toprule
      $\beta$-Gauss
      & $\sigma_0\to\infty$
      & $\sigma_0\to 0$ \\
      \midrule
      $\chi_{\mathrm{lo}}$
      &
      $c_\beta \sigma_0\left(1  +o\left(1\right)\right)$
      &
      $\exp\!\bigl(-\Omega(\sigma_0^{-\beta})\bigr)$
      \\
      $\chi_{\mathrm{lo}}/\chi_{\mathrm{up}}$
      &
      $1-O\!\left(\frac{1}{\sigma_0}\right)$
      &
      $\frac{1}{\sqrt{2}} + o(1)$
      \\
      \bottomrule
    \end{tabular}
  \end{minipage}
  \begin{minipage}[t]{0.5\linewidth}
    \centering
    \begin{tabular}{@{}lcc@{}}
      \toprule
      $k$RR $\mid$ $2$RR
      & $\varepsilon_0\to 0$
      & $\varepsilon_0\to \infty$ \\
      \midrule
      $\chi_{\mathrm{lo}}$
      &
      $\sqrt{\frac{k}{2}} \frac{1}{\varepsilon_0}\left(1+O(\varepsilon_0)\right)$
      &
      $\frac{e^{-\frac{\varepsilon_0}{2}}}{\sqrt2}\left(1+O(e^{-\varepsilon_0})\right)$
      \\
      $\chi_{\mathrm{lo}}/\chi_{\mathrm{up}}$
      &
      $1$ $\mid$ $1-O(\varepsilon_0)$
      &
      $1$ $\mid$ $\frac{1}{\sqrt{2}} + O(e^{-\varepsilon_0})$
      \\
      \bottomrule
    \end{tabular}
  \end{minipage}%
}
\label{tab:shuffle_index_asymptotics}
\end{table}

We illustrate the shuffle index of two local randomizers: the $\beta$-Gaussian mechanism and $k$-RR. 
Concretely, we examine (i) the asymptotics of the lower shuffle index $\chi_{\mathrm{lo}}$, and (ii) the ratio $\chi_{\mathrm{lo}}/\chi_{\mathrm{up}}$ as a proxy for the tightness of the blanket-divergence sandwich. 
Closed-form derivations and some plots of shuffle indices are deferred to Appendix~\ref{app:computation-shuffle-indices}. 
The resulting asymptotics are summarized in Table~\ref{tab:shuffle_index_asymptotics}: for $\beta$-Gaussian we parametrize by the local noise scale $\sigma_0$, while for $k$-RR we parametrize by the pure local-DP level $\varepsilon_0$, and in both cases we report limits as the parameter tends to $0$ and $\infty$.
These asymptotics precisely capture the fundamental behavior established in theoretical bounds~\cite{erlingssonAmplificationShufflingLocal2019} and explain recent empirical observations~\cite{chenPrivacyAmplificationCompression2023} that shuffling yields highly efficient privacy amplification in the high-privacy regime, while providing negligible amplification in the low-privacy regime.

\textbf{$\beta$-Gaussian.}
For large noise ($\sigma_0\to\infty$), the lower shuffle index grows linearly, $\chi_{\mathrm{lo}}\approx c_\beta\sigma_0$, where $c_\beta$ depends only on the shape parameter $\beta$.
Concretely, $c_\beta$ increases from $1/\sqrt{2}$ at $\beta=1$ (Laplace) up to $1$ at $\beta=2$ (Gaussian), and then decreases for $\beta>2$; hence, in the high-noise regime, the Gaussian case $\beta=2$ yields the largest shuffle index.
Moreover, the tightness ratio converges to $1$, so the blanket-divergence bounds essentially coincide asymptotically when $\sigma_0$ is sufficiently large.
That is, in the mean estimation task, since $\sigma_0$ monotonically determines the accuracy of the estimate, $\beta=2$ induces the best privacy-utility trade-off in the high-noise regime (see Section~\ref{sec:distribution-estimation} for the exploration of this topic).
In contrast, in the low-noise regime ($\sigma_0\to 0$), $\chi_{\mathrm{lo}}$ decays exponentially to $0$, which reflects that shuffling provides little amplification when the local noise is too small.

\begin{remark}[Dimensional dependence]
\label{remark:dimensional_dependence}
For the Gaussian case ($\beta=2$), it is natural to extend the input domain to $d$ dimensions.
Concretely, consider the Gaussian local randomizer $\mathcal R(x)=x+\mathcal N(0,\sigma_0^2 I_d)$ with fixed $\sigma_0$ and inputs in an $\ell_2$-ball of fixed radius.
Then the blanket mass satisfies $\gamma=\exp(-\Theta(\sqrt d))$, i.e., it decays exponentially fast in $\sqrt d$.
Intuitively, this decay is driven by the thin-shell phenomenon of high-dimensional Gaussians: a typical output exponentially concentrates on a thin shell around its mean $x$ (i.e., $\|\mathcal R(x)-x\|\approx \sigma_0\sqrt d$).
On this typical shell, the blanket distribution takes an infimum over all inputs in the $\ell_2$-ball, so the blanket becomes exponentially thin in $\sqrt d$ (see Section~\ref{app:dim-dependence-gamma} for details).
Moreover, while $\chi_{\mathrm{up}}$ is not directly affected by $\gamma$, $\chi_{\mathrm{lo}}$ is suppressed by the factor $\sqrt{\gamma}$ and thus shrinks exponentially with $\sqrt{d}$; consequently, the gap between the corresponding upper and lower bounds becomes large.
Closing this gap appears to require techniques that go beyond the blanket divergence, which remains an open problem.
\end{remark}

\textbf{$k$-RR.}
In large noise ($\varepsilon_0\to 0$), we have $\chi_{\mathrm{lo}}\approx\sqrt{\frac{k}{2}}\frac{1}{\varepsilon_0}$.
Importantly, for $k\ge 3$ the tightness ratio satisfies $\chi_{\mathrm{lo}}/\chi_{\mathrm{up}}=1$, i.e., the two indices coincide and the blanket-divergence sandwich is tight at leading order; for $2$-RR, the ratio is also close to $1$ as $\varepsilon_0\to 0$.
In the low-noise regime ($\varepsilon_0\to\infty$), $\chi_{\mathrm{lo}}$ again vanishes exponentially, and thus shuffling becomes ineffective.

\begin{remark}[Comparison to the clone paradigm]
\label{remark:comparison_to_clone}
As noted in Remark~\ref{remark:clone_paradigm}, the clone paradigm can be understood in terms of a blanket divergence with $\gamma=1$ and an appropriate choice of the categorical distributions. 
Hence, we can induce a shuffle index of the clone paradigm.
For $k$-RR, the explicit reduction of \citet[Theorem~7.4]{feldmanStrongerPrivacyAmplification2023a} leads to a shuffle index that coincides with our $\chi_{\mathrm{lo}}$ (see Section~\ref{app:krr_comparison_to_clone} for details).
While their work only demonstrates empirical closeness to known lower bounds, our framework provides a theoretical and asymptotic explanation for the closeness.
\end{remark}

\paragraph{Tightness of Blanket-Divergence Analysis.}
We have seen that $\chi_\mathrm{lo}/\chi_\mathrm{up}=1$ for $k$-RR with $k\ge 3$.
Here, we explore the structural conditions under which this equality holds, which implies that the blanket-divergence bounds are asymptotically tight in the sense of Theorem~\ref{thm:moderate-deviation-shuffle}. 
Formally, we have the following characterization.

\begin{restatable}{theorem}{globalbandcollapse}
\label{thm:global-band-collapse}
Let $\mathcal{R}\in\mathfrak{R}$ with blanket mass $\gamma$ and
assume that the suprema in the definitions of $\chi_{\mathrm{lo}}(\mathcal{R})$ and $\chi_{\mathrm{up}}(\mathcal{R})$ are attained by some pairs.
Then, $\chi_{\mathrm{up}}(\mathcal{R}) = \chi_{\mathrm{lo}}(\mathcal{R})$ if and only if there exists a pair which attains the suprema for the lower shuffle index $(x_1^\ast,x_1^{\prime\,\ast})$ and an input $x^\ast\in\mathcal X$ such that
  \[
    \mathcal{R}_{x^\ast}(y)
    = \gamma\mathcal{R}_{\mathrm{BG}}(y)
    \ \text{ for almost every }y\in\bigl\{ y\in\mathcal Y : \mathcal{R}_{x_1^\ast}(y)\neq \mathcal{R}_{x_1^{\prime\,\ast}}(y) \bigr\}.
  \]
In particular, when these equivalent conditions hold, the lower and upper curves in Theorem~\ref{thm:moderate-deviation-shuffle} coincide asymptotically, and
there exists an asymptotic privacy curve $\varepsilon_n^\ast(\alpha)$ realizing $\delta_{\mathcal{S}\circ\mathcal{R}^n}(\varepsilon_n^\ast(\alpha)) = \frac{\alpha}{n}(1+o(1))$ such that
$
  \varepsilon_n^\ast(\alpha)
  = \varepsilon_n\bigl(\alpha,\chi_{\mathrm{lo}}(\mathcal{R})\bigr)\,(1+o(1))
$ as $n\to\infty.$
\end{restatable}
We defer the proof of Theorem~\ref{thm:global-band-collapse} to
Appendix~\ref{app:proof-shuffle-index-structure}. 
In particular, the theorem applies to the case of $k$-RR ($k\geq 3$), and hence to mechanisms that are built on top of $k$-RR or share the same blanket structure, such as local hashing~\cite{wangLocallyDifferentiallyPrivate2017} and (for sufficiently small caps) PrivUnit~\cite{asiOptimalAlgorithmsMean2022}. 
Thus, in many practically relevant shuffle protocols, the blanket-divergence
analysis is asymptotically optimal in the sense of Theorem~\ref{thm:global-band-collapse}.

\section{Analysis of Blanket Divergence via the Fast Fourier Transform (FFT)}

The results in the previous section provide an asymptotic characterization of the blanket divergence. 
However, in applications one is ultimately interested in the finite-$n$ behavior of the shuffled mechanism.
To this end, we develop a complementary, non-asymptotic approach for computing the blanket divergence at finite $n$.

Our goal is to design an algorithm that, for a given local randomizer
$\mathcal R$ and parameters $(n,\varepsilon)$, produces rigorous upper and
lower bounds on
$\mathcal D^{\mathrm{blanket}}_{\mathrm e^\varepsilon,n,\mathcal R_{\mathrm{ref}},\gamma}
(\mathcal R_{x_1}\|\mathcal R_{x_1'})$ with explicitly controlled error and
provable running time guarantees. 
The main idea is to exploit the same
structure that underlies our CLT-based analysis: the blanket divergence can be
expressed in terms of sums of i.i.d.\ privacy amplification random variables.
This observation allows us to approximate the distribution of the sum by means
of the Fast Fourier Transform (FFT). 
To obtain rigorous bounds rather than
heuristic approximations, we carefully track the errors introduced by
truncation of the summands, discretization on a grid, and FFT wrap-around.
By combining these error bounds with the asymptotic expansion from
Lemma~\ref{lem:blanket-asymptotics-general}, we can both tune the numerical
parameters (window size, grid resolution, truncation level) and analyze the
overall complexity of the method.

The starting point is another representation of the blanket divergence.

\begin{restatable}{lemma}{blankettransformfft}
\label{lemma:trans_fft}
Let $M\sim 1+\mathrm{Bin}(n-1, \gamma)$. Then,

$
\mathcal{D}_{\mathrm{e}^\varepsilon,n,\mathcal{R}_\mathrm{ref},\gamma}^\mathrm{blanket}(\mathcal{R}_{x_1} \| \mathcal{R}_{x_1^\prime})=\Pr_{\substack{Y_1\sim\mathcal{R}_{x_1} \\ Y_{2:n}\stackrel{\mathrm{i.i.d.}}{\sim} \mathcal{R}_\mathrm{ref}}}\left[\sum_{i=1}^{M} l_\varepsilon(Y_i)>0\right]-\mathrm{e}^\varepsilon\Pr_{\substack{Y_1\sim\mathcal{R}_{x_1^\prime} \\ Y_{2:n}\stackrel{\mathrm{i.i.d.}}{\sim} \mathcal{R}_\mathrm{ref}}}\left[\sum_{i=1}^{M} l_\varepsilon(Y_i)>0\right].
$
\end{restatable}

The proof is provided in Appendix~\ref{app:proof-blanket-transform-fft}.
Here, we work with a representation of the blanket divergence that does not involve explicit expectations or integrals, which is particularly convenient for
deriving rigorous bounds on the truncation error.
In particular, in our FFT-based computation we must truncate the support of
the random variable in order to approximate an unbounded distribution on a finite grid.
This issue is relevant for non-pure local DP (see, e.g., the discussion of the Gaussian mechanism in Balle et al.~\cite{ballePrivacyBlanketShuffle2019}); hence this formulation is crucial for our purposes.

Su et al.~\cite{suDecompositionBasedOptimalBounds2025} also propose an FFT-based
accountant for the blanket divergence, but their running time scales
as $O(n^2)$ in the number of users when targeting a constant additive error.
In contrast, by adopting the concentration-inequality-based error control of
Gopi et al.~\cite{gopiNumericalCompositionDifferential2021} with our results in Section~\ref{sec:clt-analysis}, we can tune the
truncation and discretization parameters so that our algorithm runs in
near-linear time $\widetilde O(n)$ while achieving a prescribed relative error.

\subsection{FFT Algorithm}
\label{sec:fft}

Algorithms~\ref{alg:main_term_random_m}
and~\ref{alg:calc_final_prob} together form our FFT-based procedure.
As we show in Theorem~\ref{thm:fft-blanket-error} below, the output of
these algorithms yields rigorous upper and lower bounds on the blanket
divergence.

\renewcommand{\algorithmicrequire}{\textbf{Input:}}
\renewcommand{\algorithmicensure}{\textbf{Output:}}

\begin{algorithm}
\caption{\textsc{CalculatePMF}}
\label{alg:main_term_random_m}
\begin{algorithmic}[1]
\Require
\State $n, \gamma$: Parameters for the number of summands.
\State $F_{\mathrm{ref}}(t)$: The CDF of the summands $l(Y_i)$ for $i \geq 2$.
\State $[s, s+w^\mathrm{in}]$: Truncation interval for $l(Y_i)$, $i \geq 2$.
\State $h$: Bin width for discretization.
\State $w^\mathrm{out}$: The window size of FFT.

\State \Comment{\textbf{Step 1: Create PMF for a single truncated and discretized summand}}
\State $p_{\mathrm{out}} \gets F_{\mathrm{ref}}(s) + (1 - F_{\mathrm{ref}}(s+w^\mathrm{in}))$.
\State $p_{\mathrm{in}} \gets 1 - p_{\mathrm{out}}$.
\State Let $Z^{\mathrm{tr}}$ be a random variable for a single summand conditioned on being in $[s, s+w^\mathrm{in}]$.
\State The CDF of $Z^{\mathrm{tr}}$ is $F_{Z^{\mathrm{tr}}}(t) = \left(F_{\mathrm{ref}}(t) - F_{\mathrm{ref}}(s)\right)/p_{\mathrm{in}}$.
\State $\mu_{Z^\mathrm{tr}} \gets \left(
       (s+w^{\mathrm{in}})F_{\mathrm{ref}}(s+w^{\mathrm{in}})
       - sF_{\mathrm{ref}}(s)
       - \displaystyle\int_{s}^{s+w^{\mathrm{in}}} F_{\mathrm{ref}}(z)\,dz\right) / \left(F_{\mathrm{ref}}(s+w^{\mathrm{in}}) - F_{\mathrm{ref}}(s)\right).$
\State Define a discrete grid $\mathcal{G} = \{x_j = j \cdot h \mid j \in \mathbb{Z}\}$.
\State For each grid point $x_j$ in the truncated range, compute the PMF of the discretized variable $Z^{\mathrm{di}}$: $\mathbf{p}_{Z^{\mathrm{di}}}[j] \gets F_{Z^{\mathrm{tr}}}(x_j + h/2) - F_{Z^{\mathrm{tr}}}(x_j - h/2)$.
\State $\mu_{Z^\mathrm{di}} \gets \sum_j x_j  \mathbf{p}_{Z^\mathrm{di}}[j]$
\State \Comment{\textbf{Step 2: Calculate the characteristic function of the sum $S$}}
\State $\mu_{S^{\mathrm{tr}}} \gets (n-1)\,\gamma\,\mu_{Z}^{\mathrm{tr}}$, $\mu_{S^{\mathrm{di}}} \gets (n-1)\,\gamma\,\mu_{Z}^{\mathrm{di}}$
\State Set $N \gets \mathrm{Round}(w^\mathrm{out}/h)$, construct zero-padded $\mathbf{p}_{\mathrm{pad}}$ of length $N$ from $\mathbf{p}_{Z^{\mathrm{di}}}$, and let $\omega_k = \frac{2\pi k}{N h}$.
\State $\mathbf{\psi}_{Z^{\mathrm{di}}} \gets \texttt{FFT}\bigl(\texttt{FFTSHIFT}(\mathbf{p}_{\mathrm{pad}})\bigr)$
\State $\mathbf{\psi}_S \gets \big((1-\gamma)\,\mathbf{1} + \gamma\,\mathbf{\psi}_{Z^{\mathrm{di}}}\big)^{\odot (n-1)}$,  $\mathbf{\psi} \gets \big[e^{-i\,\omega_k\,\mu_{S^\mathrm{di}}}\big]_{k=0}^{N-1}$ where $\odot$ denotes element-wise operations.
\State $\mathbf{p}_{S-\mu_{S^\mathrm{di}}} \gets  \mathrm{Re}\left(\texttt{IFFTSHIFT}\left(\texttt{IFFT}\big(\mathbf{\psi}\odot\mathbf{\psi}_S\big)\right)\right)$ \Comment{PMF of the mean-centered sum $S-\mu_{S^\mathrm{di}}$. Re$(\cdot)$ denotes the real part.}
\State \Return $\mathbf{p}_{S-\mu_{S^\mathrm{di}}}$, $\mu_{S^{\mathrm{tr}}}$

\end{algorithmic}
\end{algorithm}

\renewcommand{\algorithmicrequire}{\textbf{Input:}}
\renewcommand{\algorithmicensure}{\textbf{Output:}}
\begin{algorithm}[H]
\caption{\textsc{CalculateMainTerm}}
\label{alg:calc_final_prob}
\begin{algorithmic}[1]
\Require
\State $\mathbf{p}_{S-\mu_{S^\mathrm{di}}}$, $\mu_{S^{\mathrm{tr}}}$: The PMF and the mean of the truncated sum.
\State $F(t)$: The CDF of $l_\varepsilon(Y;x_1,x_1',\mathcal R_{\mathrm{ref}})$ under $\mathcal{R}_{x_1}$ or $\mathcal{R}_{x_1^\prime}$.
\State $c$: The constant for the discretization error bound.

\State $p_{\text{main}} \gets 0$.
\For{each grid point $z_j$ with corresponding PMF value $\mathbf{p}_{S-\mu_{S^\mathrm{di}}}[j]$}
    \State $p_{\text{main}} \gets p_{\text{main}} + \mathbf{p}_{S-\mu_{S^\mathrm{di}}}[j] \cdot \left(1-F(-c-z_j-\mu_{S^{\mathrm{tr}}})\right)$.
\EndFor
\State \Return $p_{\text{main}}$
\end{algorithmic}
\end{algorithm}

\begin{restatable}{theorem}{fftblanketerror}\footnote{
Our error analysis assumes exact real arithmetic. 
While practical implementations relying on floating-point arithmetic introduce round-off errors, analyzing their propagation is orthogonal to the algorithmic errors characterized in this work.
}
\label{thm:fft-blanket-error}
Fix a distinct pair $x_1,x_1' \in \mathcal{X}$ and let $\mathcal R$ be a local randomizer.
Let the reference distribution be
$
  \mathcal R_{\mathrm{ref}}
  \in
  \bigl\{\mathcal R_{\mathrm{BG}}\bigr\}
  \,\cup\,
  \bigl\{\mathcal R_x : x\in\mathcal X\bigr\}.
$
For $Y\sim\mathcal R_{\mathrm{ref}}$, denote by $F_{\mathrm{ref}}$ the CDF of
the privacy amplification random variable
$l_\varepsilon(Y;x_1,x_1',\mathcal R_{\mathrm{ref}})$, and for
$j\in\{x_1,x_1'\}$ denote by $F_j$ the CDF of
$l_\varepsilon(Y;x_1,x_1',\mathcal R_{\mathrm{ref}})$ under $Y\sim\mathcal R_j$.
Let $(B_i)_{i=2}^n$ be i.i.d.\ $\mathrm{Bernoulli}(\gamma)$ and
$(Y_i)_{i=2}^n$ be i.i.d.\ with $\mathcal R_{\mathrm{ref}}$,
independent of $(B_i)_{i=2}^n$.

Given parameters $(n,\gamma,F_{\mathrm{ref}},s,w^{\mathrm{in}},h,w^{\mathrm{out}})$, we define the truncation event
$
  \perp
  :=
  \Bigl\{
    \exists i\in\{2,\dots,n\} :
    B_i = 1,\ 
    l_\varepsilon\bigl(Y_i;x_1,x_1',\mathcal R_{\mathrm{ref}}\bigr)
      \notin [s,s+w^{\mathrm{in}}]
  \Bigr\}
$ and let
$
  \bigl(\mathbf{p}_{S-\mu_S^{\mathrm{di}}},\mu_S^{\mathrm{tr}}\bigr)
  :=
  \textsc{CalculatePMF}
  \bigl(
    n,\gamma,F_{\mathrm{ref}},
    s,w^{\mathrm{in}},h,w^{\mathrm{out}}
  \bigr).
$
For $x\in\{x_1,x_1'\}$ and $c\ge0$, define
$
  P(c;x)
  :=
  \bigl(1-\Pr[\perp]\bigr)\,
  \textsc{CalculateMainTerm}
  \bigl(
    \mathbf{p}_{S-\mu_S^{\mathrm{di}}},\mu_S^{\mathrm{tr}},
    F_x,c
  \bigr).
$

Then, for any $c \ge 0$,
\begin{align*}
  &P(-c;x_1) - \mathrm e^\varepsilon P(c;x_1')
   - \delta_{\mathrm{err}}^{\mathrm{low}}(c,h,s,w^{\mathrm{in}},w^{\mathrm{out}})
   \leq
  D_{\mathrm e^\varepsilon,n,\mathcal R_{\mathrm{ref}},\gamma}^{\mathrm{blanket}}
    (\mathcal R_{x_1}\|\mathcal R_{x_1'})
  \\
  &\qquad\leq
  P(c;x_1) - \mathrm e^\varepsilon P(-c;x_1')
  + \delta_{\mathrm{err}}^{\mathrm{up}}(c,h,s,w^{\mathrm{in}},w^{\mathrm{out}}),
\end{align*}
where $\delta_{\mathrm{err}}^{\mathrm{low}}$ and
$\delta_{\mathrm{err}}^{\mathrm{up}}$ are the total error terms 
, each of which
admits a rigorous explicit bound.
\end{restatable}

See Appendix~\ref{app:proof-fft-blanket-error} for the full proof. At a high level, Algorithm~\ref{alg:main_term_random_m}
(\textsc{CalculatePMF}) constructs, via FFT, a discrete approximation to
the distribution of the sum
$
  \sum_{i=2}^{M} l_\varepsilon(Y_i),
$
where $M\sim 1+\mathrm{Bin}(n-1,\gamma)$ as in
Lemma~\ref{lemma:trans_fft} and $Y_i$ are i.i.d.\ draws from $\mathcal{R}_{\mathrm{ref}}$.
More precisely, we first truncate the privacy amplification random
variable to a bounded interval, discretize it on a uniform grid, and
then use the FFT to compute the probability mass function of the
resulting discretized sum. This step replaces the sum of
i.i.d.\ (continuous) variables by a finitely supported discrete
distribution that approximates
$\sum_{i=2}^{M} l_\varepsilon(Y_i)$ up to controlled
truncation, discretization and wrap-around errors.

Algorithm~\ref{alg:calc_final_prob}
(\textsc{CalculateMainTerm}) then uses this discrete distribution approximating $\sum_{i=2}^{M} l_\varepsilon(Y_i)$ together with the CDFs of $l_\varepsilon(Y)$ under
$\mathcal R_{x_1}$ and $\mathcal R_{x_1'}$ to approximate the
probabilities
$
  \Pr\!\Bigl[\sum_{i=2}^{M} l_\varepsilon(Y_i) > -\,l_\varepsilon(Y_1)\Bigr].
$ 
Because the distribution of $\sum_{i=2}^{M} l_\varepsilon(Y_i)$ has been
discretized, these probabilities reduce to finite sums over grid points,
and are therefore straightforward to evaluate numerically.

\subsection{Analysis of FFT Algorithm}
\label{sec:analysis_fft_algorithm}

Algorithms~\ref{alg:main_term_random_m} and~\ref{alg:calc_final_prob} depend
on several numerical parameters
$(c,s,w^{\mathrm{in}},h,w^{\mathrm{out}})$ that control truncation,
discretization, and FFT wrap-around.
Using the asymptotic characterization of the blanket divergence from
Lemma~\ref{lem:blanket-asymptotics-general}, we can tune these parameters to control both the relative error and the
computational cost of the FFT-based approximation.
The following theorem makes this precise.

\begin{restatable}{theorem}{fftrelativeerror}
\label{thm:fft-relative-error}
Fix a distinct pair $x_1,x_1' \in \mathcal{X}$ and let $\mathcal R\in\mathfrak{R}$.
Let the reference distribution be
$
  \mathcal R_{\mathrm{ref}}
  \in
  \bigl\{\mathcal R_{\mathrm{BG}}\bigr\}
  \,\cup\,
  \bigl\{\mathcal R_x : x\in\mathcal X\bigr\}.
$
Fix target knobs
$
  \eta_{\mathrm{main}},\,
  \eta_{\mathrm{trunc}},\,
  \eta_{\mathrm{disc}},\,
  \eta_{\mathrm{alias}}
  \in \mathbb{R}.
$
Let $\varepsilon_n\to0$ be a sequence of privacy levels such that
$
  \varepsilon_n= \Theta\bigl(\sqrt{\log n / n}\bigr).
$
We denote by $L_n$ and $U_n$ the lower and upper bounds from
Theorem~\ref{thm:fft-blanket-error}, that is,
$
  L_n
  :=
  P(-c;x_1) - \mathrm e^\varepsilon P(c;x_1')
  - \delta^{\mathrm{low}}_{\mathrm{err}}
$ and
$
  U_n
  :=
  P(c;x_1) - \mathrm e^\varepsilon P(-c;x_1')
  + \delta^{\mathrm{up}}_{\mathrm{err}}
$
We define our numerical approximation as the midpoint
$
  \widehat D_n := \frac{U_n + L_n}{2}.
$

Then we can tune parameters
$
  (c_n,s_n,w_n^{\mathrm{in}},h_n,w_n^{\mathrm{out}})
$
depending on
$(\eta_{\mathrm{main}},\eta_{\mathrm{trunc}},\eta_{\mathrm{disc}},\eta_{\mathrm{alias}})$
such that
\[
  \frac{|\widehat D_n - \mathcal{D}^{\mathrm{blanket}}|}{\mathcal{D}^{\mathrm{blanket}}}
  = O\bigl(
    \eta_{\mathrm{main}}
    + \eta_{\mathrm{trunc}}
    + \eta_{\mathrm{disc}}
    + \eta_{\mathrm{alias}}
  \bigr)
  \qquad (n\to\infty),
\]
when $|\mu_{S^{\mathrm{di}}}-\mu_{S^{\mathrm{tr}}}|=O(c_n)$\footnote{While this condition is not theoretically guaranteed for discrete mechanisms, we empirically confirmed $|\mu_{S^{\mathrm{di}}}-\mu_{S^{\mathrm{tr}}}|\leq c_n$ in the regime of our experiments.} and $l_\varepsilon(Y;x_1,x_1^\prime,\mathcal{R}_{\mathrm{ref}})$ is a non-lattice distribution\footnote{A finitely supported random variable is lattice iff its support is contained in a set of the form $a + h\mathbb{Z}$ for some $a\in\mathbb{R}$ and $h>0$, that is, all pairwise differences between support points are integer multiples of a common nonzero number. For mechanisms with at least three output symbols this requires special algebraic relations, so for generic choices of parameters $l_\varepsilon(Y;x_1,x_1',\mathcal{R}_{\mathrm{ref}})$ is non-lattice. The $2$-RR is a nongeneric exception: in that case $l_\varepsilon$ takes only two values and is lattice.}.
The overall running time of
Algorithms~\ref{alg:main_term_random_m}--\ref{alg:calc_final_prob} with these parameters
is
$
  O\left(\frac{n}{\eta_{\mathrm{main}}}(\log \frac{n}{\eta_{\mathrm{main}}})^3\right),
$
where $\log\log$ factors are omitted.
\end{restatable}

See Appendix~\ref{app:proof-fft-relative-error} for the full proof. In particular, in the regime
$\varepsilon_n = \Theta(\sqrt{\log n / n})$, the theorem
shows that we can prescribe a target relative error level
$\eta>0$ and choose the numerical knobs so that the relative error becomes $O(\eta)$, which leads to an overall running time
$
  \widetilde O\!\left(\frac{n}{\eta}\right),
$
up to polylogarithmic factors in $n$ and $1/\eta$.
We provide details on parameter selection and implementation in Appendix~\ref{sec:fft-implementation-details}.

We briefly discuss regimes in which the algorithm can become less effective in practice. 
First, when the shuffle index $\chi$ is small, the relevant distributions for the privacy-amplification sum are effectively more spread out, so representing them on a discrete grid requires a larger window and more bins; heuristically, the required grid size (and hence runtime) can grow like $O(1/\chi^2)$, making computations expensive for small $\chi$.
In practice, as discussed in Section~\ref{sec:shuffle_index}, privacy amplification by shuffling is meaningful only when the shuffle index is at least moderately large; the parameter regime of primary interest typically corresponds to larger $\chi$, where the FFT computation remains well behaved.
Second, our parameter-tuning rules are asymptotic, so for small $n$ they may be miscalibrated and yield insufficient accuracy.
That said, the method produces certified upper and lower bounds, so one can directly inspect the resulting gap to assess tightness and, if needed, retune the numerical parameters to achieve the desired accuracy. 
Finally, while floating-point round-off is not captured by the real-arithmetic analysis, prior work shows that the FFT-based computations can remain accurate down to about $10^{-10}$ when using long double precision~\cite{gopiNumericalCompositionDifferential2021}, which is sufficient for the typical range of $\delta$ values of interest in DP.

\section{Experiments}
\begin{figure}[t]
  \centering
  \begin{subfigure}{0.59\linewidth}
    \centering
    \begin{tabular}{cc}
      \includegraphics[width=0.48\linewidth,height=0.18\textheight,keepaspectratio]{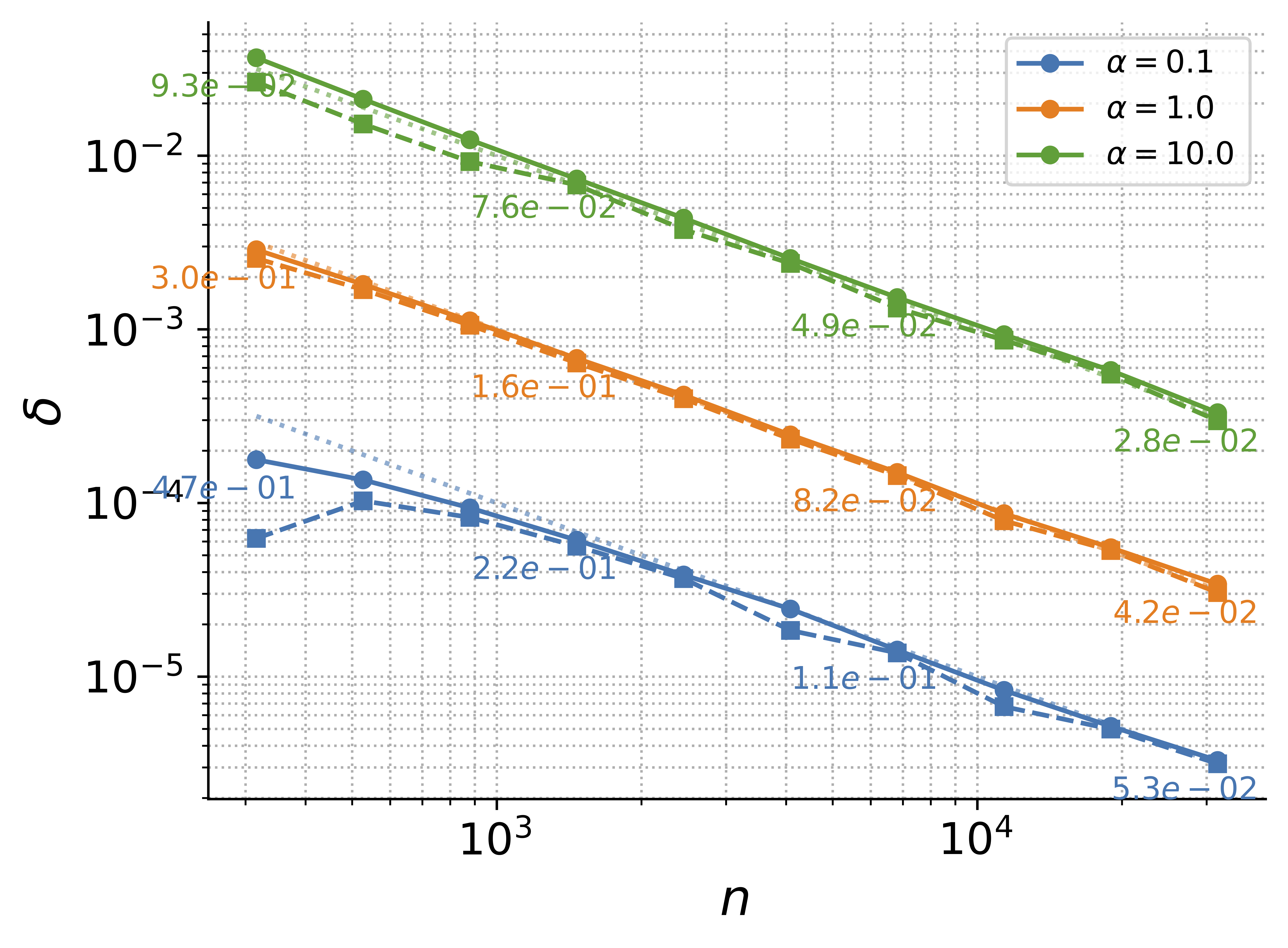} &
      \includegraphics[width=0.48\linewidth,height=0.18\textheight,keepaspectratio]{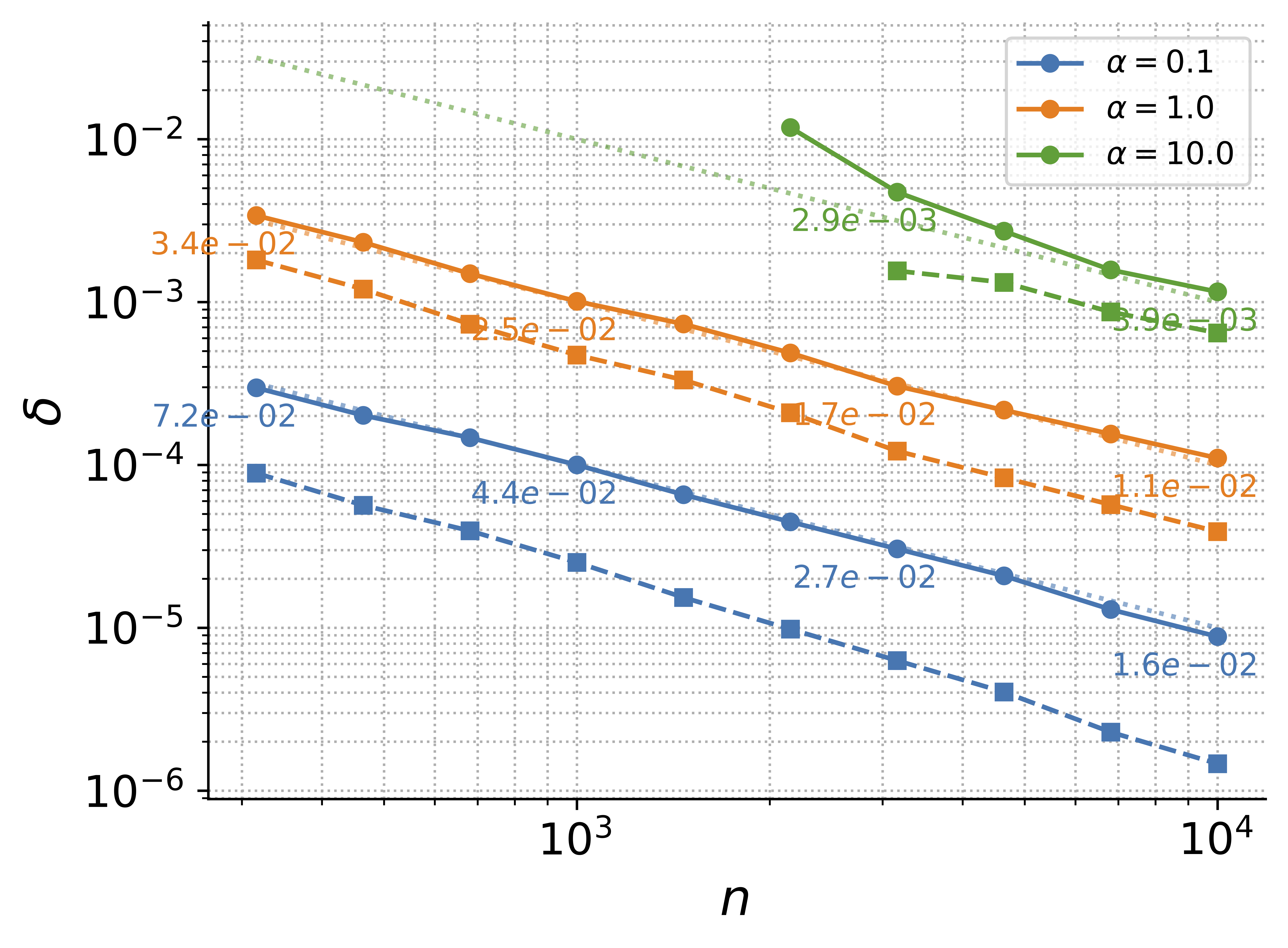}
    \end{tabular}
    \caption{Upper ($\circ$) and lower ($\square$) shuffled mechanism's DP $\delta$ for $\varepsilon_n(\alpha,\underline\chi_{\mathrm{lo}})$ of
      3-RR with $\varepsilon_0=2$ (left) and the Gaussian mechanism with
      $\sigma_0=2$ (right), both with annotated $\varepsilon_n$; the dotted
      line indicates $\alpha/n$.}
    \label{fig:delta_bands}
  \end{subfigure}
  \begin{subfigure}{0.4\linewidth}
    \centering
    \includegraphics[width=\linewidth,height=0.18\textheight,keepaspectratio]{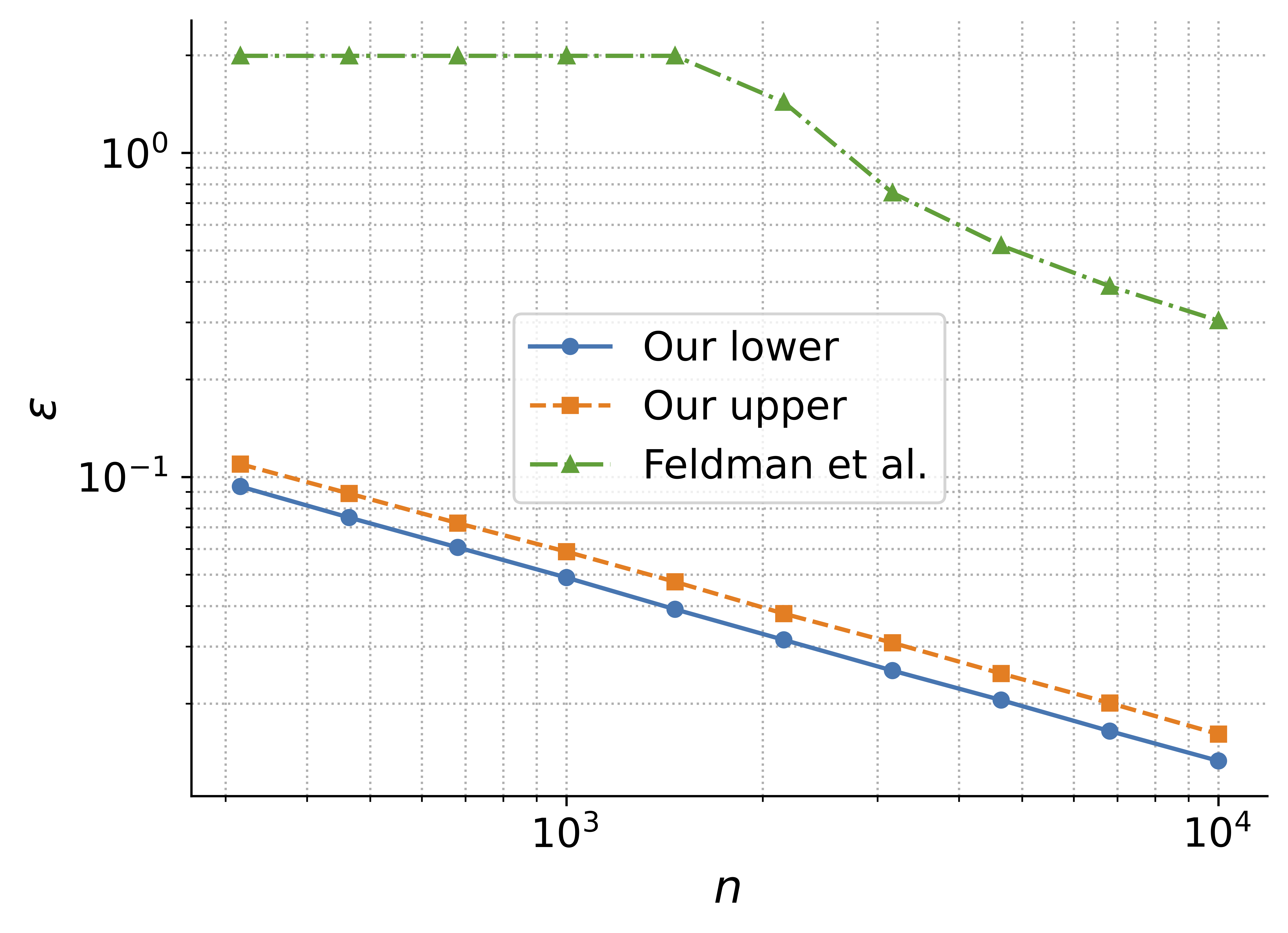}
    \caption{Comparison of the shuffled $\varepsilon$ as a function of $n$ between our analysis and that of \citet{feldmanHidingClonesSimple2022}. We use the Gaussian mechanism with $\sigma_0=2$ and fix $\delta=10^{-5}$.}
    \label{fig:vs_feldman}
  \end{subfigure}
  \begin{subfigure}{0.31\linewidth}
    \centering
    \includegraphics[width=0.9\linewidth,height=0.18\textheight,keepaspectratio]{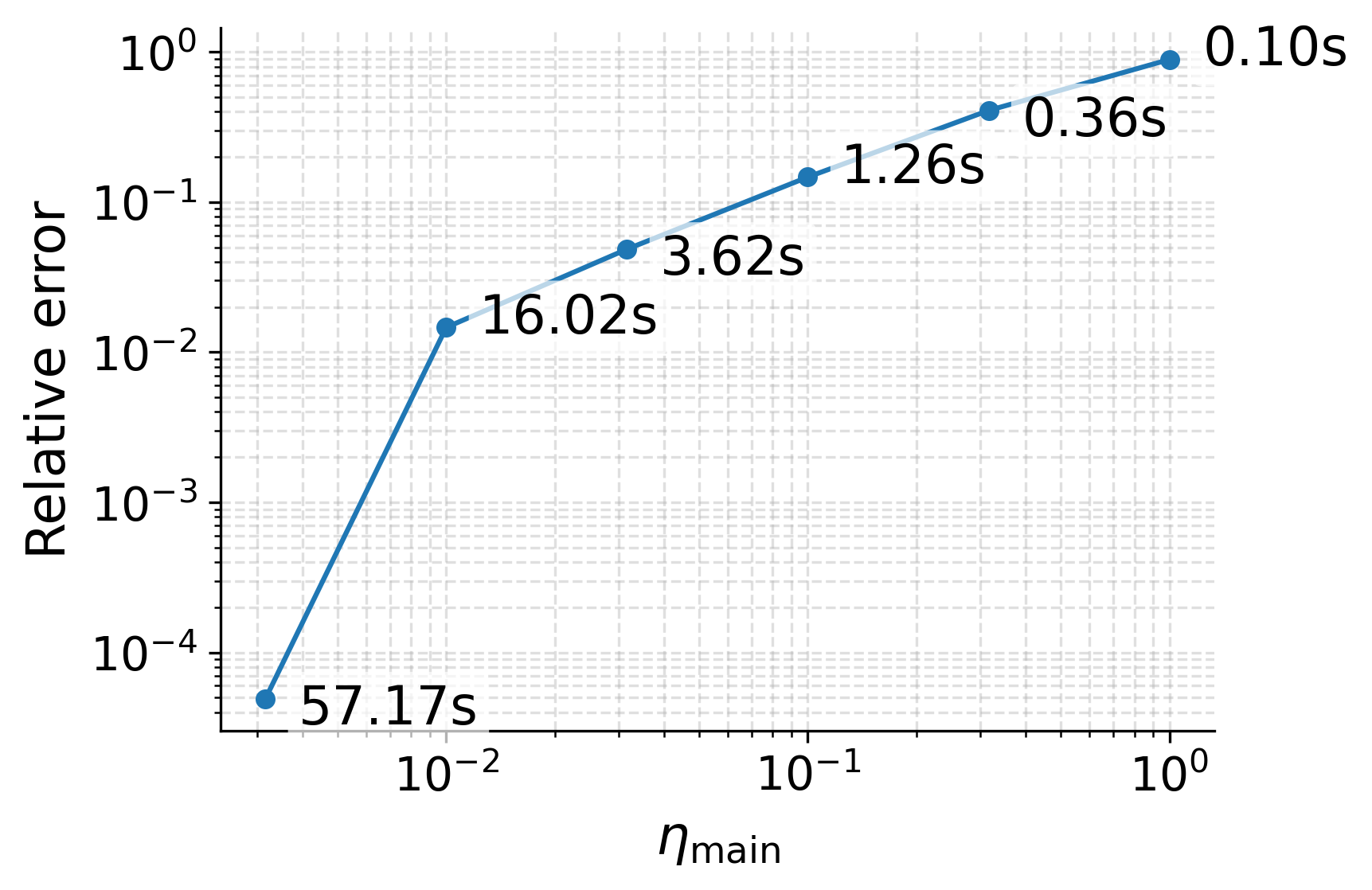}
    \caption{Relative error as a function of
      the accuracy knob $\eta_{\mathrm{main}}$ with annotated runtime.}
    \label{fig:eta_dependence}
  \end{subfigure}
  \begin{subfigure}{0.31\linewidth}
    \centering
    \includegraphics[width=\linewidth,height=0.18\textheight,keepaspectratio]{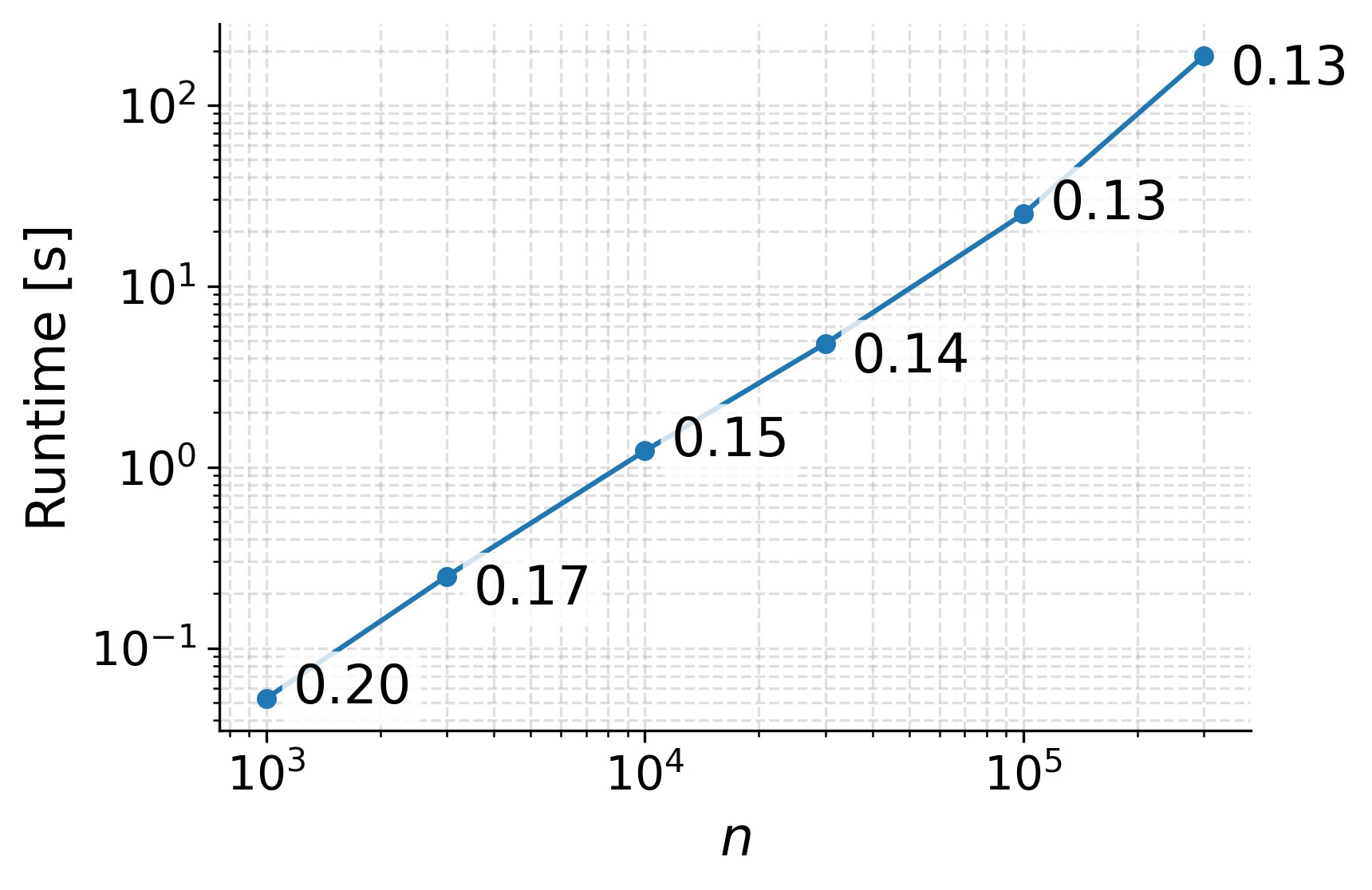}
    \caption{Runtime as a function of
      the number of users $n$ with annotated relative error.}
    \label{fig:n_dependence}
  \end{subfigure}
  \begin{subfigure}{0.31\linewidth}
    \centering
    \includegraphics[width=\linewidth,height=0.18\textheight,keepaspectratio]{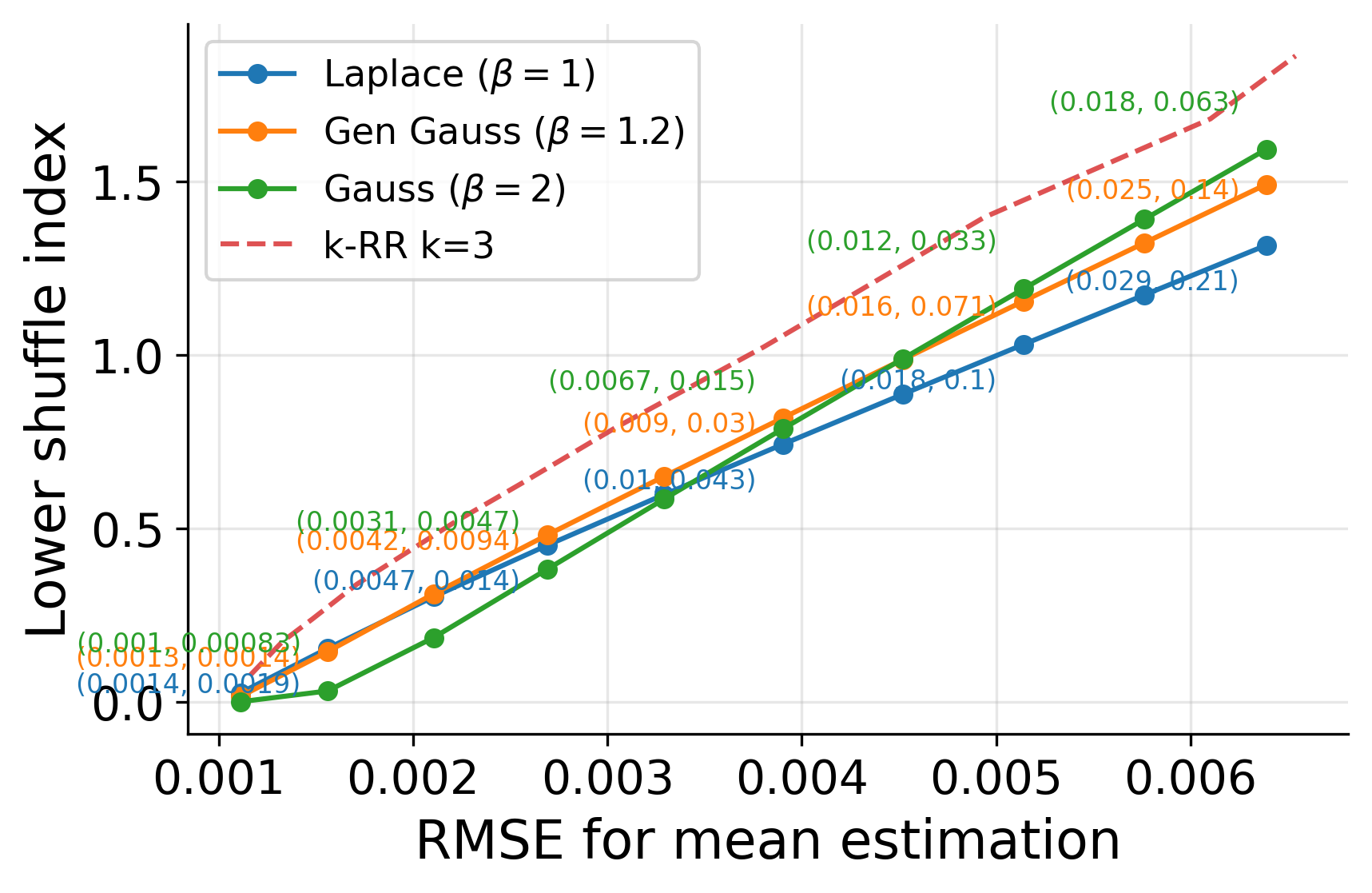}
    \caption{$\chi_\mathrm{lo}$ as a function of the RMSE for mean estimation with annotated RMSE for variance and $3$rd moment estimation.}
    \label{fig:shuffle_index_vs_mean_sigma}
  \end{subfigure}

  \caption{
    Experimental validation of the FFT-based blanket-divergence accountant and
    the shuffle-index analysis.
  }
  \label{fig:experiments}
\end{figure}

In this section, we complement our theoretical results with five experiments.
The first two visualize the privacy curves from Theorems~\ref{thm:moderate-deviation-shuffle} and~\ref{thm:global-band-collapse}.
The next two illustrate the accuracy-runtime trade-offs of our FFT-based accountant from
Theorems~\ref{thm:fft-blanket-error} and~\ref{thm:fft-relative-error}.
The final experiment compares generalized Gaussian mechanisms with $k$-RR in a distribution-estimation task.

Unless stated otherwise, all computations are carried out in double precision and we choose
$(c,s,w^{\mathrm{in}},h,w^{\mathrm{out}})$ according to
Theorem~\ref{thm:fft-blanket-error} and set $\eta_{\mathrm{trunc}},
  \eta_{\mathrm{disc}},\eta_{\mathrm{alias}}=10^{-3}$.
For each mechanism we first identify the asymptotically worst-case neighboring
pair $(x_1,x_1')$ and the corresponding shuffle indices
$\chi_{\mathrm{lo}}$ and $\chi_{\mathrm{up}}$ as in
Theorem~\ref{thm:moderate-deviation-shuffle}, and then run the FFT-based
accountant at this pair.

\subsection{Privacy Curves and Shuffle Index}
\label{subsec:exp-shuffle-index}

We empirically evaluate the privacy guarantees of
Theorem~\ref{thm:moderate-deviation-shuffle}.
For a fixed target constant $\alpha>0$ and a given local randomizer we set
$
  \varepsilon_n
  = \varepsilon_n(\alpha,\chi_{\mathrm{lo}})
$
and explore the behavior of the shuffled mechanism's hockey-stick divergence
$\delta_n$ as a function of $n$ for two mechanisms.
That is, we use the blanket-based \emph{upper} and \emph{lower} bounds on the
\emph{true} shuffled hockey-stick divergence:
the Balle et al.\ upper bound
(Lemma~\ref{lem:upper-blanket-divergence}) and the Su et al.\ lower bound
(Theorem~\ref{theorem:lower-blanket-divergence}), both evaluated via our
FFT accountant at the asymptotically worst-case pair.

The left panel of Fig.~\ref{fig:delta_bands} shows the results for the case of $3$-RR with $\varepsilon_0=2$.
For each $n$ we plot the upper and lower bounds on the true hockey-stick divergence $\delta_n$ of the shuffled mechanism corresponding to $\varepsilon_n(\alpha,\chi_{\mathrm{lo}})$.
Across the entire range of $n$ the two curves are practically indistinguishable.
This behavior is exactly what Theorem~\ref{thm:global-band-collapse}
predicts for $k$-RR with $k\ge3$.
On the theoretical side, the lower and upper shuffle indices coincide,
$\chi_{\mathrm{lo}} = \chi_{\mathrm{up}}$, so the
asymptotic band in Theorem~\ref{thm:moderate-deviation-shuffle} collapses.
Empirically, we see that this collapse already occurs for moderate values
of $n$, confirming that the blanket-divergence analysis is effectively
optimal for this mechanism.

The right panel of Fig.~\ref{fig:delta_bands} reports the same experiment for the Gaussian local randomizer with noise parameter $\sigma=2$.
Here the upper and lower curves form a visible but relatively narrow band.
This matches our theoretical findings: the Gaussian mechanism does not satisfy the structural condition of
Theorem~\ref{thm:global-band-collapse}, so the shuffle indices $\chi_{\mathrm{lo}}$ and $\chi_{\mathrm{up}}$ need not coincide.
Nevertheless, their ratio stays close to one (above $0.7$ in our parameter range), and the band remains tight across all $n$.

Next, we compare our direct analysis of the blanket divergence with the mechanism-agnostic shuffle analysis of Feldman et al.~\cite{feldmanHidingClonesSimple2022} of the calibrated Gaussian local randomizer~\cite{balleImprovingGaussianMechanism2018}.
Fig.~\ref{fig:vs_feldman} shows the results.
Because of the mechanism-agnostic nature and the reliance on a coarse reduction to a pure LDP surrogate, the resulting bounds incur a large constant-factor loss.
Empirically, we observe that even at $n=10^3$, the method predicts essentially no privacy amplification, and for larger $n$ the predicted shuffled $\varepsilon$ remains $10$ times larger than our direct analysis.

\subsection{Accuracy and Runtime of the FFT Accountant}
\label{subsec:exp-fft}

We empirically validate the error and complexity guarantees of
Theorems~\ref{thm:fft-blanket-error} and~\ref{thm:fft-relative-error}.
Throughout this subsection we fix the local randomizer to be $3$-RR with $\varepsilon_0=2$ and vary either the
accuracy knob $\eta_{\mathrm{main}}$ or the number of users~$n$.

Fig.~\ref{fig:eta_dependence} shows the effect of the parameter
$\eta_{\mathrm{main}}$ on both accuracy and runtime.
For each value of $\eta_{\mathrm{main}}$ we run the FFT-based algorithm for a
fixed $(n,\varepsilon)$.
Theorem~\ref{thm:fft-blanket-error} provides rigorous lower and upper bounds
on the blanket divergence,
$
  D_{\mathrm{low}} \;\le\;
  \mathcal{D}^{\mathrm{blanket}}
  \;\le\; D_{\mathrm{up}},
$
so the true value is guaranteed to lie in the interval
$[D_{\mathrm{low}},D_{\mathrm{up}}]$.
As our notion of relative error, we plot the certified relative uncertainty bandwidth
$
  \frac{D_{\mathrm{up}} - D_{\mathrm{low}}}{D_{\mathrm{up}}}.
$
As predicted by Theorem~\ref{thm:fft-relative-error}, the relative error shrinks approximately linearly as $\eta_{\mathrm{main}}$ decreases, while the
runtime grows roughly like $1/\eta_{\mathrm{main}}$.

In Fig.~\ref{fig:n_dependence} we fix $\eta_{\mathrm{main}}$ and vary the number of users $n$.
For each $n$ we again compute the interval
$[D_{\mathrm{low}},D_{\mathrm{up}}]$ and plot both the runtime and the
relative error.
The relative error remains essentially flat as $n$ increases, confirming
that our parameter selection keeps the relative error under control uniformly
in $n$.
At the same time, the runtime grows almost linearly in $n$, with only mild
logarithmic deviations.
This matches the theoretical $\widetilde O(n/\eta)$ complexity of
Theorem~\ref{thm:fft-relative-error}.

\subsection{Distribution Estimation}
\label{sec:distribution-estimation}
Finally, we compare the generalized Gaussian mechanisms with $k$-RR in a distribution estimation task.
We fix the data-generating distribution to be $X \sim \mathrm{Unif}[0,1]$ and set the number of users to $n = 10^5$.
For each local mechanism, we evaluate the empirical root-mean-squared error (RMSE) of the mean estimation.
Fig.~\ref{fig:shuffle_index_vs_mean_sigma} plots the RMSE on the horizontal axis and the lower shuffle index~$\chi_{\mathrm{lo}}$ on the vertical axis.
For generalized Gaussian mechanisms we additionally annotate each point with the RMSEs for variance and third-moment estimation, whereas for $k$-RR we only plot the mean-accuracy vs.\ shuffle-index trade-off.

In this experiment we focus on $k$-RR with $k=3$, which we found to yield the best shuffle index among the values of $k$ under the same mean-estimation accuracy.
Under this choice, the figure shows that, for the same level of mean RMSE, $3$-RR achieves the largest lower shuffle index in our parameter range, and hence provides the strongest privacy.
However, a crucial caveat is that $3$-RR (and, more generally, $k$-RR) does not natively support estimation of higher-order moments: estimating the variance or third central moment typically requires designing a separate mechanism (and hence additional privacy budget).
In contrast, generalized Gaussian mechanisms simultaneously support estimation of multiple moments from the same shuffled output.
Fig.~\ref{fig:shuffle_index_vs_mean_sigma} reveals that the generalized Gaussian family exhibits a nontrivial dependence on the shape parameter~$\beta$: the value of $\beta$ that optimizes the shuffle index varies with the target mean-estimation accuracy.
In particular, in more noisy regimes (larger local noise), the Gaussian mechanism ($\beta=2$) attains the largest lower shuffle index among the generalized Gaussian family.
Moreover, over a wide range of mean-accuracy levels, the Gaussian mechanism also achieves the smallest RMSE for variance and third-moment estimation among the generalized Gaussian mechanisms, as indicated by the annotations.
Taken together, these results suggest that generalized Gaussian mechanisms ($\beta>1$) provide a better privacy-utility trade-off than pure local DP mechanisms.
Importantly, this kind of optimization is made possible by our unified framework beyond pure local DP.

\section{Limitation and Conclusion}



Our moderate-deviation guarantees are formulated in terms of an \emph{asymptotically} worst-case neighboring pair, that is, a pair $(x_1^\ast,x_1^{\prime\ast})$ that attains the worst-case shuffle index in the limit $n\to\infty$.
While this is natural from the point of view of our CLT-based analysis, it does not rule out the possibility that, for finite $n$, special pairs might induce slightly worse $(\varepsilon,\delta)$ trade-offs than those induced by the asymptotic worst-case pair.
Establishing conditions under which the asymptotically worst-case pair is also worst-case at finite $n$ is a problem that we leave for future work.

Our work is, to the best of our knowledge, the first to go beyond the pure local DP viewpoint and organize privacy amplification by shuffling via the shuffle index.
Our framework provides a unified way to reason about shuffle amplification and to obtain tight $(\varepsilon,\delta)$ guarantees.
We believe that these conceptual and technical advantages are valuable both for the theoretical understanding of shuffle DP and for its practical deployment in real systems.

\bibliographystyle{plain}
\bibliography{shuffle-gauss}


\appendix






\section{Missing Proofs}

\subsection{The proof of Lemma~\ref{lem:blanket-asymptotics-general}}
\label{app:proof-blanket-asymptotics-general}

\blanketasymptoticsgeneral*

To prove Lemma~\ref{lem:blanket-asymptotics-general}, we start by deriving two alternative representations of the blanket divergence that are more amenable to asymptotic analysis.

\begin{restatable}[Transformations for the blanket divergence]{lemma}{blankettransform}
\label{lem:blanket-transform}
Let $\mathcal{R}$ be a local randomizer and a fixed pair
$x_1,x_1' \in \mathcal X$.
Let the reference distribution be
$
  \mathcal R_{\mathrm{ref}}
  \in
  \bigl\{\mathcal R_{\mathrm{BG}}\bigr\}
  \,\cup\,
  \bigl\{\mathcal R_x : x\in\mathcal X\bigr\}.
$
For $n\in\mathbb N$, $\gamma\in(0,1]$, and $\varepsilon\geq 0$, define the
following random variables:
  $Y_1,\dots,Y_n\stackrel{\mathrm{i.i.d.}}{\sim}\mathcal R_{\mathrm{ref}}$,
  $X_i := l_{\varepsilon}(Y_i;x_1,x_1',\mathcal R_{\mathrm{ref}})$
    for $i=1,\dots,n$,
  $B_1,\dots,B_n$ i.i.d.\ $\mathrm{Bernoulli}(\gamma)$, independent of
    $(Y_i)_{i=1}^n$,
    and
  $Z_i := B_i X_i$ and $S_m := \sum_{i=1}^m Z_i$ for $m\ge1$.
Let
$
  \mu_\varepsilon := \mathbb E[Z_1]
$ and
$
  \sigma_\varepsilon^2 := \mathrm{Var}(Z_1)\in(0,\infty).
$
For each $m\ge1$ define the normalized sum
$
  T_m
  := \frac{S_m - m\mu_\varepsilon}{\sigma_\varepsilon\sqrt{m}}
$
and its CDF
$
  F_m(x) := \Pr\left[T_m\le x\right]
$
together with the thresholds
$
  t_m := -\frac{\mu_\varepsilon\sqrt{m}}{\sigma_\varepsilon},
  $
and
$
  \eta_m(z) := \frac{z}{\sigma_\varepsilon\sqrt{m}}.
$

Then the blanket divergence admits the following two equivalent representations:
  \begin{equation}
    \label{eq:blanket-transform-i}
    \begin{aligned}
    \mathcal D^{\mathrm{blanket}}_{e^{\varepsilon},n,\mathcal R_{\mathrm{ref}},\gamma}(
      \mathcal{R}_{x_1} \| \mathcal{R}_{x_1^\prime})
    &= \frac{1}{\gamma}\,
       \mathbb E\Bigl[
         Z_1\,
         \Pr\bigl[
           S_{n-1} > -Z_1
         \bigr]
       \Bigr] 
    &= \frac{1}{\gamma}\,
       \mathbb E\!\left[
         Z_1\,
         \bigl\{
           1 - F_{n-1}\bigl(t_{n-1} - \eta_{n-1}(Z_1)\bigr)
         \bigr\}
       \right].
    \end{aligned}
  \end{equation}

  \begin{equation}
    \label{eq:blanket-transform-ii}
    \begin{aligned}
    \mathcal D^{\mathrm{blanket}}_{e^{\varepsilon},n,\mathcal R_{\mathrm{ref}},\gamma}(
      \mathcal{R}_{x_1} \| \mathcal{R}_{x_1^\prime})
    &= \frac{1}{n\gamma}
       \int_0^\infty
         \Pr[S_n>u]\,du 
    &= \frac{\sigma_\varepsilon}{\sqrt{n}\gamma}\,
       \int_{t_n}^\infty \bigl(1-F_n(x)\bigr)\,dx.
    \end{aligned}
  \end{equation}
\end{restatable}

We retain both transformations because the conditions required to identify the leading term in Lemma~\ref{lem:blanket-asymptotics-general} differ between the continuous (Cont) and bounded (Bound) assumptions of the local randomizer (see Assumption~\ref{assump:regularity}). 
In the continuous case, we apply Theorem~\ref{thm:edgeworth-triangular} to the representation \eqref{eq:blanket-transform-i}.
In the bounded case (Bound), we instead apply Theorem~\ref{thm:slastnikov-triangular} to the representation \eqref{eq:blanket-transform-ii}. 
The resulting leading terms coincide, and we obtain the result.

\begin{proof}
We start from Balle et al.'s representation (Lemma~\ref{lem:upper-blanket-divergence}), adapted to our notation.
Let $Y_1,\dots,Y_n \stackrel{\mathrm{i.i.d.}}{\sim}\mathcal R_{\mathrm{ref}}$,
let $X_i := l_\varepsilon(Y_i;x_1,x_1',\mathcal R_{\mathrm{ref}})$, and let
$B_1,\dots,B_n$ be i.i.d.\ $\mathrm{Bernoulli}(\gamma)$, independent of
$(Y_i)_{i=1}^n$. Set $Z_i := B_i X_i$ and
\[
  S_n := \sum_{i=1}^n Z_i.
\]
Then,
\begin{equation}
  \label{eq:balle-representation}
  \mathcal D^{\mathrm{blanket}}_{e^{\varepsilon},n,\mathcal R_{\mathrm{ref}},\gamma}(
    \mathcal{R}_{x_1} \| \mathcal{R}_{x_1^\prime})
  = \frac{1}{n\gamma}\,
    \mathbb E\bigl[(S_n)_+\bigr],
  \qquad
  (S_n)_+ := \max\{S_n,0\}.
\end{equation}
Indeed, if $M\sim\mathrm{Bin}(n,\gamma)$ and $(Y_i)_{i\ge1}$ are i.i.d.\
$\mathcal R_{\mathrm{ref}}$ then
$\sum_{i=1}^M l_\varepsilon(Y_i)$ has the same distribution as
$\sum_{i=1}^n Z_i$, so \eqref{eq:balle-representation} is equivalent to
Balle et al.'s original formula.

We now show that \eqref{eq:balle-representation} implies the two
representations in the statement.

\smallskip
\noindent\emph{Proof of \eqref{eq:blanket-transform-i}.}
We have
\[
  (S_n)_+
  = S_n \mathbf 1\{S_n>0\}
  = \Bigl(\sum_{i=1}^n Z_i\Bigr)\mathbf 1\{S_n>0\}
  = \sum_{i=1}^n Z_i\,\mathbf 1\{S_n>0\}.
\]
Taking expectations and using linearity of expectation,
\[
  \mathbb E[(S_n)_+]
  = \sum_{i=1}^n
      \mathbb E\bigl[Z_i\,\mathbf 1\{S_n>0\}\bigr].
\]
Since $(Z_i)_{i=1}^n$ are i.i.d., all summands are equal, hence
\[
  \mathbb E[(S_n)_+]
  = n\,\mathbb E\bigl[Z_1\,\mathbf 1\{S_n>0\}\bigr].
\]
Substituting this into \eqref{eq:balle-representation} yields
\begin{equation}
  \label{eq:D-blanket-Z1}
  \mathcal D^{\mathrm{blanket}}_{e^{\varepsilon},n,\mathcal R_{\mathrm{ref}},\gamma}(
    \mathcal{R}_{x_1} \| \mathcal{R}_{x_1^\prime})
  = \frac{1}{\gamma}\,
    \mathbb E\bigl[Z_1\,\mathbf 1\{S_n>0\}\bigr].
\end{equation}

Next, decompose the sum $S_n$ as
\[
  S_n = Z_1 + \sum_{i=2}^n Z_i.
\]
Define
\[
  S_{2:n} := \sum_{i=2}^n Z_i.
\]
Then
\[
  \{S_n>0\} = \{Z_1 + S_{2:n} > 0\}.
\]
Using the definition of conditional expectation and the fact that
$Z_1$ is independent of $S_{2:n}$,
\[
\begin{aligned}
  \mathbb E\bigl[Z_1\,\mathbf 1\{S_n>0\}\bigr]
  &= \mathbb E\Bigl[
       Z_1\,\mathbf 1\{Z_1 + S_{2:n}>0\}
     \Bigr] \\
  &= \mathbb E\Bigl[
       Z_1\,\Pr \bigl[
         Z_1 + S_{2:n}>0
         \,\big|\,
         Z_1
       \bigr]
     \Bigr] \\
  &= \mathbb E\Bigl[
       Z_1\,\Pr \bigl[
         S_{2:n} > -Z_1
       \bigr]
     \Bigr].
\end{aligned}
\]
Since $(Z_i)_{i=1}^n$ are i.i.d., we have
\[
  S_{2:n} = \sum_{i=2}^n Z_i \;\stackrel{d}{=}\; \sum_{i=1}^{n-1} Z_i =: S_{n-1},
\]
and $S_{2:n}$ is independent of $Z_1$. Hence
\[
  \Pr \bigl[S_{2:n} > -Z_1\bigr]
  = \Pr \bigl[S_{n-1} > -Z_1\bigr],
\]
and therefore
\begin{equation}
  \label{eq:D-blanket-Sn-1}
  \mathcal D^{\mathrm{blanket}}_{e^{\varepsilon},n,\mathcal R_{\mathrm{ref}},\gamma}(
    \mathcal{R}_{x_1} \| \mathcal{R}_{x_1^\prime})
  = \frac{1}{\gamma}\,
    \mathbb E\bigl[
      Z_1\,\Pr \bigl[S_{n-1} > -Z_1\bigr]
    \bigr].
\end{equation}

By definition,
\[
  \mu_\varepsilon := \mathbb E[Z_1],
  \qquad
  \sigma_\varepsilon^2 := \mathrm{Var}(Z_1) \in (0,\infty),
\]
and
\[
  T_{n-1}
  := \frac{S_{n-1} - (n-1)\mu_\varepsilon}{\sigma_\varepsilon\sqrt{n-1}},
  \qquad
  F_{n-1}(x) := \Pr\left[T_{n-1}\le x\right].
\]
We can rewrite $S_{n-1}$ as
\[
  S_{n-1} = (n-1)\mu_\varepsilon + \sigma_\varepsilon\sqrt{n-1}\,T_{n-1}.
\]
For any real $z$,
\[
\begin{aligned}
  \{S_{n-1} > -z\}
  &\iff
  \bigl\{
    (n-1)\mu_\varepsilon + \sigma_\varepsilon\sqrt{n-1}\,T_{n-1}
    > -z
  \bigr\} \\
  &\iff
  \left\{
    T_{n-1}
    >
    \frac{-z - (n-1)\mu_\varepsilon}{\sigma_\varepsilon\sqrt{n-1}}
  \right\}.
\end{aligned}
\]
Define
\[
  t_{n-1}
  := -\frac{\mu_\varepsilon\sqrt{n-1}}{\sigma_\varepsilon},
  \qquad
  \eta_{n-1}(z)
  := \frac{z}{\sigma_\varepsilon\sqrt{n-1}},
\]
so that
\[
  \frac{-z - (n-1)\mu_\varepsilon}{\sigma_\varepsilon\sqrt{n-1}}
  = t_{n-1} - \eta_{n-1}(z).
\]
Therefore,
\[
  \Pr \left[S_{n-1}>-z\right]
  = \Pr \left(
      T_{n-1} > t_{n-1} - \eta_{n-1}(z)
    \right)
  = 1 - F_{n-1}\left(t_{n-1} - \eta_{n-1}(z)\right).
\]
Applying this with $z=Z_1$ in \eqref{eq:D-blanket-Sn-1}, we obtain
\[
\begin{aligned}
  \mathcal D^{\mathrm{blanket}}_{e^{\varepsilon},n,\mathcal R_{\mathrm{ref}},\gamma}(
    \mathcal{R}_{x_1} \| \mathcal{R}_{x_1^\prime})
  &= \frac{1}{\gamma}\,
     \mathbb E\Bigl[
       Z_1\,
       \bigl\{
         1 - F_{n-1}\bigl(t_{n-1} - \eta_{n-1}(Z_1)\bigr)
       \bigr\}
     \Bigr],
\end{aligned}
\]
which is exactly \eqref{eq:blanket-transform-i}.

\smallskip
\noindent\emph{Proof of \eqref{eq:blanket-transform-ii}.}
For any real number $x$ we have the identity
\[
  x_+ = \max\{x,0\}
      = \int_0^\infty \mathbf 1\{x>u\}\,du.
\]
Applying this with $x = S_n$ and using Tonelli's theorem (Fubini's theorem
for nonnegative functions) gives
\[
\begin{aligned}
  \mathbb E[(S_n)_+]
  &= \mathbb E\Bigl[\int_0^\infty \mathbf 1\{S_n>u\}\,du\Bigr] \\
  &= \int_0^\infty \Pr\left[S_n>u\right]\,du.
\end{aligned}
\]
Substituting into \eqref{eq:balle-representation},
\begin{equation}
  \label{eq:D-blanket-tail-Sn}
  \mathcal D^{\mathrm{blanket}}_{e^{\varepsilon},n,\mathcal R_{\mathrm{ref}},\gamma}(
    \mathcal{R}_{x_1} \| \mathcal{R}_{x_1^\prime})
  = \frac{1}{n\gamma}
    \int_0^\infty
      \Pr\left[S_n>u\right]\,du.
\end{equation}

We now express the tail probability in terms of the normalized sum $T_n$.
By definition,
\[
  T_n
  := \frac{S_n - n\mu_\varepsilon}{\sigma_\varepsilon\sqrt{n}},
  \qquad
  F_n(x) := \Pr\left[T_n\le x\right],
\]
so that
\[
  S_n = n\mu_\varepsilon + \sigma_\varepsilon\sqrt n\,T_n.
\]
For any $u\in\mathbb R$,
\[
\begin{aligned}
  \Pr\left[S_n>u\right]
  &= \Pr\Bigl[
       n\mu_\varepsilon + \sigma_\varepsilon\sqrt n\,T_n
       > u
     \Bigr] \\
  &= \Pr\Bigl[
       T_n
       >
       \frac{u - n\mu_\varepsilon}{\sigma_\varepsilon\sqrt n}
     \Bigr] \\
  &= 1 - F_n\Bigl(
       \frac{u - n\mu_\varepsilon}{\sigma_\varepsilon\sqrt n}
     \Bigr).
\end{aligned}
\]
Perform the change of variables
\[
  x = \frac{u - n\mu_\varepsilon}{\sigma_\varepsilon\sqrt n},
  \qquad
  u = n\mu_\varepsilon + \sigma_\varepsilon\sqrt n\,x,
  \qquad
  du = \sigma_\varepsilon\sqrt n\,dx.
\]
When $u=0$, the corresponding lower limit is
\[
  x = \frac{0 - n\mu_\varepsilon}{\sigma_\varepsilon\sqrt n}
    = -\frac{\mu_\varepsilon\sqrt n}{\sigma_\varepsilon}
    =: t_n.
\]
Thus
\[
\begin{aligned}
  \int_0^\infty \Pr[S_n>u]\,du
  &= \int_0^\infty
       \Bigl\{
         1 - F_n\Bigl(
               \frac{u - n\mu_\varepsilon}{\sigma_\varepsilon\sqrt n}
             \Bigr)
       \Bigr\}\,du \\
  &= \int_{t_n}^\infty
       \bigl(1 - F_n(x)\bigr)\,
       \sigma_\varepsilon\sqrt n\,dx \\
  &= \sigma_\varepsilon\sqrt n
     \int_{t_n}^\infty
       \bigl(1 - F_n(x)\bigr)\,dx.
\end{aligned}
\]
Substituting this into \eqref{eq:D-blanket-tail-Sn} gives
\[
\begin{aligned}
  \mathcal D^{\mathrm{blanket}}_{e^{\varepsilon},n,\mathcal R_{\mathrm{ref}},\gamma}(
    \mathcal{R}_{x_1} \| \mathcal{R}_{x_1^\prime})
  &= \frac{1}{n\gamma}\,
     \sigma_\varepsilon\sqrt n
     \int_{t_n}^\infty
       \bigl(1 - F_n(x)\bigr)\,dx \\
  &= \frac{\sigma_\varepsilon}{\sqrt n\,\gamma}\,
     \int_{t_n}^\infty
       \bigl(1 - F_n(x)\bigr)\,dx,
\end{aligned}
\]
which is \eqref{eq:blanket-transform-ii}. This completes the proof.
\end{proof}

We now begin the proof of Lemma~\ref{lem:blanket-asymptotics-general}.
We analyze the blanket divergence by deriving an asymptotic expansion based on central limit theorems (CLT). The general CLT-based results used in our analysis are stated below in a form tailored to our setting, with the level of abstraction
kept to the minimum necessary. In particular, we treat the continuous case
\textup{(Cont)} and the bounded case \textup{(Bound)} separately, using
a triangular-array Edgeworth expansion in the former and a triangular-array
moderate deviation theorem in the latter.

\begin{theorem}[Triangular-array Edgeworth expansion after Angst--Poly~\cite{angstWeakCramerCondition2017}]
\label{thm:edgeworth-triangular}
Let $(X_{i,n})_{1\le i\le n,\,n\ge1}$ be a triangular array of
real-valued random variables such that, for each fixed $n\ge1$,
the variables $X_{1,n},\dots,X_{n,n}$ are independent and identically
distributed (row-wise i.i.d.), while the common distribution may depend
on $n$. For each $n$ set
$
  S_n := \sum_{i=1}^n X_{i,n},
$
$
  B_n^2 := \mathrm{Var}(S_n),
$ and 
$
  F_n(x) := \Pr\left[\frac{S_n}{B_n}\le x\right].
$

Fix an integer $s\ge 3$. Assume:
  (i)  $\forall r\in \mathbb{N}, \sup_n \mathbb{E}[|X_{1,n}|^r]<\infty$.
  (ii) For all $n$, $\mathbb{E}[X_{1,n}]=0$ and $\inf_n \mathrm{Var}(X_{1,n})\geq C>0$.
  (iii) $X_{1,n}$ satisfies Cram\'er's condition. That is,
$
  \sup_n \limsup_{|t|\to\infty} \bigl|\mathbb{E}\big[e^{i t \cdot X_{1,n}}\big]\bigr| < 1.
$
For $r\ge3$ denote by
$
  \kappa_{r,n} := \frac{1}{B_n^r}\sum_{i=1}^n \kappa_r(X_{i,n})
$
the normalized cumulants of $S_n/B_n$, and let $P_{1,n},\dots,P_{s-2,n}$
be polynomials in $x$, whose coefficients are universal polynomials in
$\kappa_{3,n},\dots,\kappa_{s,n}$.

Then we have the Edgeworth expansion
\[
  \sup_{x\in\mathbb{R}}
  \biggl|
    F_n(x) - \Phi(x)
    - \sum_{j=1}^{s-2} P_{j,n}(x)\,\varphi(x)
  \biggr|
  = o\bigl(n^{-(s-2)/2}\bigr),
  \qquad n\to\infty.
\]
\end{theorem}

\begin{remark}
Theorem~\ref{thm:edgeworth-triangular} is obtained by specializing the
general triangular-array Edgeworth expansion of Angst and
Poly~\cite[Thm.~4.3]{angstWeakCramerCondition2017} to the
one-dimensional row-wise i.i.d.\ case. Their theorem covers much more
general triangular arrays (including non-identically distributed and
vector-valued arrays); in this paper we only need the resulting
triangular-array extension of the classical Cramér-type Edgeworth
expansion of Petrov~\cite[Ch.~VI, Thms.~3–4]{petrovSumsIndependentRandom1975}.
Compared with the hypotheses of~\cite[Thm.~4.3]{angstWeakCramerCondition2017},
we adopt stronger assumptions in the spirit of Petrov: we require the
existence of moments of all integer orders and replace the weak mean
Cramér condition by the classical Cramér condition for the
common distribution of $X_{1,n}$. When the distribution of $X_{1,n}$ does not depend on
$n$, Theorem~\ref{thm:edgeworth-triangular} reduces to Petrov's
one-dimensional Edgeworth expansion for sums of i.i.d.\ random
variables.
\end{remark}

\begin{theorem}[Triangular-array moderate deviations after Slastnikov~\cite{slastnikovLimitTheoremsModerate1979}]
\label{thm:slastnikov-triangular}
Let $(X_{i,n})_{1\le i\le n,\,n\ge1}$ be a triangular array of
real-valued random variables such that, for each fixed $n\ge1$,
the variables $X_{1,n},\dots,X_{n,n}$ are independent and identically
distributed (row-wise i.i.d.), while the common distribution may depend
on $n$. For each $n$ set
$
  S_n := \sum_{i=1}^n X_{i,n},
$
$
  B_n^2 := \mathrm{Var}(S_n),
$ and
$
  F_n(x) := \Pr\!\left[\frac{S_n}{B_n}\le x\right].
$
Assume:
  (i) $\mathbb{E}[X_{1,n}] = 0$ for all $n\ge1$;
  (ii) there exists a constant $\sigma_0>0$ such that
  $\mathrm{Var}(X_{1,n}) \ge \sigma_0^2$ for all $n\ge1$ (equivalently, $B_n^2= \Theta(n)$);
  (iii) the array is uniformly bounded: there exists $K<\infty$ such that
  $|X_{1,n}|\le K$ almost surely for all $n\ge1$.
Then for every fixed $c>0$ we have, as $n\to\infty$,
\[
  \sup_{0\le x\le c\sqrt{\log n}}
  \left|
    \frac{1-F_n(x)}{1-\Phi(x)} - 1
  \right|
  \;\longrightarrow\; 0.
\]
\end{theorem}

\begin{remark}
Theorem~\ref{thm:slastnikov-triangular} is a specialization of
Slastnikov's moderate deviation theorem for triangular arrays
\cite[Thm.~1.1]{slastnikovLimitTheoremsModerate1979}
to the row-wise i.i.d.\ case, under the uniform boundedness
assumption~(iii).
In Slastnikov's notation, assumption~(iii) implies his moment condition for any fixed $c>0$, and (i)–(iii) together yield
$B_n^2 = \mathrm{Var}(S_n) = \Theta(n)$.
Hence his conclusion
$
  1-F_n(x) \sim 1-\Phi(x)
$
holds uniformly in the moderate deviation region
$0 \le x \le c\sqrt{\log B_n}$ for any fixed $c>0$.
Since $\log B_n = \tfrac12 \log n + O(1)$, this is equivalent to
uniformity on $0 \le x \le c\sqrt{\log n}$ as stated in
Theorem~\ref{thm:slastnikov-triangular}.
\end{remark}

\begin{proof}


We first explain the structure of the proof.
The argument proceeds in three steps.

First, in order to work with bounded summands that simplify the analysis, we introduce a
mean-preserving truncation of the privacy amplification variable in the (Cont) case and define the corresponding truncated blanket divergence. 
We show that the truncation error is negligible compared to the main term; see Lemma~\ref{lem:truncation-blanket}.

Second, for the truncated variables we analyze the expectation appearing in
the transformation formula of Lemma~\ref{lem:blanket-transform}.  In the
continuous case (Cont) we apply a triangular-array Edgeworth expansion
(Theorem~\ref{thm:edgeworth-triangular}) to the normalized sums and derive
a precise asymptotic expression for the expectation; this is
summarized in Lemma~\ref{lem:inner-expectation}.  In the discrete bounded
case (Bound) we instead rely on Slastnikov's triangular-array moderate
deviation theorem (Theorem~\ref{thm:slastnikov-triangular}); see
Corollary~\ref{cor:blanket-asymptotics-discrete-Disc}.

Finally, we express the resulting asymptotic formula in terms of the
variance
\(
  \sigma^2 = \mathrm{Var}(l_0(Y))
\)
of the underlying privacy amplification variable and the parameter
\(
  \chi = \sqrt{\gamma}/\sigma.
\)
This is achieved by showing that
\(
  \sigma_{\varepsilon_n} \to \sqrt{\gamma}\,\sigma
\)
and that the Gaussian tail factor
\(
  \varphi(t_n)
\)
can be reparametrized as
\(
  \varphi(\chi(\mathrm e^{\varepsilon_n}-1)\sqrt{n})
\)
up to a relative $o(1)$ error.  Putting the three steps together yields
the claimed asymptotic expansion in Lemma~\ref{lem:blanket-asymptotics-general}.

\begin{definition}[Truncated blanket divergence]
  \footnote{%
  The truncation step is not essential for the validity of the asymptotic expansion.
  One could work directly with $l_{\varepsilon_n}(Y)$ by splitting expectations over
  the events $\{|l_{\varepsilon_n}(Y)|\le \tau_n\}$ and its complement, and controlling
  the tail contribution via higher moments and Markov's inequality.
  We introduce the truncated variables $X_{i,n}^{\tau_n}$ only for technical convenience:
  the uniform bound $|X_{i,n}^{\tau_n}|\le C \tau_n$ allows us to treat the error terms
  in the Edgeworth and Taylor expansions in a simpler and more uniform way.%
}
\label{def:truncated-blanket}
Let $(\varepsilon_n)_{n\ge 1}$ be a sequence of positive numbers and
$(\tau_n)_{n\ge 1}$ a sequence of truncation levels with $\tau_n>0$.
For each $n\ge1$, let $Y_1,Y_2,\dots$ be i.i.d.\ with distribution
$\mathcal R_{\mathrm{ref}}$ and define
\[
  X_{i,n} := l_{\varepsilon_n}(Y_i;x_1,x_1^\prime,\mathcal{R}_{\mathrm{ref}}), \qquad
  \mu_n := \mathbb E[X_{1,n}] = 1 - \mathrm e^{\varepsilon_n}.
\]
Let $B_1,B_2,\dots$ be i.i.d.\ $\mathrm{Bernoulli}(\gamma)$,
independent of $(Y_i)_{i\ge1}$, and set
\[
  Z_{i,n} := B_i X_{i,n}, \qquad i\ge1.
\]

We introduce the truncated random variable at level $\tau_n$ for $l_{\varepsilon_n}$ where
$Y\sim \mathcal R_{\mathrm{ref}}$ by mean preserving truncation:
\[
  \tilde l_{\varepsilon_n,\tau_n}(Y)
  := l_{\varepsilon_n}(Y)\,
     \mathbf 1\bigl\{|l_{\varepsilon_n}(Y)| \le \tau_n\bigr\},
\]
\[
  l_{\varepsilon_n,\tau_n}(Y)
  := \tilde l_{\varepsilon_n,\tau_n}(Y)
     - \mathbb E\bigl[\tilde l_{\varepsilon_n,\tau_n}(Y_1)\bigr]
     + \mu_n,
\]
and write
\[
  X_{i,n}^{\tau_n} := l_{\varepsilon_n,\tau_n}(Y_i), \qquad
  Z_{i,n}^{\tau_n} := B_i X_{i,n}^{\tau_n}, \qquad i\ge 1.
\]

For each $n$, define the partial sums
\[
  S_n := \sum_{i=1}^n Z_{i,n}, \qquad
  S_n^{\tau_n} := \sum_{i=1}^n Z_{i,n}^{\tau_n}.
\]
Using the Bernoulli–selector representation of the blanket divergence,
we write
\[
  \mathcal D^{\mathrm{blanket}}
  := \frac{1}{n\gamma}\,\mathbb E\bigl[(S_n)_+\bigr],
  \qquad
  \mathcal D^{\mathrm{tblanket}}
  := \frac{1}{n\gamma}\,\mathbb E\bigl[(S_n^{\tau_n})_+\bigr].
\]
We call $\mathcal D^{\mathrm{tblanket}}$ the \emph{truncated blanket divergence}
at level $\tau_n$.
\end{definition}

\begin{lemma}[Truncation error for the blanket divergence]
\label{lem:truncation-blanket}
Assume that the local randomizer $\mathcal R$ satisfies
Assumption~\ref{assump:regularity} with (Cont).
Then, for every integer $p\ge 2$ there exists a constant
$C<\infty$ such that
\begin{equation}
  \label{eq:truncation-error-basic}
  \bigl|\mathcal D^{\mathrm{blanket}}
        -\mathcal D^{\mathrm{tblanket}}\bigr|
  \;\le\;
  \frac{C}{\tau_n^{p-1}},
  \qquad n\ge 1.
\end{equation}
In particular, if $\tau_n = n^\alpha$ for some $0<\alpha<1/2$, then
\[
  \bigl|\mathcal D^{\mathrm{blanket}}
        -\mathcal D^{\mathrm{tblanket}}\bigr|
  = O\bigl(n^{-\alpha(p-1)}\bigr), n\to\infty.
\]
so by choosing $p$ sufficiently large the truncation error can be made
smaller than any prescribed polynomial order in $n$.

Moreover, the truncated random variables
$Z_{i,n}^{\tau_n} := B_i l_{\varepsilon_n,\tau_n}(Y_i)$ also satisfy
the conditions of Assumption~\ref{assump:regularity} (possibly with different constants
in that assumption) uniformly in $n$.
In addition, there exist constants $C<\infty$ and $n_0\in\mathbb N$ such that
\[
  |Z_{i,n}^{\tau_n}|
  \le C \tau_n \qquad \text{almost surely for all }n\ge n_0,\; i\ge 1.
\]
\end{lemma}

\begin{proof}
Write
\[
  \Delta_{i,n} := Z_{i,n} - Z_{i,n}^{\tau_n}
  = B_i\bigl(X_{i,n} - X_{i,n}^{\tau_n}\bigr), \qquad i\ge 1.
\]
By the definition of the blanket divergence and the truncated blanket divergence,
\[
  \mathcal D^{\mathrm{blanket}}
  - \mathcal D^{\mathrm{tblanket}}
  = \frac{1}{n\gamma}\,
    \mathbb E\bigl[(S_n)_+ - (S_n^{\tau_n})_+\bigr].
\]
Since $x\mapsto x_+$ is $1$-Lipschitz, we have
\[
  \bigl|(S_n)_+ - (S_n^{\tau_n})_+\bigr|
  \le \bigl|S_n - S_n^{\tau_n}\bigr|
  = \biggl|\sum_{i=1}^n \Delta_{i,n}\biggr|.
\]
Therefore
\begin{equation}
  \label{eq:truncation-step1-Z}
  \bigl|\mathcal D^{\mathrm{blanket}}
        - \mathcal D^{\mathrm{tblanket}}\bigr|
  \le \frac{1}{n\gamma}\,
       \mathbb E\biggl|\sum_{i=1}^n \Delta_{i,n}\biggr|.
\end{equation}

By the triangle inequality,
\[
  \mathbb E\biggl|\sum_{i=1}^n \Delta_{i,n}\biggr|
  \le \sum_{i=1}^n \mathbb E\bigl[|\Delta_{i,n}|\bigr]
  = n\,\mathbb E\bigl[|\Delta_{1,n}|\bigr].
\]
Substituting this into \eqref{eq:truncation-step1-Z} yields
\[
  \bigl|\mathcal D^{\mathrm{blanket}}
        - \mathcal D^{\mathrm{tblanket}}\bigr|
  \le \frac{1}{\gamma}\,\mathbb E\bigl[|\Delta_{1,n}|\bigr].
\]

It remains to estimate $\mathbb E[|\Delta_{1,n}|]$.
By the definition of the mean-preserving truncation,
\[
  X_{1,n}^{\tau_n}
  = X_{1,n}\mathbf 1\{|X_{1,n}|\le \tau_n\}
    - \mathbb E\bigl[X_{1,n}\mathbf 1\{|X_{1,n}|\le \tau_n\}\bigr]
    + \mu_n,
\]
and since $\mathbb E[X_{1,n}] = \mu_n$ we obtain, after rearranging,
\[
  X_{1,n}^{\tau_n}
  = X_{1,n}\mathbf 1\{|X_{1,n}|\le \tau_n\}
    + \mathbb E\bigl[X_{1,n}\mathbf 1\{|X_{1,n}|> \tau_n\}\bigr].
\]
Hence
\[
  X_{1,n} - X_{1,n}^{\tau_n}
  = X_{1,n}\mathbf 1\{|X_{1,n}|>\tau_n\}
    - \mathbb E\bigl[X_{1,n}\mathbf 1\{|X_{1,n}|>\tau_n\}\bigr].
\]
Using $|B_1|\le1$, we get
\[
  |\Delta_{1,n}|
  = |B_1|\;|X_{1,n} - X_{1,n}^{\tau_n}|
  \le |X_{1,n}|\mathbf 1\{|X_{1,n}|>\tau_n\}
       + \mathbb E\bigl[|X_{1,n}|\mathbf 1\{|X_{1,n}|>\tau_n\}\bigr].
\]
Taking expectations gives
\[
  \mathbb E\bigl[|\Delta_{1,n}|\bigr]
  \le 2\,\mathbb E\bigl[|X_{1,n}|\mathbf 1\{|X_{1,n}|>\tau_n\}\bigr].
\]

Let $p\ge 2$ be fixed. By Markov's inequality and the
uniform moment bound in Assumption~\ref{assump:regularity}, we have
\[
  \mathbb E\bigl[|X_{1,n}|\mathbf 1\{|X_{1,n}|>\tau_n\}\bigr]
  \le \frac{\mathbb E[|X_{1,n}|^p]}{\tau_n^{p-1}}
  \le \frac{C_p}{\tau_n^{p-1}},
\]
for some constant $C_p<\infty$ independent of $n$.
Combining the last three displays yields
\[
  \mathbb E\bigl[|\Delta_{1,n}|\bigr]
  \le \frac{2C_p}{\tau_n^{p-1}}.
\]
Recalling that
\(
  \bigl|\mathcal D^{\mathrm{blanket}}
        - \mathcal D^{\mathrm{tblanket}}\bigr|
  \le \gamma^{-1}\mathbb E[|\Delta_{1,n}|]
\),
we obtain \eqref{eq:truncation-error-basic} with
$C := 2C_p/\gamma$, which proves the first claim.

For the second claim, take $\tau_n = n^\alpha$ with $0<\alpha<1/2$. Then
\[
  \bigl|\mathcal D^{\mathrm{blanket}}
        - \mathcal D^{\mathrm{tblanket}}\bigr|
  \le C n^{-\alpha(p-1)}.
\]
Given any $K>0$, choose $p$ large enough so that
$\alpha(p-1) > K$. Then
\(
  n^K
  \bigl|\mathcal D^{\mathrm{blanket}}
        - \mathcal D^{\mathrm{tblanket}}\bigr|
  \to 0
\)
as $n\to\infty$, i.e., the truncation error decays faster than
the rate $n^{-K}$.

It remains to check that the truncated random variables
$Z_{i,n}^{\tau_n} := B_i l_{\varepsilon_n,\tau_n}(Y_i)$ inherit
Assumption~\ref{assump:regularity}.
Write $X_n := l_{\varepsilon_n}(Y_1)$ and note that
\[
  X_{1,n}^{\tau_n}
  = X_n \mathbf{1}\{|X_n|\le \tau_n\}
    + \mathbb{E}[X_n \mathbf{1}\{|X_n|>\tau_n\}],
\]
so
\[
  |X_{1,n}^{\tau_n}|
  \le |X_n| + \mathbb{E}|X_n|.
\]
Thus all moment bounds required in
Assumption~\ref{assump:regularity} for $X_n$ carry over
to $X_{1,n}^{\tau_n}$ (up to different constants), and a fortiori to
$Z_{1,n}^{\tau_n} = B_1 X_{1,n}^{\tau_n}$ since $|B_1|\le1$.

For the continuity part of Assumption~\ref{assump:regularity}, the distribution of
$X_n$ has a common absolutely continuous component on a fixed bounded
interval $I=[a,b]$, with mass and density bounded away from zero and infinity,
uniformly in $n$.
Since $Z_{1,n} = B_1 X_{1,n}$ with $B_1$ independent of $X_{1,n}$,
the distribution of $Z_{1,n}$ is a mixture of an atom at $0$ (coming from $B_1=0$)
and the distribution of $X_{1,n}$ (coming from $B_1=1$), so $Z_{1,n}$ and
$Z_{1,n}^{\tau_n}$ inherit the same absolutely continuous component on $I$
(up to a constant factor and, for $Z_{1,n}^{\tau_n}$, up to a translation).

Finally, by the above display and the uniform moment bounds, there exist
constants $C<\infty$ and $n_0\in\mathbb N$ such that
\[
  |X_{1,n}^{\tau_n}|\le C \tau_n \quad\text{for all }n\ge n_0,
\]
and therefore $|Z_{1,n}^{\tau_n}| = |B_1||X_{1,n}^{\tau_n}|\le C \tau_n$ as well.
The same bound holds for all $i\ge1$ by identical distribution, which
completes the proof.
\end{proof}

We now proceed to the second step using the truncated blanket divergence.

\begin{lemma}[Evaluation of the expectation]
\label{lem:inner-expectation}
Assume the setting of Assumption~\ref{assump:regularity} and
Definition~\ref{def:truncated-blanket}.  Let $(\varepsilon_n)_{n\ge1}$
be a sequence with $\varepsilon_n>0$, $\varepsilon_n\to0$,
$\varepsilon_n=\omega(n^{-1/2})$, and
$\varepsilon_n = O\!\bigl(\sqrt{(\log n)/n}\bigr)$.

For each $n$ let $Y_1,Y_2,\dots$ be i.i.d.\ with distribution
$\mathcal R_{\mathrm{ref}}$, let $B_1,B_2,\dots$ be i.i.d.\
$\mathrm{Bernoulli}(\gamma)$ independent of $(Y_i)_{i\ge1}$, and define
the truncated variables
\[
  Z_{i,n}^{\tau_n} := B_i\,l_{\varepsilon_n,\tau_n}(Y_i),\qquad i\ge1.
\]
Write
\[
  Z := Z_{1,n}^{\tau_n},\qquad
  \mu_{n} := \mathbb E[Z],\qquad
  \sigma_{\varepsilon_n}^2 := \mathrm{Var}(Z)\in(0,\infty).
\]
For $n\ge2$ define
\[
  S_{n-1} := \sum_{i=2}^n Z_{i,n}^{\tau_n},\qquad
  T_{n-1} := \frac{S_{n-1} - (n-1)\mu_n}{\sigma_{\varepsilon_n}\sqrt{n-1}},
\]
and let $F_{n-1}$ denote the distribution function of $T_{n-1}$.
Set
\[
  t_n := -\frac{\mu_n\sqrt{n-1}}{\sigma_{\varepsilon_n}},
  \qquad
  \eta_n(z) := \frac{z}{\sigma_{\varepsilon_n}\sqrt{n-1}}.
\]

Then there exists a sequence $(r_n)_{n\ge1}$ with $r_n\in\mathbb R$
such that
\begin{equation}
  \label{eq:inner-expectation-asymptotic}
  \mathbb E\Bigl[
    Z\bigl(1-F_{n-1}\bigl(t_n-\eta_n(Z)\bigr)\bigr)
  \Bigr]
  =
  \frac{\varphi(t_n)}{\sqrt{n-1}}\,
  \frac{\sigma_{\varepsilon_n}}{t_n^2}\,
  \bigl(1 + r_n\bigr),
\end{equation}
and $r_n\to0$ as $n\to\infty$.

Moreover, the sequence $(r_n)$ depends only on the constants appearing
in Assumption~\ref{assump:regularity}.
\end{lemma}

\begin{proof}
Fix $n\ge2$. Throughout the proof we suppress the dependence on $n$ in
the notation whenever this causes no confusion, and write
$Z_i := Z_{i,n}^{\tau_n}$.

\medskip\noindent\textbf{Step 1: Edgeworth expansion for the truncated
array.}
For each $n$ the variables $Z_1,Z_2,\dots$ are i.i.d.\ with
\[
  \mu_n := \mathbb E[Z_1],\qquad
  \sigma_{\varepsilon_n}^2 := \mathrm{Var}(Z_1).
\]
By Assumption~\ref{assump:regularity} and the truncation
Lemma~\ref{lem:truncation-blanket}, the triangular array
$(Z_{i,n}^{\tau_n})_{1\le i\le n,\,n\ge1}$ satisfies the hypotheses of
Theorem~\ref{thm:edgeworth-triangular} (uniformly in $n$).
In particular, for any fixed integer $s\ge3$ there exist
polynomials $P_{1,k},\dots,P_{s-2,k}$ and constants
$C_s<\infty$ and $k_0\in\mathbb N$ such that, for all
$k\ge k_0$, the distribution $F_k$ of
\[
  T_k
  := \frac{S_k-k\mu_n}{\sigma_{\varepsilon_n}\sqrt{k}},
  \qquad
  S_k:=\sum_{i=1}^k Z_i,
\]
satisfies
\[
  \sup_{u\in\mathbb R}
  \biggl|
    F_k(u) - \Phi(u)
    - \sum_{j=1}^{s-2}k^{-j/2}P_{j,k}(u)\varphi(u)
  \biggr|
  \le C_s k^{-(s-2)/2}.
\]

We apply this with $k=n-1$ (for $n$ large enough so that $n-1\ge k_0$).
Thus,
\begin{equation}
  \label{eq:edgeworth-remainder-n}
  \sup_{u\in\mathbb R}
  \biggl|
    F_{n-1}(u) - \Phi(u)
    - \sum_{j=1}^{s-2}(n-1)^{-j/2}P_{j,n-1}(u)\varphi(u)
  \biggr|
  \le C_s (n-1)^{-(s-2)/2}.
\end{equation}
The polynomials $P_{j,n-1}$ and the constant $C_s$ may be chosen
independently of $n$.

In the sequel we use the shorthand
\[
  t_n := -\frac{\mu_n\sqrt{n-1}}{\sigma_{\varepsilon_n}},
  \qquad
  \eta_n(z) := \frac{z}{\sigma_{\varepsilon_n}\sqrt{n-1}}.
\]

\medskip\noindent\textbf{Step 2: Decomposition of the inner expectation.}
We start from
\[
  \mathbb E\Bigl[
    Z\bigl(1-F_{n-1}(t_n-\eta_n(Z))\bigr)
  \Bigr].
\]
Using the expansion \eqref{eq:edgeworth-remainder-n} at the point
$u=t_n-\eta_n(Z)$, we write
\[
  1-F_{n-1}(u)
  = (1-\Phi(u))
    - \sum_{j=1}^{s-2}(n-1)^{-j/2}P_{j,n-1}(u)\varphi(u)
    + R_{n-1,s}(u),
\]
where the remainder $R_{n-1,s}$ satisfies
\(
  |R_{n-1,s}(u)|
  \le C_s (n-1)^{-(s-2)/2}
\)
for all $u\in\mathbb R$. Substituting $u=t_n-\eta_n(Z)$ and taking the
expectation with respect to $Z$ yields
\begin{align*}
  \mathbb E\Bigl[
    Z\bigl(1-F_{n-1}(t_n-\eta_n(Z))\bigr)
  \Bigr]
  &= A_n - \sum_{j=1}^{s-2}B_{n,j} + E_n,
\end{align*}
where we set
\begin{align*}
  A_n
  &:= \mathbb E\Bigl[
        Z\bigl(1-\Phi(t_n-\eta_n(Z))\bigr)
      \Bigr],\\
  B_{n,j}
  &:= (n-1)^{-j/2}\,
      \mathbb E\Bigl[
        Z P_{j,n-1}\bigl(t_n-\eta_n(Z)\bigr)
          \varphi\bigl(t_n-\eta_n(Z)\bigr)
      \Bigr],\\
  E_n
  &:= \mathbb E\Bigl[
        Z R_{n-1,s}\bigl(t_n-\eta_n(Z)\bigr)
      \Bigr].
\end{align*}

\medskip\noindent\textbf{Step 3: Bound on the Edgeworth remainder $E_n$.}
By Lemma~\ref{lem:truncation-blanket} there exist constants
$C,h_0>0$ such that, for all $n$ large enough,
\(
  |Z|\le C \tau_n
\)
almost surely. Using the bound on $R_{n-1,s}$ and this truncation, we
obtain
\[
  |E_n|
  \le \mathbb E\Bigl[
    |Z|\,
    \bigl|R_{n-1,s}\bigl(t_n-\eta_n(Z)\bigr)\bigr|
  \Bigr]
  \le C \tau_n C_s (n-1)^{-(s-2)/2}
  = O\bigl(\tau_n (n-1)^{-(s-2)/2}\bigr).
\]

\medskip\noindent\textbf{Step 4: Evaluation of $A_n$.}
Define $g(z):=1-\Phi(z)$. Then
\[
  g'(z) = -\varphi(z),\qquad
  g''(z) = z\varphi(z).
\]
By Taylor's theorem with Lagrange remainder, for each $z$ there exists
$\xi=\xi(z)$ between $t_n$ and $t_n-\eta_n(z)$ such that
\begin{equation}
  \label{eq:taylor-g-n}
  1-\Phi\bigl(t_n-\eta_n(z)\bigr)
  = g(t_n) - g'(t_n)\eta_n(z)
    + \frac12 g''(\xi)\eta_n(z)^2.
\end{equation}
We now control $g''(\xi)$. By the explicit formula for the Gaussian
density,
\[
  \frac{\varphi(\xi)}{\varphi(t_n)}
  = \exp\Bigl(-\frac{\xi^2-t_n^2}{2}\Bigr)
  = \exp\Bigl(
      t_n\theta\eta_n(z)-\tfrac12\theta^2\eta_n(z)^2
    \Bigr),
\]
for some $\theta\in[0,1]$. Under our scaling assumptions,
$\varepsilon_n = O(\sqrt{(\log n)/n})$ and
$\varepsilon_n=\omega(n^{-1/2})$, and using that $\mu_n$ is of order
$\varepsilon_n$ and $\sigma_{\varepsilon_n}$ is bounded away from $0$
and $\infty$ (by Assumption~\ref{assump:regularity} and the truncation
lemma), we have
\[
  t_n
  = -\frac{\mu_n\sqrt{n-1}}{\sigma_{\varepsilon_n}}
  = \Theta\bigl(\varepsilon_n\sqrt{n}\bigr),
\]
so $t_n\to\infty$ and $t_n^2 = O(\log n)$.

Moreover, choosing $\tau_n = n^\alpha$ for some fixed
$0<\alpha<1/2$, we have $|z|\le C \tau_n$ and hence
\[
  |\eta_n(z)|
  = O\Bigl(\frac{\tau_n}{\sqrt{n}}\Bigr)
  = O\bigl(n^{\alpha-1/2}\bigr)
  = o(1).
\]
Thus $|t_n\eta_n(z)| = o(1)$ and therefore
$\varphi(\xi)\le C_\xi\varphi(t_n)$ for some constant $C_\xi$
independent of $n$. Furthermore,
$|\xi|\le |t_n|+|\eta_n(z)|\le C(1+t_n)$, so altogether
\[
  |g''(\xi)| = |\xi|\varphi(\xi)
  \le C'(1+t_n)\varphi(t_n).
\]

Substituting \eqref{eq:taylor-g-n} into the definition of $A_n$ and
using $\mathbb E[Z]=\mu_n$ gives
\begin{align*}
  A_n
  &= \mathbb E\Bigl[
       Z\bigl(1-\Phi(t_n-\eta_n(Z))\bigr)
     \Bigr]\\
  &= \mu_n(1-\Phi(t_n))
     +\varphi(t_n)\mathbb E[Z\eta_n(Z)]
     +\frac12\mathbb E\bigl[Z g''(\xi)\eta_n(Z)^2\bigr].
\end{align*}
Note that
\[
  \mathbb E[Z\eta_n(Z)]
  = \frac{1}{\sigma_{\varepsilon_n}\sqrt{n-1}}\mathbb E[Z^2]
  = \frac{\sigma_{\varepsilon_n}^2+\mu_n^2}
         {\sigma_{\varepsilon_n}\sqrt{n-1}}.
\]
Using the bound on $g''(\xi)$ and the uniform boundedness of the third
moment, we obtain
\[
  \Bigl|
    \mathbb E\bigl[Z g''(\xi)\eta_n(Z)^2\bigr]
  \Bigr|
  \le C''\frac{\varphi(t_n)(1+t_n)}{n-1}.
\]
Therefore
\begin{equation}
  \label{eq:A-n-expansion}
  A_n
  = \mu_n(1-\Phi(t_n))
    +\frac{\sigma_{\varepsilon_n}^2+\mu_n^2}
          {\sigma_{\varepsilon_n}\sqrt{n-1}}\varphi(t_n)
    +O\Bigl(\frac{\varphi(t_n)(1+t_n)}{n-1}\Bigr).
\end{equation}

\medskip\noindent\textbf{Step 5: Evaluation of $B_{n,1}$.}
In one dimension, the first Edgeworth polynomial has the explicit form
\[
  P_{1,n-1}(u)
  = \frac{\kappa_{3,n-1}}{6}(u^2-1),
\]
where $\kappa_{3,n-1}$ is the normalized third cumulant of $T_{n-1}$.
Since the summands are i.i.d., we have
\[
  \kappa_{3,n-1}
  = \frac{\kappa_3(Z_1)}{\sigma_{\varepsilon_n}^3\sqrt{n-1}}
  =: \frac{\kappa_3}{\sigma_{\varepsilon_n}^3\sqrt{n-1}},
\]
where $\kappa_3$ is the third cumulant of $Z_1$, uniformly bounded in
$n$ by the moment assumptions. Hence
\[
  B_{n,1}
  = \frac{\kappa_3}{6\sigma_{\varepsilon_n}^3\sqrt{n-1}}\,
    \mathbb E\Bigl[
      Z\varphi(u)(u^2-1)
    \Bigr],
  \qquad u:=t_n-\eta_n(Z).
\]
Define $D(u):=\varphi(u)(u^2-1)$. A direct computation gives
\[
  D'(u) = -u\varphi(u)(u^2-3),
  \qquad
  D''(u) = \varphi(u)(u^4-6u^2+3).
\]
By Taylor's theorem,
\[
  D(u)
  = D(t_n) + D'(t_n)\eta_n(Z) + \tfrac12 D''(\xi')\eta_n(Z)^2
\]
for some $\xi'=\xi'(Z)$ between $t_n$ and $u$. Using the same argument
as for $g''(\xi)$, we obtain
\(
  |D''(\xi')|
  \le C_D\varphi(t_n)(1+t_n^4)
\)
for all large $n$. Therefore
\begin{align*}
  D(u)
  &= \varphi(t_n)(t_n^2-1)
    +\varphi(t_n)t_n(t_n^2-3)\eta_n(Z)
    +O\bigl(\varphi(t_n)t_n^4\eta_n(Z)^2\bigr),
\end{align*}
and hence
\begin{align}
  B_{n,1}
  &= \frac{\kappa_3}{6\sigma_{\varepsilon_n}^3\sqrt{n-1}}
     \mathbb E\bigl[Z D(u)\bigr]\notag\\
  &= \frac{\kappa_3\mu_n}{6\sigma_{\varepsilon_n}^3\sqrt{n-1}}
     \varphi(t_n)(t_n^2-1)
   + \frac{\kappa_3}{6\sigma_{\varepsilon_n}^3\sqrt{n-1}}
     \varphi(t_n)t_n(t_n^2-3)\mathbb E[Z\eta_n(Z)]\notag\\
  &\quad
   + O\Bigl(
       \frac{\varphi(t_n)t_n^4}{\sigma_{\varepsilon_n}^2(n-1)}
       \mathbb E[Z^3]
     \Bigr).\label{eq:B-n-1-raw}
\end{align}
Using $\mathbb E[Z\eta_n(Z)]$ as above and bounded third moment, we
obtain
\begin{equation}
  \label{eq:B-n-1-expansion}
  B_{n,1}
  = \frac{\kappa_3\mu_n}{6\sigma_{\varepsilon_n}^3\sqrt{n-1}}
      \varphi(t_n)(t_n^2-1)
    + \frac{\kappa_3(\sigma_{\varepsilon_n}^2+\mu_n^2)}
           {6\sigma_{\varepsilon_n}^4}
      \varphi(t_n)t_n(t_n^2-3)\frac{1}{n-1}
    + O\Bigl(
        \frac{\varphi(t_n)(1+t_n^4)}{(n-1)^{3/2}}
      \Bigr).
\end{equation}

\medskip\noindent\textbf{Step 6: Bounds for $B_{n,j}$ with $j\ge2$.}
By standard properties of Edgeworth polynomials
(cf.~Theorem~\ref{thm:edgeworth-triangular}), for each $j\ge2$ there
exist constants $C_j<\infty$ and integers $k_j\le3j$ such that
\[
  |P_{j,n-1}(u)| \le C_j\bigl(1+|u|^{k_j}\bigr),
  \qquad u\in\mathbb R.
\]
As above, we have $t_n^2=O(\log n)$ and $|\eta_n(Z)|=O(\tau_n/\sqrt{n})=o(1)$, so
$|u|=|t_n-\eta_n(Z)|\le C(1+t_n)$ and
$\varphi(u)\le C'\varphi(t_n)$ for all large $n$. Since $|Z|\le C\tau_n$,
we obtain
\begin{align}
  |B_{n,j}|
  &\le (n-1)^{-j/2}\,
      \mathbb E\bigl[
        |Z||P_{j,n-1}(u)|\varphi(u)
      \bigr]\notag\\
  &\le C_j''(n-1)^{-j/2}\,\varphi(t_n)\,(1+t_n^{k_j}),
  \label{eq:B-n-j-bound}
\end{align}
uniformly in $n$ for $2\le j\le s-2$.

\medskip\noindent\textbf{Step 7: Extraction of the main term.}
Combining \eqref{eq:A-n-expansion}, \eqref{eq:B-n-1-expansion},
\eqref{eq:B-n-j-bound}, and the bound on $E_n$, we obtain
\begin{align}
  \mathbb E\Bigl[
    Z\bigl(1-F_{n-1}(t_n-\eta_n(Z))\bigr)
  \Bigr]
  &= A_n - B_{n,1} - \sum_{j=2}^{s-2}B_{n,j} + E_n\notag\\
  &= \mu_n(1-\Phi(t_n))
     +\frac{\sigma_{\varepsilon_n}^2+\mu_n^2}
            {\sigma_{\varepsilon_n}\sqrt{n-1}}\varphi(t_n)\notag\\
  &\quad
     - \frac{\kappa_3\mu_n}{6\sigma_{\varepsilon_n}^3\sqrt{n-1}}
         \varphi(t_n)(t_n^2-1)\notag\\
  &\quad
     + O\Bigl(
         \frac{\varphi(t_n)(1+t_n^4)}{n-1}
       \Bigr)
     + O\Bigl(
         \sum_{j=2}^{s-2}
           (n-1)^{-j/2}\varphi(t_n)t_n^{k_j}
       \Bigr)\notag\\
  &\quad
     + O\bigl(\tau_n (n-1)^{-(s-2)/2}\bigr).
     \label{eq:inner-expansion-pre-n}
\end{align}

We now use the Mills ratio expansion
\[
  1-\Phi(t_n)
  = \frac{\varphi(t_n)}{t_n}
    \Bigl(1-\frac{1}{t_n^2}
           +O\Bigl(\frac{1}{t_n^4}\Bigr)\Bigr),
  \qquad t_n\to\infty.
\]
Under our scaling assumptions,
\[
  t_n
  = -\frac{\mu_n\sqrt{n-1}}{\sigma_{\varepsilon_n}}
  = \Theta\bigl(\varepsilon_n\sqrt{n}\bigr),
\]
so $t_n\to\infty$ and $t_n^2=O(\log n)$.

Substituting the Mills ratio into the first term of
\eqref{eq:inner-expansion-pre-n} and using the identity
$t_n=-\mu_n\sqrt{n-1}/\sigma_{\varepsilon_n}$, a straightforward
algebraic simplification shows that the terms of order $O(\varphi(t_n)/\sqrt{n})$ cancel out, and the leading contribution from $A_n$ is
\[
  \frac{\varphi(t_n)}{\sqrt{n-1}}\,
  \frac{\sigma_{\varepsilon_n}}{t_n^2},
\]
while the remaining terms in \eqref{eq:inner-expansion-pre-n} are of
smaller order.

More precisely, we can write
\begin{equation}
  \label{eq:inner-main-plus-remainder-n}
  \mathbb E\Bigl[
    Z\bigl(1-F_{n-1}(t_n-\eta_n(Z))\bigr)
  \Bigr]
  = \frac{\varphi(t_n)}{\sqrt{n-1}}\,
    \frac{\sigma_{\varepsilon_n}}{t_n^2}\,
    \bigl(1+r_n\bigr),
\end{equation}
where the error term $r_n$ collects:
\begin{itemize}
  \item[(a)] the higher-order terms in the Mills ratio expansion
    (of order $O(t_n^{-2})$);
  \item[(b)] the $O(\cdot)$-terms in \eqref{eq:A-n-expansion}
    and \eqref{eq:B-n-1-expansion};
  \item[(c)] the sum over $j\ge2$ in \eqref{eq:B-n-j-bound};
  \item[(d)] the remainder $E_n$.
\end{itemize}

Using $t_n^2=O(\log n)$, $\tau_n=n^\alpha$ with $0<\alpha<1/2$, and
choosing $s$ sufficiently large in the Edgeworth expansion, we deduce
from the bounds in \eqref{eq:inner-expansion-pre-n},
\eqref{eq:B-n-j-bound}, and the explicit form of the Mills ratio
remainder that $r_n\to0$ as $n\to\infty$. All constants involved depend
only on the constants in Assumption~\ref{assump:regularity}, and the
preceding estimates are uniform over the distinct pair $(x_1,x_1')$
and the admissible choices of $\mathcal R_{\mathrm{ref}}$.

Thus \eqref{eq:inner-main-plus-remainder-n} proves
\eqref{eq:inner-expectation-asymptotic}, completing the proof.
\end{proof}

\paragraph{Final replacement of $\sigma_{\varepsilon_n}$ by $\sigma$.}
We now relate the variance $\sigma_{\varepsilon_n}^2 := \mathrm{Var}(Z_{1,n}^{\tau_n})$ of the truncated variable appearing in Lemma~\ref{lem:inner-expectation} to the limiting variance $\sigma^2 = \mathrm{Var}(l_0(Y))$.

First, let $\sigma_{\mathrm{untr}}^2(\varepsilon) := \mathrm{Var}(B_1 l_{\varepsilon}(Y_1))$ denote the variance of the original \emph{untruncated} privacy amplification variable. From the explicit integral representation
\[
  \mathbb E\bigl[l_{\varepsilon}(Y)^2\bigr]
  = \int \frac{\bigl(\mathcal R_{x_1}(y) - \mathrm e^{\varepsilon}\mathcal R_{x_1^\prime}(y)\bigr)^2}
              {\mathcal R_{\mathrm{ref}}(y)}\,dy,
\]
we observe that $\varepsilon \mapsto \sigma_{\mathrm{untr}}^2(\varepsilon)$ is continuously differentiable. Specifically,
\[
  \sigma_{\mathrm{untr}}^2(\varepsilon)
  = \gamma\bigl(a_1 + a_2 \mathrm e^{2\varepsilon} - 2 a_3 \mathrm e^{\varepsilon}\bigr)
    - \gamma^2(1-\mathrm e^{\varepsilon})^2,
\]
where $a_1, a_2, a_3$ depend only on the reference distributions and they are uniformly bounded from Assumption~\ref{assump:regularity}.
Differentiation shows that the derivative is uniformly bounded for $\varepsilon$ near $0$.
Since $\sigma_{\mathrm{untr}}^2(0) = \gamma \sigma^2$, the Lipschitz continuity implies
\[
  \bigl|\sigma_{\mathrm{untr}}(\varepsilon_n) - \sqrt{\gamma}\,\sigma\bigr|
  \le C |\varepsilon_n|,
\]
for some constant $C<\infty$.

Second, we compare the truncated variance $\sigma_{\varepsilon_n}^2$ with the untruncated variance $\sigma_{\mathrm{untr}}^2(\varepsilon_n)$.
Recall that the truncated variable $Z_{1,n}^{\tau_n}$ is obtained by a mean-preserving truncation at level $\tau_n = n^\alpha$.
By the Cauchy--Schwarz inequality and the moment bounds from Assumption~\ref{assump:regularity}, the difference is dominated by the tail contribution:
\[
  \bigl|\sigma_{\varepsilon_n}^2 - \sigma_{\mathrm{untr}}^2(\varepsilon_n)\bigr|
  \le \mathbb E\bigl[Z_{1,n}^2 \mathbf 1\{|Z_{1,n}| > \tau_n\}\bigr]
  \le \frac{C_p}{\tau_n^{p-2}}.
\]
By choosing $p$ sufficiently large, this truncation error decays faster than any polynomial in $n$, and in particular is $o(\varepsilon_n)$.
Combining this with the Lipschitz bound yields
\begin{equation}
  \label{eq:sigma-convergence-final}
  \sigma_{\varepsilon_n}
  = \sqrt{\gamma}\,\sigma + O(\varepsilon_n).
\end{equation}

Finally, we simplify the Gaussian factor $\varphi(t_n)$.
Define the parameter $\chi := \sqrt{\gamma}/\sigma$.
Using the definition $t_n = -\mu_n\sqrt{n-1}/\sigma_{\varepsilon_n}$ and the expansion \eqref{eq:sigma-convergence-final}, we write
\[
  t_n
  = -\frac{\mu_n\sqrt{n-1}}{\sqrt{\gamma}\,\sigma(1+O(\varepsilon_n))}
  = -\frac{\mu_n\sqrt{n-1}}{\sqrt{\gamma}\,\sigma} + r_n,
\]
where the error term $r_n$ satisfies
$r_n = O(\mu_n\sqrt{n-1}\cdot \varepsilon_n) = O(\sqrt{n} \varepsilon_n^2)$.
Since we assume $\varepsilon_n = O(\sqrt{(\log n)/n})$, we have 
\[
t_n = -\chi(\mathrm e^{\varepsilon_n}-1)\sqrt{n-1}(1+r_n) = -\chi(\mathrm e^{\varepsilon_n}-1)\sqrt{n-1}(1+o(1)).
\]

A Taylor expansion of the Gaussian density $\varphi(t)$ around $\bar t_n := -\mu_n\sqrt{n-1}/(\sqrt{\gamma}\sigma)$ gives
\[
  \varphi(t_n)
  = \varphi(\bar t_n)\bigl(1 + o(1)\bigr)
  = \varphi\bigl(\chi(\mathrm e^{\varepsilon_n}-1)\sqrt{n-1}\bigr)\bigl(1+o(1)\bigr).
\]

\medskip

By continuity of $x\mapsto x^3$ on
$(0,\infty)$ we obtain
\begin{equation}
  \label{eq:sigma-cube-Z}
  \sigma_{\varepsilon_n}^3
  = \sigma_0^3\bigl(1+o(1)\bigr)
  = \bigl(\sqrt{\gamma}\,\sigma\bigr)^3\bigl(1+o(1)\bigr)
  = \gamma^{3/2}\sigma^3\bigl(1+o(1)\bigr).
\end{equation}

\medskip

From the Edgeworth-based analysis with $Z_{i,n}^{\tau_n}$ (i.e., \eqref{eq:inner-main-plus-remainder-n}) we previously
obtained an expansion of the form
\[
  \mathcal D^{\mathrm{blanket}}
  = \frac{1}{\gamma}\frac{\varphi(t_n)}{\sqrt{n-1}}\,
    \frac{\sigma_{\varepsilon_n}}{t_n^2}\,
    \bigl(1+o(1)\bigr),
\]
where the factor $1/\gamma$ comes from the transformation
$\mathcal D^{\mathrm{blanket}}
 = \gamma^{-1}\mathbb E[Z_1\mathbf 1\{T_{n-1}>-Z_1\}]$.
Using \eqref{eq:sigma-convergence-final} and \eqref{eq:sigma-cube-Z}, and recalling
$\mu_n = \gamma(1-{\rm e}^{\varepsilon_n})$, we obtain
\[
  \mathcal D^{\mathrm{blanket}}
  =
  \varphi\bigl(\chi(\mathrm e^{\varepsilon_n}-1)\sqrt{n-1}\bigr)\,
  \left(
    \frac{1}{\chi^3}\,
    \frac{1}{\bigl(\mathrm e^{\varepsilon_n}-1\bigr)^2 (n-1)^{3/2}}
  \right)\bigl(1+o(1)\bigr).
\]
Finally, we may replace $n-1$ by $n$ everywhere in the display, at the
cost of modifying the $o(1)$ term, which yields the desired asymptotic
expression for the blanket divergence.

\qedhere

Next, we state the proof for the (Bound) case in the similar way as above.

\begin{corollary}[Discrete bounded case]
\label{cor:blanket-asymptotics-discrete-Disc}
Under the assumptions of Lemma~\ref{lem:blanket-asymptotics-general},
the same conclusion holds if Assumption~\ref{assump:regularity} is
replaced by its bounded version~\textup{(Bound)}.
\end{corollary}

\begin{proof}
We use the integral representation (i.e., (ii) of Lemma~\ref{lem:blanket-transform})\footnote{We deliberately work with the integral representation (ii) in
Lemma~\ref{lem:blanket-transform} rather than the linear form (i).
In the discrete bounded case, starting from (i) would lead to terms of
the form $\mathbb{E}[Z_1(1-F_{n-1}(u))\rho_n (u)]$, and controlling the
contribution of the Slastnikov remainder $\rho_n (u)$ relative to the main term
would require an explicit convergence rate for the moderate deviation
error, which is not available in the version we use.  By contrast, the
integral representation keeps the factor $(1+\rho_n(x))$ inside the
same tail integral and allows us to obtain a relative $o(1)$ error
without an explicit rate.  Conversely, in the continuous case we avoid
representation (ii), because our Edgeworth-based proof works under
moment assumptions only and does not assume the existence of
exponential moments, which are convenient for the Chernoff-type tail
bounds needed in the discrete bounded analysis.} together with Bernoulli selectors.
For each $n$ let $Y_1,\dots,Y_n$ be i.i.d.\ with distribution
$\mathcal R_{\mathrm{ref}}$, let $B_1,\dots,B_n$ be i.i.d.\
$\mathrm{Bernoulli}(\gamma)$ independent of $(Y_i)_{i\ge1}$.
That is,
  \begin{equation}
    \label{eq:D-blanket-Tn}
    \begin{aligned}
    \mathcal D^{\mathrm{blanket}}_{e^{{\varepsilon_n}},n,\mathcal R_{\mathrm{ref}},\gamma}(
      \mathcal{R}_{x_1} \| \mathcal{R}_{x_1^\prime})
    = \frac{\sigma_{\varepsilon_n}}{\sqrt{n}\gamma}\,
       \int_{t_n}^\infty \bigl(1-F_n(x)\bigr)\,dx.
    \end{aligned}
  \end{equation}
By assumption (Bound) and the definition of $Z_i$ there exists $K<\infty$
such that $|Z_i|\le K$ almost surely, uniformly in $n$ and in the choice
of $(x_1,x_1',\mathcal R_{\mathrm{ref}})$.

Recall
\[
  \mu_n := \mathbb E[Z_1]
        = \gamma\,\mathbb E[X_1]
        = \gamma(1-{\rm e}^{\varepsilon_n}),
  \qquad
  \sigma_{\varepsilon_n}^2 := \mathrm{Var}(Z_1),
\]
and 
\[
  T_n
  := \frac{S_n-n\mu_n}{\sigma_{\varepsilon_n}\sqrt{n}},
  \qquad
  F_n(x) := \Pr\left[T_n\le x\right]
\]

\medskip\noindent
\textbf{Step 1: Moderate deviations for $T_n$.}
The array $(Z_i)_{1\le i\le n,\,n\ge1}$ is row-wise i.i.d., uniformly
bounded, and satisfies
\(
  \mathbb E[Z_1]=\mu_n,\;
  \mathrm{Var}(Z_1)=\sigma_{\varepsilon_n}^2\ge C>0
\)
for all $n$, where $C$ depends only on the constants in
Assumption~\ref{assump:regularity}. Hence Slastnikov's moderate
deviation theorem for triangular arrays
(Theorem~\ref{thm:slastnikov-triangular}) applies to the standardized
sums $T_n$.

In particular, for every fixed $c>0$ there exists a sequence
$\Delta_n\downarrow0$ such that
\begin{equation}
  \label{eq:slastnikov-Z}
  \sup_{0\le x\le c\sqrt{\log n}}
  \left|
    \frac{\Pr\left[T_n>x\right]}{1-\Phi(x)} - 1
  \right|
  \le \Delta_n,
  \qquad n\to\infty.
\end{equation}
Equivalently, there exists a function $\rho_n:[0,\infty)\to\mathbb R$
such that
\[
  \Pr\left[T_n>x\right] = (1-\Phi(x))\bigl(1+\rho_n(x)\bigr),
\]
and
\[
  \sup_{0\le x\le c\sqrt{\log n}}|\rho_n(x)|
  \le \Delta_n,\qquad \Delta_n\xrightarrow[n\to\infty]{}0.
\]

Under the scaling assumptions
\(
  \varepsilon_n=\omega(n^{-1/2}),
  \varepsilon_n=O(\sqrt{(\log n)/n})
\),
we have $\mu_n = \gamma(1-{\rm e}^{\varepsilon_n}) = O(\varepsilon_n)$
and $\sigma_{\varepsilon_n}$ bounded away from $0$ and $\infty$,
so
\[
  t_n
  = -\frac{\mu_n\sqrt{n}}{\sigma_{\varepsilon_n}}
  = \Theta\bigl(\varepsilon_n\sqrt{n}\bigr)
  = O\bigl(\sqrt{\log n}\bigr),
\]
and in particular $t_n\to\infty$ and, for suitable $c>0$ and all large
$n$,
\[
  t_n \le c\sqrt{\log n}.
\]

\medskip\noindent
\textbf{Step 2: Extraction of the main term.}
Substituting the representation of $\Pr\left[T_n>x\right]$ into
\eqref{eq:D-blanket-Tn} yields
\[
  \mathcal D^{\mathrm{blanket}}
  = \frac{\sigma_{\varepsilon_n}}{\gamma\sqrt{n}}
    \int_{t_n}^\infty (1-\Phi(x))\bigl(1+\rho_n(x)\bigr)\,dx.
\]
Define
\[
  \mathrm{Main}_n
  := \int_{t_n}^\infty (1-\Phi(x))\,dx,
  \qquad
  \mathrm{Err}_n
  := \int_{t_n}^\infty (1-\Phi(x))\rho_n(x)\,dx,
\]
so that
\begin{equation}
  \label{eq:D-main+err}
  \mathcal D^{\mathrm{blanket}}
  = \frac{\sigma_{\varepsilon_n}}{\gamma\sqrt{n}}\,
    \bigl(\mathrm{Main}_n + \mathrm{Err}_n\bigr).
\end{equation}

Fix $B>1$ and split
\[
  \mathrm{Err}_n
  = \int_{t_n}^{B\sqrt{\log n}} (1-\Phi(x))\rho_n(x)\,dx
    + \int_{B\sqrt{\log n}}^\infty (1-\Phi(x))\rho_n(x)\,dx
  =: I_1 + I_2.
\]

For $I_1$, note that for all large $n$ we have
$t_n\le B\sqrt{\log n}\le c\sqrt{\log n}$, so by
\eqref{eq:slastnikov-Z},
\[
  \sup_{x\in[t_n,B\sqrt{\log n}]}|\rho_n(x)|
  \le \Delta_n.
\]
Hence
\[
  |I_1|
  \le \Delta_n \int_{t_n}^{B\sqrt{\log n}} (1-\Phi(x))\,dx
  \le \Delta_n \int_{t_n}^\infty (1-\Phi(x))\,dx
  = \Delta_n\,\mathrm{Main}_n.
\]

For $I_2$ we use crude exponential bounds on the tails.
Since $Z_i$ are bounded, a Chernoff bound yields the existence of
$c_1>0$ such that
\[
  \Pr\left[T_n>x\right] \le {\rm e}^{-c_1 x^2},
  \qquad x\ge0.
\]
Moreover, by Mills' bound there exists $C_1>0$ such that
\[
  1-\Phi(x) \le C_1 {\rm e}^{-x^2/2},
  \qquad x\ge0.
\]
From the definition of $\rho_n$ we have
\(
  (1-\Phi(x))\rho_n(x)
  = \Pr\left[T_n>x\right] - (1-\Phi(x))
\),
and thus
\[
  |(1-\Phi(x))\rho_n(x)|
  \le \Pr\left[T_n>x\right] + (1-\Phi(x))
  \le {\rm e}^{-c_1 x^2} + C_1 {\rm e}^{-x^2/2}.
\]
Therefore
\begin{align*}
  |I_2|
  &\le \int_{B\sqrt{\log n}}^\infty
          \bigl({\rm e}^{-c_1 x^2} + C_1{\rm e}^{-x^2/2}\bigr)\,dx\\
  &\le C_2\, n^{-K(B)},
\end{align*}
for some finite constant $C_2<\infty$ and
$K(B):=\min\{c_1 B^2,\tfrac12 B^2\}>0$.

On the other hand, by the standard asymptotic expansion of the Mills
integral,
\begin{equation}
  \label{eq:Mills-integral}
  \mathrm{Main}_n
  = \int_{t_n}^\infty (1-\Phi(x))\,dx
  = \frac{\varphi(t_n)}{t_n^2}\bigl(1+O(t_n^{-2})\bigr),
  \qquad t_n\to\infty.
\end{equation}
Since $t_n^2 = \Theta(\log n)$, we have
\(
  \mathrm{Main}_n
  = \Theta(\varphi(t_n)/t_n^2)
\),
which decays only polynomially in $n$. Thus, by choosing $B>0$ large
enough, we can ensure that $n^{-K(B)} = o(\mathrm{Main}_n)$, and hence
$|I_2| = o(\mathrm{Main}_n)$.

Combining the bounds for $I_1$ and $I_2$, we obtain
\[
  \mathrm{Err}_n
  = \mathrm{Main}_n\cdot o(1),
\]
and therefore, by \eqref{eq:Mills-integral},
\[
  \mathrm{Main}_n + \mathrm{Err}_n
  = \frac{\varphi(t_n)}{t_n^2}\bigl(1+o(1)\bigr).
\]
Substituting back into \eqref{eq:D-main+err} gives
\begin{equation}
  \label{eq:D-blanket-pre-final}
  \mathcal D^{\mathrm{blanket}}
  = \frac{\sigma_{\varepsilon_n}}{\gamma\sqrt{n}}\,
    \frac{\varphi(t_n)}{t_n^2}\,
    \bigl(1+o(1)\bigr),
  \qquad
  t_n := -\frac{\mu_n\sqrt{n}}{\sigma_{\varepsilon_n}}.
\end{equation}
Using $t_n^2 = \mu_n^2 n/\sigma_{\varepsilon_n}^2$, we can rewrite this
as
\[
  \mathcal D^{\mathrm{blanket}}
  = \varphi(t_n)\,
    \frac{\sigma_{\varepsilon_n}^3}
         {\gamma\,\mu_n^2 n^{3/2}}\,
    \bigl(1+o(1)\bigr).
\]

\medskip\noindent
\textbf{Step 3: Final replacement of $\sigma_{\varepsilon_n}$ by $\sigma$.}
It remains to express the right-hand side in terms of the variance
$\sigma^2 = \mathrm{Var}(l_0(Y))$ and the parameter
$\chi=\sqrt{\gamma}/\sigma$. As in the ``final replacement'' paragraph
for the Bernoulli–weighted variables $Z_\varepsilon=B\,l_\varepsilon(Y)$,
one checks that:

\begin{itemize}
  \item the map $\varepsilon\mapsto\sigma_\varepsilon^2
          :=\mathrm{Var}(Z_\varepsilon)$ is $C^1$ and Lipschitz on
        $[-\varepsilon_0,\varepsilon_0]$, so
        $|\sigma_{\varepsilon_n}-\sigma_0(Z)|\le C|\varepsilon_n|$
        with $\sigma_0(Z)^2=\gamma\sigma^2$;
  \item consequently
        \(
          \sigma_{\varepsilon_n}
          = \sigma_0(Z)\bigl(1+o(1)\bigr)
          = \sqrt{\gamma}\,\sigma\bigl(1+o(1)\bigr)
        \)
        and
        \(
          \sigma_{\varepsilon_n}^3
          = \gamma^{3/2}\sigma^3\bigl(1+o(1)\bigr);
        \)
  \item with
        \(
          t
          := -\chi\bigl({\rm e}^{\varepsilon_n}-1\bigr)\sqrt{n}
        \)
        and
        $\mu_n=\gamma(1-{\rm e}^{\varepsilon_n})$, we have
        $t_n-t=o(1)$ and hence
        $\varphi(t_n) = \varphi(t)\bigl(1+o(1)\bigr)$.
\end{itemize}

Substituting these into \eqref{eq:D-blanket-pre-final} yields
\[
  \mathcal D^{\mathrm{blanket}}
  =
  \varphi\bigl(\chi({\rm e}^{\varepsilon_n}-1)\sqrt{n}\bigr)\,
  \left(
    \frac{1}{\chi^3}\,
    \frac{1}{\bigl({\rm e}^{\varepsilon_n}-1\bigr)^2 n^{3/2}}
  \right)\bigl(1+o(1)\bigr),
\]
and this is exactly \eqref{eq:blanket-asymptotics-general}.
The $o(1)$ term depends only on the constants in
Assumption~\ref{assump:regularity}.
\end{proof}

\end{proof}

\subsection{The proof of Theorem~\ref{thm:moderate-deviation-shuffle}}
\label{app:proof-moderate-deviations-shuffle}

\shufflingdeltaband*

\begin{proof}[Proof of Theorem~\ref{thm:moderate-deviation-shuffle}]
For each $n\ge1$ and $\varepsilon\ge0$, define the global privacy profile
\[
  \delta_n(\varepsilon)
  :=
  \sup_{x_{1:n}\simeq x'_{1:n}}
  \mathcal D_{e^\varepsilon}\bigl(
    \mathcal M(x_{1:n})
    \,\big\|\,
    \mathcal M(x'_{1:n})
  \bigr),
  \qquad
  \mathcal M = \mathcal S\circ\mathcal R^n.
\]
By definition of the hockey-stick divergence, the map
$\alpha\mapsto\mathcal D_\alpha(P\|Q)$ is nonincreasing, hence
$\varepsilon\mapsto\delta_n(\varepsilon)$ is nonincreasing for each fixed
$n$.

\medskip\noindent
\textbf{Step 1: A uniform asymptotic form for the blanket divergence.}
For each pair $x_1\neq x_1'\in\mathcal X$ and each reference distribution
$\mathcal R_{\mathrm{ref}}\in\{\mathcal R_{\mathrm{BG}}\}\cup\{\mathcal R_x:x\in\mathcal X\}$,
Lemma~\ref{lem:blanket-asymptotics-general} gives an asymptotic expansion
of the corresponding blanket divergence in the moderate-deviation regime.
It is convenient to isolate the leading term as
\[
  g_n(\varepsilon;\chi)
  :=
  \varphi\bigl(\chi(\mathrm e^\varepsilon-1)\sqrt{n}\bigr)\,
  \left(
    \frac{1}{\chi^3}\,
    \frac{1}{(\mathrm e^\varepsilon-1)^2 n^{3/2}}
  \right).
\]

By Assumption~\ref{assump:regularity}, the constants in the Edgeworth and
moderate-deviation bounds can be chosen uniformly over all pairs
$(x_1,x_1')$ and all admissible $\mathcal R_{\mathrm{ref}}$. Hence there
exists a sequence $\rho_n\downarrow0$ such that, for every
choice of $(x_1,x_1',\mathcal R_{\mathrm{ref}})$ and every sequence
$(\varepsilon_n)$ with
\[
  \varepsilon_n\to0,\qquad
  \varepsilon_n=\omega(n^{-1/2}),\qquad
  \varepsilon_n=O\!\Bigl(\sqrt{\tfrac{\log n}{n}}\Bigr),
\]
we have
\begin{equation}
  \label{eq:unif-blanket-asymptotic}
  \mathcal D^{\mathrm{blanket}}_{e^{\varepsilon_n},n,\mathcal R_{\mathrm{ref}},\gamma}
  \;=\;
  g_n\bigl(\varepsilon_n;\chi(x_1,x_1';\mathcal R_{\mathrm{ref}})\bigr)\,
  \bigl(1+\theta_n(x_1,x_1',\mathcal R_{\mathrm{ref}})\bigr),
  \qquad
  \bigl|\theta_n(x_1,x_1',\mathcal R_{\mathrm{ref}})\bigr|\le\rho_n,
\end{equation}
where $\chi(x_1,x_1';\mathcal R_{\mathrm{ref}})$ is the shuffle index
associated with $(x_1,x_1')$ and $\mathcal R_{\mathrm{ref}}$ as in
Lemma~\ref{lem:blanket-asymptotics-general}.

A direct computation shows that for each fixed $n$ and $\varepsilon>0$ the
map $\chi\mapsto g_n(\varepsilon;\chi)$ is strictly decreasing on
$(0,\infty)$. Indeed, if we set $u:=\chi(\mathrm e^\varepsilon-1)\sqrt{n}$,
then
\[
  g_n(\varepsilon;\chi)
  =
  \frac{1}{\sqrt{2\pi}}
  \exp\!\Bigl(-\frac{u^2}{2}\Bigr)
  \cdot
  \frac{1}{\chi u^2\sqrt{n}},
\]
and the right-hand side is strictly decreasing in $\chi$.

\medskip\noindent
\textbf{Step 2: Solving the leading equation via the Lambert $W$-function.}
Fix $\alpha>0$ and $\chi>0$. We now solve the equation
$g_n(\varepsilon;\chi)=\alpha/n$ explicitly. As above, write
$u:=\chi(\mathrm e^\varepsilon-1)\sqrt{n}>0$. Then
\[
  g_n(\varepsilon;\chi)
  =
  \frac{1}{\sqrt{2\pi}}
  \exp\!\Bigl(-\frac{u^2}{2}\Bigr)
  \cdot
  \frac{1}{\chi u^2\sqrt{n}}.
\]
The equation $g_n(\varepsilon;\chi)=\alpha/n$ is thus equivalent to
\[
  \frac{1}{\sqrt{2\pi}}
  \exp\!\Bigl(-\frac{u^2}{2}\Bigr)
  \frac{1}{\chi u^2\sqrt{n}}
  = \frac{\alpha}{n},
\]
which can be rearranged as
\[
  \exp\!\Bigl(\frac{u^2}{2}\Bigr)
  =
  \frac{\sqrt{n}}{\alpha\sqrt{2\pi}\,\chi u^2}.
\]
Multiplying both sides by $u^2/2$ yields
\[
  \frac{u^2}{2}\,
  \exp\!\Bigl(\frac{u^2}{2}\Bigr)
  =
  \frac{\sqrt{n}}{2\alpha\chi\sqrt{2\pi}}.
\]
Setting $z:=u^2/2$ we obtain $z\,e^z=\frac{\sqrt{n}}{2\alpha\chi\sqrt{2\pi}}$,
hence
\[
  z = W\!\left(
    \frac{\sqrt{n}}{2\alpha\chi\sqrt{2\pi}}
  \right),
  \qquad
  u^2
  =
  2\,W\!\left(
    \frac{\sqrt{n}}{2\alpha\chi\sqrt{2\pi}}
  \right),
\]
where $W(\cdot)$ is the principal branch of the Lambert $W$-function. It
follows that
\[
  \mathrm e^\varepsilon-1
  = \frac{u}{\chi\sqrt{n}}
  =
  \sqrt{
    \frac{2}{\chi^2 n}\,
    W\!\left(
      \frac{\sqrt{n}}{2\alpha\chi\sqrt{2\pi}}
    \right)
  },
\]
and hence
\[
  \varepsilon_n(\alpha,\chi)
  :=
  \log\left(
    1+
    \sqrt{
      \frac{2}{\chi^2 n}\,
      W\!\left(
        \frac{\sqrt{n}}{2\alpha\chi\sqrt{2\pi}}
      \right)
    }
  \right)
\]
is the unique solution of $g_n(\varepsilon;\chi)=\alpha/n$ with
$\varepsilon>0$.

From the standard asymptotics $W(z)=\log z-\log\log z+O(1)$ as $z\to\infty$ we
see that
\[
  W\!\left(
    \frac{\sqrt{n}}{2\alpha\chi\sqrt{2\pi}}
  \right)
  = \Theta(\log n),
\]
and therefore
\[
  \mathrm e^{\varepsilon_n(\alpha,\chi)}-1
  =
  \Theta\left(\sqrt{\frac{\log n}{n}}\right),
  \qquad
  \varepsilon_n(\alpha,\chi)
  =
  \Theta\left(\sqrt{\frac{\log n}{n}}\right),
\]
which implies $\varepsilon_n(\alpha,\chi)\to0$,
$\varepsilon_n(\alpha,\chi)=\omega(n^{-1/2})$ and
$\varepsilon_n(\alpha,\chi)=O(\sqrt{\log n/n})$. Thus $\varepsilon_n(\alpha,\chi)$
lies in the moderate-deviation regime where
Lemma~\ref{lem:blanket-asymptotics-general} is applicable. Moreover, from
the explicit expression it is clear that $\chi\mapsto\varepsilon_n(\alpha,\chi)$
is nonincreasing.

\medskip\noindent
\textbf{Step 3: An upper bound via the blanket divergence.}
We first derive an asymptotic upper bound on $\delta_n(\varepsilon)$.
By Lemma~\ref{lem:upper-blanket-divergence}, for any neighboring
datasets $x_{1:n}\simeq x'_{1:n}$ differing in the first coordinate we
have
\[
  \mathcal D_{e^\varepsilon}\bigl(
    \mathcal M(x_{1:n})
    \,\big\|\,
    \mathcal M(x'_{1:n})
  \bigr)
  \;\le\;
  \mathcal D^{\mathrm{blanket}}_{e^\varepsilon,n,\mathcal R_{\mathrm{BG}},\gamma}
    (\mathcal R_{x_1}\,\|\,\mathcal R_{x_1'}).
\]
Taking the supremum over $x_{1:n}\simeq x'_{1:n}$ yields
\[
  \delta_n(\varepsilon)
  \le
  \sup_{x_1\neq x_1'}
  \mathcal D^{\mathrm{blanket}}_{e^\varepsilon,n,\mathcal R_{\mathrm{BG}},\gamma}
    (\mathcal R_{x_1}\,\|\,\mathcal R_{x_1'}).
\]

For each pair $(x_1,x_1')$ define
\[
  \chi_{\mathrm{lo}}(x_1,x_1')
  :=
  \sqrt{
    \frac{\gamma}{
      \mathrm{Var}_{\mathcal R_{\mathrm{BG}}}
      \bigl(
        l_0(Y;x_1,x_1',\mathcal R_{\mathrm{BG}})
      \bigr)
    }
  },
\]
and let
\[
  \underline\chi_{\mathrm{lo}}
  :=
  \inf_{x_1\neq x_1'} \chi_{\mathrm{lo}}(x_1,x_1').
\]
Assumption~\ref{assump:regularity} implies that the variance in the
denominator is finite and bounded away from zero uniformly over all
pairs, so that $0<\underline\chi_{\mathrm{lo}}<\infty$.

Applying the uniform asymptotic
\eqref{eq:unif-blanket-asymptotic} with $\mathcal R_{\mathrm{ref}}=\mathcal R_{\mathrm{BG}}$
and $\varepsilon_n=\varepsilon_n(\alpha,\underline\chi_{\mathrm{lo}})$
(which lies in the moderate-deviation regime by Step~2) and using the
monotonicity of $g_n$ in $\chi$ yields
\[
  \begin{aligned}
  \mathcal D^{\mathrm{blanket}}_{e^{\varepsilon_n},n,\mathcal R_{\mathrm{BG}},\gamma}
    (\mathcal R_{x_1}\,\|\,\mathcal R_{x_1'})
  &=
  g_n\bigl(\varepsilon_n(\alpha,\underline\chi_{\mathrm{lo}});
           \chi_{\mathrm{lo}}(x_1,x_1')\bigr)\,
  \bigl(1+\theta_n(x_1,x_1',\mathcal R_{\mathrm{BG}})\bigr) \\
  &\le
  (1+\rho_n)\,
  g_n\bigl(\varepsilon_n(\alpha,\underline\chi_{\mathrm{lo}});
           \underline\chi_{\mathrm{lo}}\bigr),
  \end{aligned}
\]
uniformly over all $x_1\neq x_1'$. By construction of
$\varepsilon_n(\alpha,\underline\chi_{\mathrm{lo}})$ in Step~2 we have
\[
  g_n\bigl(\varepsilon_n(\alpha,\underline\chi_{\mathrm{lo}});
           \underline\chi_{\mathrm{lo}}\bigr)
  = \frac{\alpha}{n},
\]
hence
\[
  \sup_{x_1\neq x_1'}
  \mathcal D^{\mathrm{blanket}}_{e^{\varepsilon_n},n,\mathcal R_{\mathrm{BG}},\gamma}
    (\mathcal R_{x_1}\,\|\,\mathcal R_{x_1'})
  \le
  \frac{\alpha}{n}\,(1+\rho_n).
\]
Combining this with the upper bound on $\delta_n$ gives
\[
  \delta_n\bigl(\varepsilon_n(\alpha,\underline\chi_{\mathrm{lo}})\bigr)
  \le
  \frac{\alpha}{n}\,(1+\rho_n)
  =
  \frac{\alpha}{n}\,(1+o(1)).
\]
Therefore the mechanism $\mathcal M$ satisfies
\[
  \bigl(\varepsilon_n(\alpha,\underline\chi_{\mathrm{lo}}),
        \tfrac{\alpha}{n}(1+o(1))\bigr)\text{-DP}.
\]

\medskip\noindent
\textbf{Step 4: A lower bound via Su et al.'s blanket representation.}
We now show how the lower bound. By Theorem~\ref{theorem:lower-blanket-divergence}, for any
$x,a_1,a_1'\in\mathcal X$ and any $\varepsilon\ge0$,
\[
  \delta_n(\varepsilon)
  \ge
  \mathcal D^{\mathrm{blanket}}_{e^\varepsilon,n,\mathcal R_x,1}
    (\mathcal R_{a_1}\,\|\,\mathcal R_{a_1'}).
\]

For each pair $(x_1,x_1')$ define
\[
  \chi_{\mathrm{up}}(x_1,x_1')
  :=
  \inf_{x\in\mathcal X}
  \sqrt{
    \frac{1}{\mathrm{Var}_{\mathcal R_x}
      \bigl(
        l_0(Y;x_1,x_1',\mathcal R_x)
      \bigr)
    }
  },
\]
and set
\[
  \underline\chi_{\mathrm{up}}
  :=
  \inf_{x_1\neq x_1'} \chi_{\mathrm{up}}(x_1,x_1').
\]
Again, Assumption~\ref{assump:regularity} guarantees
$0<\underline\chi_{\mathrm{up}}<\infty$.

Fix $\eta>0$. By definition of the infimum, there exist
$x^{(\eta)}\in\mathcal X$ and a pair
$(a_1^{(\eta)},a_1^{(\eta)}{}')$ with $a_1^{(\eta)}\neq a_1^{(\eta)}{}'$ such
that
\[
  \chi\bigl(a_1^{(\eta)},a_1^{(\eta)}{}';\mathcal R_{x^{(\eta)}}\bigr)
  \le
  \underline\chi_{\mathrm{up}}+\eta,
\]
where the left-hand side is the shuffle index associated with the
reference distribution $\mathcal R_{x^{(\eta)}}$.
Applying the asymptotic
\eqref{eq:unif-blanket-asymptotic} with
$\mathcal R_{\mathrm{ref}}=\mathcal R_{x^{(\eta)}}$ and
$\varepsilon_n=\varepsilon_n(\alpha,\underline\chi_{\mathrm{up}})$, and
using that $g_n(\varepsilon;\chi)$ is decreasing in $\chi$, we obtain
\[
  \begin{aligned}
  \mathcal D^{\mathrm{blanket}}_{e^{\varepsilon_n},n,\mathcal R_{x^{(\eta)}},1}
    \bigl(
      \mathcal R_{a_1^{(\eta)}}
      \,\big\|\,
      \mathcal R_{a_1^{(\eta)}{}'}
    \bigr)
  &=
  g_n\bigl(\varepsilon_n(\alpha,\underline\chi_{\mathrm{up}});
           \chi(a_1^{(\eta)},a_1^{(\eta)}{}';\mathcal R_{x^{(\eta)}})\bigr)\,
  \bigl(1+\theta_n(x^{(\eta)},a_1^{(\eta)},a_1^{(\eta)}{}')\bigr) \\
  &\ge
  (1-\rho_n)\,
  g_n\bigl(\varepsilon_n(\alpha,\underline\chi_{\mathrm{up}});
           \underline\chi_{\mathrm{up}}+\eta\bigr),
  \end{aligned}
\]
for all sufficiently large $n$.

By construction of $\varepsilon_n(\alpha,\underline\chi_{\mathrm{up}})$ we
have
\[
  g_n\bigl(\varepsilon_n(\alpha,\underline\chi_{\mathrm{up}});
           \underline\chi_{\mathrm{up}}\bigr)
  =
  \frac{\alpha}{n}.
\]
The continuity of $g_n(\varepsilon;\chi)$ in $\chi$ and the fact that
$g_n$ is smooth in $\chi$ on $(0,\infty)$ imply that, for each fixed $n$,
\[
  g_n\bigl(\varepsilon_n(\alpha,\underline\chi_{\mathrm{up}});
           \underline\chi_{\mathrm{up}}+\eta\bigr)
  =
  \frac{\alpha}{n}\bigl(1+\kappa_n(\eta)\bigr),
\]
where $\kappa_n(\eta)\to0$ as $\eta\downarrow0$, uniformly in $n$ on
compact subsets of $(0,\infty)$ for $\underline\chi_{\mathrm{up}}$. Hence,
for any fixed $\eta>0$,
\[
  \mathcal D^{\mathrm{blanket}}_{e^{\varepsilon_n},n,\mathcal R_{x^{(\eta)}},1}
    \bigl(
      \mathcal R_{a_1^{(\eta)}}
      \,\big\|\,
      \mathcal R_{a_1^{(\eta)}{}'}
    \bigr)
  \ge
  \frac{\alpha}{n}
  \bigl(1-\rho_n\bigr)\bigl(1+\kappa_n(\eta)\bigr).
\]
Combining this with Theorem~\ref{theorem:lower-blanket-divergence} gives
\[
  \delta_n\bigl(\varepsilon_n(\alpha,\underline\chi_{\mathrm{up}})\bigr)
  \;\ge\;
  \frac{\alpha}{n}
  \bigl(1-\rho_n\bigr)\bigl(1+\kappa_n(\eta)\bigr).
\]
Letting $n\to\infty$ and then $\eta\downarrow0$ shows that there exists a
sequence $\xi_n\to0$ such that
\[
  \delta_n\bigl(\varepsilon_n(\alpha,\underline\chi_{\mathrm{up}})\bigr)
  \ge
  \frac{\alpha}{n}\,(1-\xi_n)
  =
  \frac{\alpha}{n}\,(1+o(1)).
\]

\medskip\noindent
\textbf{Step 5: Existence of a true optimal privacy curve.}
Fix $\alpha>0$. We now show that, for all sufficiently large $n$, there
exists an optimal privacy parameter $\varepsilon_n^{\ast}$
satisfying
\[
  \delta_n\bigl(\varepsilon_n^{\ast}\bigr)
  = \frac{\alpha}{n}.
\]

The hockey-stick divergence is continuous, nonincreasing in $\varepsilon$, and converges to $0$ as $\varepsilon\to\infty$ because $\mathcal{R}$ satisfies Assumption~\ref{assump:regularity}; the distributions are mutually absolutely continuous.
Thus, for each fixed $n$,
$
  \lim_{\varepsilon\to\infty}\delta_n(\varepsilon) = 0.
$

Fix $\alpha>0$. Since $\alpha/n\to0$ as $n\to\infty$, for all
sufficiently large $n$ we have
\[
  \delta_n(0) > \frac{\alpha}{n}
  \qquad\text{and}\qquad
  \lim_{\varepsilon\to\infty}\delta_n(\varepsilon)=0
  < \frac{\alpha}{n}.
\]
By monotonicity, the set
\[
  S_n(\alpha)
  :=
  \Bigl\{\varepsilon\ge0:
     \delta_n(\varepsilon)\le\frac{\alpha}{n}
  \Bigr\}
\]
is nonempty and contained in $[0,\infty)$. Define
\[
  \varepsilon_n^{\ast}(\alpha)
  :=
  \inf S_n(\alpha)\in[0,\infty).
\]
The right-continuity of $\delta_n$ implies
$\delta_n(\varepsilon_n^{\ast}(\alpha))\le\alpha/n$, while the
definition of the infimum and the strict inequality
$\delta_n(0)>\alpha/n$ imply
$\delta_n(\varepsilon_n^{\ast}(\alpha))\ge\alpha/n$. Hence, for
all sufficiently large $n$,
\[
  \delta_n\bigl(\varepsilon_n^{\ast}(\alpha)\bigr)
  = \frac{\alpha}{n}.
\]

Thus we have constructed an optimal privacy curve
$(\varepsilon_n^{\ast}(\alpha))_{n\ge1}$: for each large $n$ it
is the smallest privacy parameter for which the shuffled mechanism
$\mathcal M$ satisfies $(\varepsilon,\alpha/n)$-DP. Comparing this curve
with the explicit bounds obtained in Steps~3 and~4 shows that it lies
asymptotically within the band
\[
  \varepsilon_n(\alpha,\underline\chi_{\mathrm{up}})
  \;\le\;
  \varepsilon_n^{\ast}(\alpha)
  \;\le\;
  \varepsilon_n(\alpha,\underline\chi_{\mathrm{lo}})
  \quad\text{up to a vanishing relative error as }n\to\infty,
\]
which completes the argument.

\end{proof}

\subsection{The proof of Theorem~\ref{thm:global-band-collapse}}
\label{app:proof-shuffle-index-structure}

\globalbandcollapse*

\begin{proof}

In the purely discrete case, all integrals over $\mathcal Y$ should be
read as sums over $y\in\mathcal Y$.
Recall the notation from Section~\ref{sec:shuffle_index}.
For $x\in\mathcal X$ we write $\mathcal R_x(y)$ for the density (or pmf)
of $\mathcal R_x$ with respect to the reference measure on $\mathcal Y$.
We set
\[
  \underline{\mathcal R}(y)
  := \inf_{z\in\mathcal X} \mathcal R_z(y),\qquad
  \gamma := \int_{\mathcal Y} \underline{\mathcal R}(y)\,dy,
  \qquad
  \mathcal R_{\mathrm{BG}}(y)
  := \frac{1}{\gamma}\,\underline{\mathcal R}(y).
\]
Thus $\underline{\mathcal R}(y) = \gamma\,\mathcal R_{\mathrm{BG}}(y)$
for all $y\in\mathcal Y$.
For $x_1\neq x_1'\in\mathcal X$ we define the disagreement set
\[
  A(x_1,x_1')
  := \bigl\{ y\in\mathcal Y : \mathcal R_{x_1}(y)\neq \mathcal R_{x_1'}(y)\bigr\}.
\]

We first prove the following lemma.

\begin{lemma}[Pairwise ordering and equality condition]
\label{lem:pairwise-chi-structure}
Let $x_1\neq x_1'\in\mathcal X$ be fixed and consider the pairwise
shuffle indices
\[
  \chi_{\mathrm{lo}}(x_1,x_1')
  := \sqrt{
        \frac{\gamma}{
          \mathrm{Var}_{Y\sim\mathcal R_{\mathrm{BG}}}
          \bigl(l_0(Y;x_1,x_1',\mathcal R_{\mathrm{BG}})\bigr)
        }
      },
  \qquad
  \chi_{\mathrm{up}}(x_1,x_1')
  := \inf_{x\in\mathcal X}
      \sqrt{
        \frac{1}{
          \mathrm{Var}_{Y\sim\mathcal R_x}
          \bigl(l_0(Y;x_1,x_1',\mathcal R_x)\bigr)
        }
      }.
\]
Then
\[
  \chi_{\mathrm{up}}(x_1,x_1')
  \;\ge\;
  \chi_{\mathrm{lo}}(x_1,x_1')
  \qquad\text{for all }x_1\neq x_1'.
\]
Moreover, equality holds if and only if there exists an input
$x^\ast\in\mathcal X$ such that
\[
  \mathcal R_{x^\ast}(y)
  = \underline{\mathcal R}(y)
  \qquad\text{for almost every }y\in A(x_1,x_1').
\]
\end{lemma}

\begin{proof}
Fix $x_1\neq x_1'\in\mathcal X$ and abbreviate
\[
  \Delta(y)
  := \mathcal R_{x_1}(y) - \mathcal R_{x_1'}(y),
  \qquad y\in\mathcal Y.
\]
Note that $\int_{\mathcal Y}\Delta(y)\,dy=0$, since both
$\mathcal R_{x_1}$ and $\mathcal R_{x_1'}$ are probability distributions
with respect to the same reference measure.

\textbf{Step 1: Variance under $\mathcal R_{\mathrm{BG}}$.}
For $Y\sim\mathcal R_{\mathrm{BG}}$ we have
\[
  l_0(Y;x_1,x_1',\mathcal R_{\mathrm{BG}})
  = \frac{\mathcal R_{x_1}(Y)-\mathcal R_{x_1'}(Y)}
         {\mathcal R_{\mathrm{BG}}(Y)}
  = \frac{\Delta(Y)}{\mathcal R_{\mathrm{BG}}(Y)}.
\]
Since
$\mathbb E_{\mathcal R_{\mathrm{BG}}}[l_0(Y;x_1,x_1',\mathcal R_{\mathrm{BG}})]
= 0$, its variance is
\begin{align*}
  \sigma_{\mathrm{BG}}^2(x_1,x_1')
  &:= \mathrm{Var}_{Y\sim\mathcal R_{\mathrm{BG}}}
      \bigl(l_0(Y;x_1,x_1',\mathcal R_{\mathrm{BG}})\bigr)\\
  &= \int_{\mathcal Y}
       \frac{\Delta(y)^2}{\mathcal R_{\mathrm{BG}}(y)}\,dy.
\end{align*}
On the complement of $A(x_1,x_1')$ we have $\Delta(y)=0$, so
\[
  \sigma_{\mathrm{BG}}^2(x_1,x_1')
  = \int_{A(x_1,x_1')}
      \frac{\Delta(y)^2}{\mathcal R_{\mathrm{BG}}(y)}\,dy.
\]

\textbf{Step 2: Variance under $\mathcal R_x$.}
Similarly, for any $x\in\mathcal X$ and $Y\sim\mathcal R_x$,
\[
  l_0(Y;x_1,x_1',\mathcal R_x)
  = \frac{\mathcal R_{x_1}(Y)-\mathcal R_{x_1'}(Y)}
         {\mathcal R_x(Y)}
  = \frac{\Delta(Y)}{\mathcal R_x(Y)},
\]
and again the mean is zero. Hence
\begin{align*}
  \sigma_x^2(x_1,x_1')
  &:= \mathrm{Var}_{Y\sim\mathcal R_x}
      \bigl(l_0(Y;x_1,x_1',\mathcal R_x)\bigr)\\
  &= \int_{\mathcal Y}
       \frac{\Delta(y)^2}{\mathcal R_x(y)}\,dy
   = \int_{A(x_1,x_1')}
       \frac{\Delta(y)^2}{\mathcal R_x(y)}\,dy.
\end{align*}

\textbf{Step 3: Comparison of variances.}
By definition of the blanket, for all $x\in\mathcal X$ and $y\in\mathcal Y$
\[
  \underline{\mathcal R}(y)
  \le \mathcal R_x(y).
\]
Since
$\mathcal R_{\mathrm{BG}}(y)
 = \underline{\mathcal R}(y)/\gamma$, this implies
\[
  \mathcal R_x(y)
  \;\ge\;
  \underline{\mathcal R}(y)
  = \gamma\,\mathcal R_{\mathrm{BG}}(y)
  \qquad\text{for all }x\in\mathcal X,\ y\in\mathcal Y.
\]
Consequently, for all $x\in\mathcal X$ and all $y$ with $\Delta(y)\neq 0$,
\[
  \frac{\Delta(y)^2}{\mathcal R_x(y)}
  \;\le\;
  \frac{\Delta(y)^2}{\gamma\,\mathcal R_{\mathrm{BG}}(y)}.
\]
Integrating over $A(x_1,x_1')$ yields
\[
  \sigma_x^2(x_1,x_1')
  = \int_{A(x_1,x_1')}
      \frac{\Delta(y)^2}{\mathcal R_x(y)}\,dy
  \;\le\;
  \frac{1}{\gamma}
  \int_{A(x_1,x_1')}
      \frac{\Delta(y)^2}{\mathcal R_{\mathrm{BG}}(y)}\,dy
  = \frac{1}{\gamma}\,\sigma_{\mathrm{BG}}^2(x_1,x_1').
\]
Taking the supremum over $x\in\mathcal X$ gives
\[
  \sup_{x\in\mathcal X} \sigma_x^2(x_1,x_1')
  \;\le\;
  \frac{1}{\gamma}\,\sigma_{\mathrm{BG}}^2(x_1,x_1').
\]

By definition,
\[
  \chi_{\mathrm{lo}}(x_1,x_1')
  = \sqrt{\frac{\gamma}{\sigma_{\mathrm{BG}}^2(x_1,x_1')}},
  \qquad
  \chi_{\mathrm{up}}(x_1,x_1')
  = \inf_{x\in\mathcal X}
      \sqrt{\frac{1}{\sigma_x^2(x_1,x_1')}}
  = \frac{1}{\sqrt{\sup_{x\in\mathcal X} \sigma_x^2(x_1,x_1')}}.
\]
Using the inequality on the variances we obtain
\[
  \chi_{\mathrm{up}}(x_1,x_1')
  = \frac{1}{\sqrt{\sup_x \sigma_x^2(x_1,x_1')}}
  \;\ge\;
  \sqrt{\frac{\gamma}{\sigma_{\mathrm{BG}}^2(x_1,x_1')}}
  = \chi_{\mathrm{lo}}(x_1,x_1').
\]

\textbf{Step 4: Characterization of equality.}
We now analyze when equality
$\chi_{\mathrm{up}}(x_1,x_1')=\chi_{\mathrm{lo}}(x_1,x_1')$ holds.
By the definitions above this is equivalent to
\[
  \sup_{x\in\mathcal X} \sigma_x^2(x_1,x_1')
  = \frac{1}{\gamma}\,\sigma_{\mathrm{BG}}^2(x_1,x_1').
\]
From the pointwise inequality
\[
  \frac{\Delta(y)^2}{\mathcal R_x(y)}
  \;\le\;
  \frac{\Delta(y)^2}{\gamma\,\mathcal R_{\mathrm{BG}}(y)},
  \qquad y\in A(x_1,x_1'),
\]
we see that for any fixed $x\in\mathcal X$ we have
\[
  \sigma_x^2(x_1,x_1')
  = \frac{1}{\gamma}\,\sigma_{\mathrm{BG}}^2(x_1,x_1')
  \quad\Longleftrightarrow\quad
  \mathcal R_x(y)
  = \gamma\,\mathcal R_{\mathrm{BG}}(y)
  = \underline{\mathcal R}(y)
  \quad\text{for almost every }y\in A(x_1,x_1').
\]
Indeed, if $\mathcal R_x(y)>\gamma\,\mathcal R_{\mathrm{BG}}(y)$ on a
subset of $A(x_1,x_1')$ of positive measure where $\Delta(y)\neq0$, the
inequality would be strict on that set and hence strict after integration.

Thus equality
$\sup_x \sigma_x^2(x_1,x_1') = \sigma_{\mathrm{BG}}^2(x_1,x_1')/\gamma$
holds if and only if there exists $x^\ast\in\mathcal X$ such that
$\sigma_{x^\ast}^2(x_1,x_1') = \sigma_{\mathrm{BG}}^2(x_1,x_1')/\gamma$,
which in turn is equivalent to
\[
  \mathcal R_{x^\ast}(y)
  = \underline{\mathcal R}(y)
  \quad\text{for almost every }y\in A(x_1,x_1').
\]
Combining this with the previous step proves the lemma.
\end{proof}

Assume that the infima in the definitions of
$\chi_{\mathrm{lo}}$ and $\chi_{\mathrm{up}}$ are
attained by some distinct pairs
$(x_1^{\mathrm{lo}},x_1^{\prime\,\mathrm{lo}})$ and
$(x_1^{\mathrm{up}},x_1^{\prime\,\mathrm{up}})$, respectively.

\textbf{Step 1: $(i)\Rightarrow(ii)$.}
Suppose first that
$\chi_{\mathrm{up}}=\chi_{\mathrm{lo}}$.
By definition,
\[
  \chi_{\mathrm{lo}}
  = \chi_{\mathrm{lo}}(x_1^{\mathrm{lo}},x_1^{\prime\,\mathrm{lo}}),
  \qquad
  \chi_{\mathrm{up}}
  = \chi_{\mathrm{up}}(x_1^{\mathrm{up}},x_1^{\prime\,\mathrm{up}}).
\]
Lemma~\ref{lem:pairwise-chi-structure} implies that for every neighboring
pair $(x_1,x_1')$,
\[
  \chi_{\mathrm{up}}(x_1,x_1')
  \;\ge\;
  \chi_{\mathrm{lo}}(x_1,x_1').
\]
In particular,
\[
  \chi_{\mathrm{up}}(x_1^{\mathrm{lo}},x_1^{\prime\,\mathrm{lo}})
  \;\ge\;
  \chi_{\mathrm{lo}}(x_1^{\mathrm{lo}},x_1^{\prime\,\mathrm{lo}})
  = \chi_{\mathrm{lo}},
\]
and by the definition of $\chi_{\mathrm{up}}$,
\[
  \chi_{\mathrm{up}}
  \le \chi_{\mathrm{up}}(x_1^{\mathrm{lo}},x_1^{\prime\,\mathrm{lo}}).
\]
Combining these inequalities with
$\chi_{\mathrm{up}}=\chi_{\mathrm{lo}}$ yields
\[
  \chi_{\mathrm{up}}(x_1^{\mathrm{lo}},x_1^{\prime\,\mathrm{lo}})
  = \chi_{\mathrm{lo}}(x_1^{\mathrm{lo}},x_1^{\prime\,\mathrm{lo}})
  = \chi_{\mathrm{lo}}.
\]
Applying Lemma~\ref{lem:pairwise-chi-structure} to the pair
$(x_1^\ast,x_1^{\prime\,\ast})
 := (x_1^{\mathrm{lo}},x_1^{\prime\,\mathrm{lo}})$, we conclude that
there exists $x^\ast\in\mathcal X$ such that
\[
  \mathcal R_{x^\ast}(y)
  = {\mathcal R}(y)
  \quad\text{for almost every }y\in A(x_1^\ast,x_1^{\prime\,\ast}),
\]
and, by construction,
$\chi_{\mathrm{lo}}(x_1^\ast,x_1^{\prime\,\ast})=\chi_{\mathrm{lo}}$.
This is precisely condition~\textup{(ii)}.

\textbf{Step 2: $(ii)\Rightarrow(i)$.}
Conversely, suppose there exist a distinct pair
$(x_1^\ast,x_1^{\prime\,\ast})$ and $x^\ast\in\mathcal X$ such that
\[
  \chi_{\mathrm{lo}}(x_1^\ast,x_1^{\prime\,\ast})
  = \chi_{\mathrm{lo}}
  \qquad\text{and}\qquad
  \mathcal R_{x^\ast}(y)
  = \underline{\mathcal R}(y)
  \ \text{ for almost every }y\in A(x_1^\ast,x_1^{\prime\,\ast}).
\]
By Lemma~\ref{lem:pairwise-chi-structure}, this implies
\[
  \chi_{\mathrm{up}}(x_1^\ast,x_1^{\prime\,\ast})
  = \chi_{\mathrm{lo}}(x_1^\ast,x_1^{\prime\,\ast})
  = \chi_{\mathrm{lo}}.
\]
On the other hand, for any distinct pair $(x_1,x_1')$ we have
\[
  \chi_{\mathrm{up}}(x_1,x_1')
  \;\ge\;
  \chi_{\mathrm{lo}}(x_1,x_1')
  \;\ge\;
  \chi_{\mathrm{lo}}
\]
by Lemma~\ref{lem:pairwise-chi-structure} and the definition of
$\chi_{\mathrm{lo}}$.
Therefore
\[
  \chi_{\mathrm{up}}
  = \inf_{x_1\neq x_1'} \chi_{\mathrm{up}}(x_1,x_1')
  \;\ge\;
  \chi_{\mathrm{lo}},
\]
while the existence of $(x_1^\ast,x_1^{\prime\,\ast})$ with
$\chi_{\mathrm{up}}(x_1^\ast,x_1^{\prime\,\ast})=\chi_{\mathrm{lo}}$
shows that
\[
  \chi_{\mathrm{up}}
  \le \chi_{\mathrm{up}}(x_1^\ast,x_1^{\prime\,\ast})
  = \chi_{\mathrm{lo}}.
\]
Hence $\chi_{\mathrm{up}}=\chi_{\mathrm{lo}}$, which
is condition~\textup{(i)}. This completes the proof.

\end{proof}

\subsection{The proof of Lemma~\ref{lemma:trans_fft}}
\label{app:proof-blanket-transform-fft}

\blankettransformfft*

\begin{proof}
Recall the definition of the blanket divergence (generalized to an
arbitrary reference distribution $\mathcal R_{\mathrm{ref}}$):
for $M_0 \sim \mathrm{Bin}(n,\gamma)$ and
$Y_1,Y_2,\dots \stackrel{\mathrm{i.i.d.}}{\sim}\mathcal R_{\mathrm{ref}}$,
\[
  \mathcal{D}_{\mathrm{e}^\varepsilon,n,\mathcal{R}_\mathrm{ref},\gamma}^\text{blanket}(\mathcal{R}_{x_1} \| \mathcal{R}_{x_1^\prime})
  \;=\;
  \frac{1}{n\gamma}
  \,\mathbb{E}\!\left[
    \Bigl( \sum_{i=1}^{M_0} l_\varepsilon(Y_i) \Bigr)_{+}
  \right],
\]
where $a_+ := \max\{a,0\}$ and
\[
  l_{\varepsilon}(y)
  = l_{\varepsilon}(y;x_1,x_1',\mathcal R_{\mathrm{ref}})
  := \frac{\mathcal R_{x_1}(y)-\mathrm e^\varepsilon \mathcal R_{x_1'}(y)}
           {\mathcal R_{\mathrm{ref}}(y)}.
\]

We will transform this expression into the claimed probability
representation in two steps.

\medskip\noindent
\textbf{Step 1: Size-biasing the binomial and isolating one summand.}
Let $X_i := l_\varepsilon(Y_i)$ and
\[
  S_m := \sum_{i=1}^m X_i, \qquad m\ge1.
\]
Then, for each fixed $m\ge1$,
\[
  (S_m)_+ = S_m \mathbf 1\{S_m>0\}
  = \sum_{i=1}^m X_i\,\mathbf 1\{S_m>0\}.
\]
Taking expectation and using the symmetry of the $X_i$,
\[
  \mathbb E\bigl[(S_m)_+\bigr]
  = \sum_{i=1}^m \mathbb E\bigl[X_i\,\mathbf 1\{S_m>0\}\bigr]
  = m\,\mathbb E\bigl[X_1\,\mathbf 1\{S_m>0\}\bigr].
\]
Hence, conditioning on $M_0$,
\begin{align*}
  \mathcal{D}_{\mathrm{e}^\varepsilon,n,\mathcal{R}_\mathrm{ref},\gamma}^\text{blanket}
  &= \frac{1}{n\gamma}\,
     \mathbb{E}\bigl[ (S_{M_0})_+ \bigr] \\
  &= \frac{1}{n\gamma}\,
     \mathbb E\bigl[\,\mathbb E\bigl[(S_{M_0})_+\bigm|M_0\bigr]\,\bigr] \\
  &= \frac{1}{n\gamma}\,
     \mathbb E\bigl[\,M_0\,
       \mathbb E\bigl[X_1\,\mathbf 1\{S_{M_0}>0\}\bigm|M_0\bigr]
     \bigr].
\end{align*}
Now recall that the size-biased version of a nonnegative integer-valued
random variable $M_0$ is defined by
\[
  \Pr[M=m]
  := \frac{m\,\Pr[M_0=m]}{\mathbb E[M_0]}\qquad(m\ge1).
\]
When $M_0\sim\mathrm{Bin}(n,\gamma)$, its size-biased version is
\[
  M \;\stackrel{d}{=}\; 1+\mathrm{Bin}(n-1,\gamma),
\]
and $\mathbb E[M_0] = n\gamma$.
Therefore, for any function $g$,
\[
  \frac{1}{n\gamma}\,\mathbb E\bigl[M_0 g(M_0)\bigr]
  = \mathbb E\bigl[g(M)\bigr].
\]
Applying this with
$g(m) = \mathbb E[X_1\,\mathbf 1\{S_m>0\}\mid M_0=m]$, we obtain
\[
  \mathcal{D}_{\mathrm{e}^\varepsilon,n,\mathcal{R}_\mathrm{ref},\gamma}^\text{blanket}
  = \mathbb E\bigl[
      X_1\,\mathbf 1\{S_M>0\}
    \bigr],
\]
where now $M\sim 1+\mathrm{Bin}(n-1,\gamma)$ and
$S_M = \sum_{i=1}^M X_i = \sum_{i=1}^M l_\varepsilon(Y_i)$.

In other words,
\begin{equation}
  \label{eq:blanket-as-expectation-X1}
  \mathcal{D}_{\mathrm{e}^\varepsilon,n,\mathcal{R}_\mathrm{ref},\gamma}^\text{blanket}
  =
  \mathbb E_{\substack{
    M\sim 1+\mathrm{Bin}(n-1,\gamma)\\
    Y_{1:n}\stackrel{\mathrm{i.i.d.}}{\sim}\mathcal R_{\mathrm{ref}}
  }}
  \Bigl[
    l_\varepsilon(Y_1)\,
    \mathbf 1\Bigl\{
      \sum_{i=1}^M l_\varepsilon(Y_i) > 0
    \Bigr\}
  \Bigr].
\end{equation}

\medskip\noindent
\textbf{Step 2: Change of measure on $Y_1$.}
By definition of $l_\varepsilon$,
\[
  l_\varepsilon(y)\,\mathcal R_{\mathrm{ref}}(y)
  = \mathcal R_{x_1}(y) - \mathrm e^\varepsilon \mathcal R_{x_1'}(y).
\]
Hence, for any bounded measurable function $h$ (possibly depending on
other random variables),
\begin{align*}
  \mathbb E_{Y_1\sim\mathcal R_{\mathrm{ref}}}
    \bigl[l_\varepsilon(Y_1)\,h(Y_1)\bigr]
  &= \int h(y)\,l_\varepsilon(y)\,\mathcal R_{\mathrm{ref}}(y)\,dy \\
  &= \int h(y)\,\bigl(\mathcal R_{x_1}(y)-\mathrm e^\varepsilon \mathcal R_{x_1'}(y)\bigr)\,dy \\
  &= \mathbb E_{Y_1\sim\mathcal R_{x_1}}[h(Y_1)]
     - \mathrm e^\varepsilon
       \mathbb E_{Y_1\sim\mathcal R_{x_1'}}[h(Y_1)].
\end{align*}
In \eqref{eq:blanket-as-expectation-X1}, take
\[
  h(Y_1)
  := \mathbf 1\Bigl\{
       \sum_{i=1}^M l_\varepsilon(Y_i) > 0
     \Bigr\},
\]
and treat $M$ and $Y_{2:n}$ as part of the outer expectation.
Then
\begin{align*}
  \mathcal{D}_{\mathrm{e}^\varepsilon,n,\mathcal{R}_\mathrm{ref},\gamma}^\text{blanket}
  &=
  \mathbb E_{M,Y_{2:n}}\Bigl[
    \mathbb E_{Y_1\sim\mathcal R_{\mathrm{ref}}}
    \bigl[
      l_\varepsilon(Y_1)\,
      \mathbf 1\{\sum_{i=1}^M l_\varepsilon(Y_i)>0\}
      \bigm| M,Y_{2:n}
    \bigr]
  \Bigr] \\
  &=
  \mathbb E_{M,Y_{2:n}}\Bigl[
    \Pr_{\substack{Y_1\sim\mathcal R_{x_1}}}
      \Bigl[
        \sum_{i=1}^M l_\varepsilon(Y_i)>0
      \Bigm| M,Y_{2:n}
      \Bigr]
    - \mathrm e^\varepsilon
      \Pr_{\substack{Y_1\sim\mathcal R_{x_1'}}}
      \Bigl[
        \sum_{i=1}^M l_\varepsilon(Y_i)>0
      \Bigm| M,Y_{2:n}
      \Bigr]
  \Bigr] \\
  &=
  \Pr_{\substack{
    Y_1\sim\mathcal R_{x_1}\\
    Y_{2:n}\stackrel{\mathrm{i.i.d.}}{\sim}\mathcal R_{\mathrm{ref}}\\
    M\sim 1+\mathrm{Bin}(n-1,\gamma)
  }}
  \Bigl[
    \sum_{i=1}^M l_\varepsilon(Y_i)>0
  \Bigr]
  - \mathrm e^\varepsilon
  \Pr_{\substack{
    Y_1\sim\mathcal R_{x_1'}\\
    Y_{2:n}\stackrel{\mathrm{i.i.d.}}{\sim}\mathcal R_{\mathrm{ref}}\\
    M\sim 1+\mathrm{Bin}(n-1,\gamma)
  }}
  \Bigl[
    \sum_{i=1}^M l_\varepsilon(Y_i)>0
  \Bigr].
\end{align*}
This is exactly the desired expression.
\end{proof}

\subsection{The proof of Theorem~\ref{thm:fft-blanket-error}}
\label{app:proof-fft-blanket-error}

\fftblanketerror*

\begin{proof}
This proof is an adaptation of the FFT error analysis of Gopi et al.~\cite{gopiNumericalCompositionDifferential2021} to our blanket-divergence setting.

Fix $n\in\mathbb N$, $\gamma\in(0,1]$, $\varepsilon\in\mathbb R$, and distinct inputs
$x_1,x_1'\in\mathcal X$.
Let $M\sim 1+\mathrm{Bin}(n-1,\gamma)$ and
$
  Y_{2:n} \stackrel{\mathrm{i.i.d.}}{\sim} \mathcal R_{\mathrm{ref}},
$
all independent of $M$.
Define
\begin{equation}
  \label{eq:px-def}
  p(x)
  :=
  \Pr_{Y_1\sim \mathcal{R}_{x}}\Bigl[l_\varepsilon(Y_1;x_1,x_1',\mathcal R_{\mathrm{ref}})+\sum_{i=2}^M l_\varepsilon(Y_i;x_1,x_1',\mathcal R_{\mathrm{ref}}) > 0\Bigr].
\end{equation}
We write $l(Y_i)$ as a shorthand for
$l_\varepsilon(Y_i;x_1,x_1',\mathcal R_{\mathrm{ref}})$ in the following.
From the transformation lemma (i.e., Lemma~\ref{lemma:trans_fft}), the blanket divergence can be written as
\[
  D_{\mathrm{e}^\varepsilon,n,\mathcal{R}_\mathrm{ref},\gamma}^\mathrm{blanket}(\mathcal{R}_{x_1} \| \mathcal{R}_{x_1^\prime})
  \;=\;
  p(x_1) - \mathrm e^\varepsilon p(x_1').
\]

Our goal is to compare each $p(x)$ with the corresponding main-term approximation
$P(c,F_x)$ output by the algorithms, and then combine the two inequalities.

\medskip

\textbf{Step 1: Truncation error.}
Fix constants $s, w^{\mathrm{in}}$ and,
for $i\ge2$, let
\[
  Z_i^{\mathrm{tr}}
  \sim l(Y_i)
  \ \text{conditioned on}\
  l(Y_i)\in\bigl[s, s + w^{\mathrm{in}}\bigr].
\]
Then for each $x\in\{x_1,x_1'\}$ we have
\begin{align}
  p(x)
  &= \Pr\Bigl[l(Y_1)+\textstyle\sum_{i=2}^M l(Y_i)>0\Bigr]
  \notag\\
  &\le
  \Pr\Bigl[
    l(Y_1)+\sum_{i=2}^M Z_i^{\mathrm{tr}} > 0
  \Bigr]
  + \Pr\bigl[\perp\bigr],
  \label{eq:trunc-upper}
  \\
  p(x)
  &\ge
  \Pr\Bigl[
    l(Y_1)+\sum_{i=2}^M Z_i^{\mathrm{tr}} > 0
  \Bigr]
  - \Pr\bigl[\perp\bigr].
  \label{eq:trunc-lower}
\end{align}
We abbreviate the truncation error by
\[
  E_{\mathrm{trunc}}(x)
  :=
  \Pr\bigl[\perp\bigr].
\]

\medskip

\textbf{Step 2: Discretization error.}
Fix a bin width $h>0$, and for $i\ge2$ let $Z_i^{\mathrm{di}}$ be a discretized version of
$Z_i^{\mathrm{tr}}$ on the grid $\mathcal G = \{x_j=jh:j\in\mathbb Z\}$ such that
$|Z_i^{\mathrm{tr}}-Z_i^{\mathrm{di}}|\le h/2$ almost surely.
Let
\[
  \Delta_{\mathrm{ES}} := \mathbb E\Bigl[\sum_{i=2}^M Z_i^{\mathrm{tr}} - \sum_{i=2}^M Z_i^{\mathrm{di}}\Bigr].
\]
Then using the union bound for any $c\ge0$ we can write
\begin{align*}
  &\Pr\Bigl[l(Y_1)+\sum_{i=2}^M Z_i^{\mathrm{tr}} > 0\Bigr] \\
  &= \Pr\Bigl[
    \sum_{i=2}^M Z_i^{\mathrm{tr}} > -l(Y_1)
  \Bigr] \\
  &= \Pr\Bigl[
    \sum_{i=2}^M Z_i^{\mathrm{tr}} - \sum_{i=2}^M Z_i^{\mathrm{di}}
    + \sum_{i=2}^M Z_i^{\mathrm{di}}
    - \Delta_{\mathrm{ES}}
    > -l(Y_1)-\Delta_{\mathrm{ES}}
  \Bigr] \\
  &\le
  \Pr\Bigl[
    \sum_{i=2}^M Z_i^{\mathrm{tr}} - \sum_{i=2}^M Z_i^{\mathrm{di}} - \Delta_{\mathrm{ES}} > c
  \Bigr]
  +
  \Pr\Bigl[
    \sum_{i=2}^M Z_i^{\mathrm{di}} > -c - l(Y_1) - \Delta_{\mathrm{ES}}
  \Bigr].
\end{align*}
We therefore define the discretization error
\[
  E_{\mathrm{disc}}(x;c,h,s,w^{\mathrm{in}})
  :=
  \Pr\Bigl[
    \sum_{i=2}^M Z_i^{\mathrm{tr}} - \sum_{i=2}^M Z_i^{\mathrm{di}} - \Delta_{\mathrm{ES}} > c
  \Bigr]
\]
and obtain the upper bound
\begin{equation}
  \label{eq:disc-upper}
  \Pr\Bigl[l(Y_1)+\sum_{i=2}^M Z_i^{\mathrm{tr}} > 0\Bigr]
  \;\le\;
  E_{\mathrm{disc}}(x;c,h,s,w^{\mathrm{in}})
  +
  \Pr\Bigl[
    \sum_{i=2}^M Z_i^{\mathrm{di}} > -c - l(Y_1) - \Delta_{\mathrm{ES}}
  \Bigr].
\end{equation}
A symmetric argument, using a window shifted by $-c$, yields a corresponding lower bound
\begin{equation}
  \label{eq:disc-lower}
  \Pr\Bigl[l(Y_1)+\sum_{i=2}^M Z_i^{\mathrm{tr}} > 0\Bigr]
  \;\ge\;
  -E_{\mathrm{disc}}(x;c,h,s,w^{\mathrm{in}})
  +
  \Pr\Bigl[
    \sum_{i=2}^M Z_i^{\mathrm{di}} > +c - l(Y_1) - \Delta_{\mathrm{ES}}
  \Bigr].
\end{equation}

\medskip

\textbf{Step 3: Aliasing error and FFT main term.}
Fix an outer window size $w^{\mathrm{out}}>0$ and define
$[-L^{\mathrm{out}},L^{\mathrm{out}}] := [-w^{\mathrm{out}}/2,w^{\mathrm{out}}/2]$.
We extend the PMF of the discretized variable $Z_i^{\mathrm{di}}$ by zero outside
$[-L^{\mathrm{out}},L^{\mathrm{out}}]$ and use the FFT on a grid of size $N$ such that
the period $P=2L^{\mathrm{out}}$ satisfies $P/m=h$ for some integer $m$.
This exactly corresponds to computing the distribution of the random sum
\[
  S
  :=
  \sum_{i=2}^M \widetilde Z_i^{\mathrm{di}},
\]
where $\widetilde Z_i^{\mathrm{di}}$ is the zero-padded periodic extension of $Z_i^{\mathrm{di}}$.
The difference between $S$ and the true discretized sum
$\sum_{i=2}^M Z_i^{\mathrm{di}}$ is the aliasing error; we define
\[
  E_{\mathrm{alias}}(x;h,s,w^{\mathrm{in}},w^{\mathrm{out}})
  :=
  \Pr\Bigl[\sum_{i=2}^M Z_i^{\mathrm{di}} \neq S\Bigr].
\]

Then using the union bound for the term appearing in~\eqref{eq:disc-upper} we have
\begin{align}
  &\Pr\Bigl[
    \sum_{i=2}^M Z_i^{\mathrm{di}} > -c - l(Y_1) - \Delta_{\mathrm{ES}}
  \Bigr]
  \notag\\
  &\le
  \Pr\Bigl[\sum_{i=2}^M Z_i^{\mathrm{di}} \neq S\Bigr]
  +
  \Pr\bigl[
    S > -c - l(Y_1) - \Delta_{\mathrm{ES}}
  \bigr]
  \notag\\
  &=
  E_{\mathrm{alias}}(x;h,s,w^{\mathrm{in}},w^{\mathrm{out}})
  +
  \sum_z \Pr[S=z]\,
  \Pr\bigl[l(Y_1) > -c - z - \mathbb E[\textstyle\sum_{i=2}^M Z_i^{\mathrm{tr}}]\bigr],
  \label{eq:main-term-raw}
\end{align}
where in the last equality we have simply re-centered the threshold by the mean of
$\sum_{i=2}^M Z_i^{\mathrm{tr}}$, which only affects $\Delta_{\mathrm{ES}}$.

Let $F_x$ be the CDF of $l(Y)$ under $Y\sim\mathcal R_x$.
Then the inner probability in~\eqref{eq:main-term-raw} is
\[
  \Pr\bigl[l(Y_1) > -c - z - \mathbb E[\textstyle\sum_{i=2}^M Z_i^{\mathrm{tr}}]\bigr]
  =
  1 - F_x\bigl(-c - z - \mathbb E[\textstyle\sum_{i=2}^M Z_i^{\mathrm{tr}}]\bigr).
\]
By construction, the algorithm \textsc{CalculatePMF} computes the PMF
$\mathbf p_{S-\mu_{S^\mathrm{di}}}$ of the centered sum $S-\mu_{S^\mathrm{di}}$, where
\[
  \mu_{S^\mathrm{di}} := (n-1)\gamma\,\mu_{Z^{\mathrm{di}}},
\]
and \textsc{CalculateMainTerm} computes
\[
  \widehat p_c(x)
  :=
  \textsc{CalculateMainTerm}(\mathbf p_{S-\mu_{S^\mathrm{di}}},F_x,c)
  =
  \sum_j \mathbf p_{S-\mu_{S^\mathrm{di}}}[j]\,
    \Bigl(1-F_x\bigl(-c-z_j-\mathbb E[\textstyle\sum_{i=2}^M Z_i^{\mathrm{tr}}]\bigr)\Bigr),
\]
which coincides with the right-hand probability in~\eqref{eq:main-term-raw}.
Thus, combining \eqref{eq:disc-upper} and \eqref{eq:main-term-raw} with the truncation step
\eqref{eq:trunc-upper}, we obtain
\begin{align}
  p(x)
  &\le
  \widehat p_c(x)
  + E_{\mathrm{alias}}(x;h,s,w^{\mathrm{in}},w^{\mathrm{out}})
  + E_{\mathrm{disc}}(x;c,h,s,w^{\mathrm{in}})
  + E_{\mathrm{trunc}}(x).
  \label{eq:p-upper-hat}
\end{align}

A completely symmetric argument, using $-c$ instead of $+c$ in the union bounds,
gives a second quantity
$\widehat p_{-c}(x)=\textsc{CalculateMainTerm}(\mathbf p_{S-\mu_S},F_x,-c)$ such that
\begin{align}
  p(x)
  &\ge
  \widehat p_{-c}(x)
  - E_{\mathrm{alias}}(x;h,s,w^{\mathrm{in}},w^{\mathrm{out}})
  - E_{\mathrm{disc}}(x;c,h,s,w^{\mathrm{in}})
  - E_{\mathrm{trunc}}(x).
  \label{eq:p-lower-hat}
\end{align}

\medskip

\textbf{Step 4: Incorporating the factor $(1-\Pr[\perp])$.}
In the theorem statement, the main term is defined as
\[
  P(c,F_x)
  :=
  (1-\Pr[\perp])\,
  \widehat p_c(x),
  \qquad
  P(-c,F_x)
  :=
  (1-\Pr[\perp])\,
  \widehat p_{-c}(x),
\]
where $\Pr[\perp]=E_{\mathrm{trunc}}(x)$ is the truncation probability (note that $\perp$
depends only on $(Y_i)_{i\ge2}$, hence is the same for $x_1$ and $x_1'$).

Since $0\le\widehat p_{\pm c}(x)\le1$, we have
\[
  \widehat p_c(x) - E_{\mathrm{trunc}}(x)
  \;\le\;
  P(c,F_x)
  \;\le\;
  \widehat p_c(x),
\]
and similarly for $-c$.
Therefore, from~\eqref{eq:p-upper-hat} we obtain
\begin{align}
  p(x)
  &\le
  P(c,F_x)
  + E_{\mathrm{alias}}(x;h,s,w^{\mathrm{in}},w^{\mathrm{out}})
  + E_{\mathrm{disc}}(x;c,h,s,w^{\mathrm{in}})
  + 2E_{\mathrm{trunc}}(x),
  \label{eq:p-upper-final}
\end{align}
and from~\eqref{eq:p-lower-hat}
\begin{align}
  p(x)
  &\ge
  P(-c,F_x)
  - E_{\mathrm{alias}}(x;h,s,w^{\mathrm{in}},w^{\mathrm{out}})
  - E_{\mathrm{disc}}(x;c,h,s,w^{\mathrm{in}})
  - 2E_{\mathrm{trunc}}(x).
  \label{eq:p-lower-final}
\end{align}

\medskip

\textbf{Step 5: Combining the two hypotheses.}
Apply \eqref{eq:p-lower-final} with $x=x_1$ and \eqref{eq:p-upper-final} with $x=x_1'$:
\begin{align*}
  p(x_1)
  &\ge
  P(-c,F_{x_1})
  - E_{\mathrm{alias}}(x_1)
  - E_{\mathrm{disc}}(x_1)
  - 2E_{\mathrm{trunc}}(x_1), \\
  p(x_1')
  &\le
  P(c,F_{x_1'})
  + E_{\mathrm{alias}}(x_1')
  + E_{\mathrm{disc}}(x_1')
  + 2E_{\mathrm{trunc}}(x_1'),
\end{align*}
where we have suppressed the explicit dependence on $(c,h,s,w^{\mathrm{in}},w^{\mathrm{out}})$
in the notation of the error terms.

Hence
\begin{align*}
  D_{\mathrm{e}^\varepsilon,n,\mathcal{R}_\mathrm{ref},\gamma}^\mathrm{blanket}(\mathcal{R}_{x_1} \| \mathcal{R}_{x_1^\prime})
  &= p(x_1) - \mathrm e^\varepsilon p(x_1') \\
  &\ge
  P(-c,F_{x_1}) - \mathrm e^\varepsilon P(c,F_{x_1'})
  - \delta_{\mathrm{err}}^{\mathrm{low}}(c,h,s,w^{\mathrm{in}},w^{\mathrm{out}}),
\end{align*}
where
\[
  \delta_{\mathrm{err}}^{\mathrm{low}}
  :=
  E_{\mathrm{alias}}(x_1)
  + E_{\mathrm{disc}}(x_1)
  + 2E_{\mathrm{trunc}}(x_1)
  + \mathrm e^\varepsilon\bigl(
    E_{\mathrm{alias}}(x_1')
    + E_{\mathrm{disc}}(x_1')
    + 2E_{\mathrm{trunc}}(x_1')
  \bigr).
\]

Similarly, apply \eqref{eq:p-upper-final} with $x=x_1$ and \eqref{eq:p-lower-final} with $x=x_1'$:
\begin{align*}
  p(x_1)
  &\le
  P(c,F_{x_1})
  + E_{\mathrm{alias}}(x_1)
  + E_{\mathrm{disc}}(x_1)
  + 2E_{\mathrm{trunc}}(x_1), \\
  p(x_1')
  &\ge
  P(-c,F_{x_1'})
  - E_{\mathrm{alias}}(x_1')
  - E_{\mathrm{disc}}(x_1')
  - 2E_{\mathrm{trunc}}(x_1'),
\end{align*}
so that
\begin{align*}
  D_{\mathrm{e}^\varepsilon,n,\mathcal{R}_\mathrm{ref},\gamma}^\mathrm{blanket}(\mathcal{R}_{x_1} \| \mathcal{R}_{x_1^\prime})
  &= p(x_1) - \mathrm e^\varepsilon p(x_1') \\
  &\le
  P(c,F_{x_1}) - \mathrm e^\varepsilon P(-c,F_{x_1'})
  + \delta_{\mathrm{err}}^{\mathrm{up}}(c,h,s,w^{\mathrm{in}},w^{\mathrm{out}}),
\end{align*}
where
\[
  \delta_{\mathrm{err}}^{\mathrm{up}}
  :=
  E_{\mathrm{alias}}(x_1)
  + E_{\mathrm{disc}}(x_1)
  + 2E_{\mathrm{trunc}}(x_1)
  + \mathrm e^\varepsilon\bigl(
    E_{\mathrm{alias}}(x_1')
    + E_{\mathrm{disc}}(x_1')
    + 2E_{\mathrm{trunc}}(x_1')
  \bigr).
\]

By construction, both $\delta_{\mathrm{err}}^{\mathrm{low}}$ and
$\delta_{\mathrm{err}}^{\mathrm{up}}$ are nonnegative and depend only on the numerical
parameters $(c,h,s,w^{\mathrm{in}},w^{\mathrm{out}})$ and on the distributions
$\mathcal R_{\mathrm{ref}},\mathcal R_{x_1},\mathcal R_{x_1'}$ through the truncation,
discretization, and aliasing probabilities.

See Section~\ref{sec:fft-implementation-details} for a discussion of how to rigorously control these error terms in practice.

This proves the desired two-sided bounds.
\end{proof}

\subsection{The proof of Theorem~\ref{thm:fft-relative-error}}
\label{app:proof-fft-relative-error}

\fftrelativeerror*

\begin{proof}

Since $\mathcal{D}^{\mathrm{blanket}} \in [L_n,U_n]$ from Theorem~\ref{thm:fft-blanket-error}, we obtain
\begin{align}
  \bigl|\widehat D_n - \mathcal{D}^{\mathrm{blanket}}\bigr|
  &\le \frac{U_n - L_n}{2} \notag\\
  &= \frac{1}{2}
     \Bigl(
       P(c;x_1) - P(-c;x_1)
       + \mathrm e^\varepsilon\bigl\{P(c;x_1') - P(-c;x_1')\bigr\}
     \Bigr)
     + \frac{\delta^{\mathrm{up}}_\mathrm{err} + \delta^{\mathrm{low}}_\mathrm{err}}{2} \notag\\
  &\le \frac{1}{2}\bigl|P(c;x_1) - P(-c;x_1)\bigr|
     + \frac{\mathrm e^\varepsilon}{2}\bigl|P(c;x_1') - P(-c;x_1')\bigr|
     + \frac{\delta^{\mathrm{up}}_\mathrm{err} + \delta^{\mathrm{low}}_\mathrm{err}}{2}.
  \label{eq:midpoint-error}
\end{align}
In particular, the total approximation error decomposes into a main
term, given by the differences $P(c;\cdot)-P(-c;\cdot)$, and the
numerical error terms $\delta^{\mathrm{up}}_\mathrm{err}$ and $\delta^{\mathrm{low}}_\mathrm{err}$
(truncation, discretization, and aliasing).
We bound each term in~\eqref{eq:midpoint-error} relative to
the true blanket divergence $\mathcal{D}^{\mathrm{blanket}}$.

We first consider the main term.
Recall that $M \sim 1+\mathrm{Bin}(n-1,\gamma)$ and
$Y_{2:n} \stackrel{\mathrm{i.i.d.}}{\sim} \mathcal R_{\mathrm{ref}}$, independent of $M$.
Let
\[
  S_{n-1}
  :=
  \sum_{i=2}^M l_\varepsilon\bigl(Y_i;x_1,x_1',\mathcal R_{\mathrm{ref}}\bigr),
\]
and let $\widetilde S_{n-1}$ denote the truncated, discretized, and
periodized version of $S_{n-1}$ that is actually used in the FFT
computation.
Moreover, for $x\in\{x_1,x_1'\}$ we write
\[
  Z_1 := l_\varepsilon\bigl(Y_1;x_1,x_1',\mathcal R_{\mathrm{ref}}\bigr),
  \qquad Y_1 \sim \mathcal R_x,
\]
and $\Delta_{\mathrm{ES}}$ is defined as in the proof of
Theorem~\ref{thm:fft-blanket-error} (i.e., $\mu_{S^{\mathrm{tr}}}-\mu_{S^{\mathrm{di}}}$).
With this notation, the main term appearing in~\eqref{eq:midpoint-error}
can be expressed as
\begin{align}
  P(c;x) - P(-c;x)
  &= (1-\Pr[\perp])\,
     \mathbb E_{Y_1\sim\mathcal R_x}\Bigl[
       \Pr\bigl[-c
                < \widetilde S_{n-1} + Z_1 + \Delta_{\mathrm{ES}}
                \le c
                \,\big|\,
                Z_1
          \bigr]
     \Bigr],
  \label{eq:main-band-prob}
\end{align}
where $\perp$ is the truncation event defined in
Theorem~\ref{thm:fft-blanket-error}.
In other words, the quantity $P(c;x)-P(-c;x)$ is precisely the probability
that the random sum $\widetilde S_{n-1}+Z_1+\Delta_{\mathrm{ES}}$ falls
inside the band $(-c,c)$, up to the multiplicative factor $(1-\Pr[\perp])$.

The approximating sum $\widetilde S_{n-1}$ is somewhat inconvenient to
analyze directly, so we next replace it by the true sum $S_{n-1}$ at the
price of an additional error term.
Let
\[
  \widetilde S_{n-1}
  = S_{n-1} + E_{n-1},
  \qquad
  E_{n-1} := \widetilde S_{n-1} - S_{n-1},
\]
where $E_{n-1}$ collects the truncation, discretization, and aliasing
errors.
Fix $c>0$ and define the good event
\[
  G := \bigl\{|E_{n-1}| \le c/2\bigr\},
  \qquad
  G^c := \bigl\{|E_{n-1}| > c/2\bigr\}.
\]
Then, for any realization of $Z_1$ and on the event $G$, the inclusion
\[
  \bigl\{
    -c < \widetilde S_{n-1} + Z_1 + \Delta_{\mathrm{ES}} \le c
  \bigr\}
  \subset
  \bigl\{
    |S_{n-1} + Z_1 + \Delta_{\mathrm{ES}}|
    < c + |E_{n-1}|
  \bigr\}
  \subset
  \bigl\{
    |S_{n-1} + Z_1 + \Delta_{\mathrm{ES}}|
    < \tfrac{3}{2}c
  \bigr\}
\]
holds.
Consequently,
\begin{align}
  &\Pr\bigl[
    -c < \widetilde S_{n-1} + Z_1 + \Delta_{\mathrm{ES}} \le c
    \,\big|\,
    Z_1
  \bigr]
  \notag\\
  &\qquad\le
  \Pr\bigl[
    |S_{n-1} + Z_1 + \Delta_{\mathrm{ES}}|
    < \tfrac{3}{2}c
    \,\big|\,
    Z_1
  \bigr]
  +
  \Pr\bigl[|E_{n-1}|>c/2\bigr].
  \label{eq:band-replace-S}
\end{align}
Inserting this bound into~\eqref{eq:main-band-prob} and taking the
expectation over $Z_1$ yields
\begin{align}
  |P(c;x) - P(-c;x)|
  &\le
  \mathbb E_{Y_1\sim\mathcal R_x}\Bigl[
    \Pr\bigl[
      |S_{n-1} + Z_1 + \Delta_{\mathrm{ES}}|
      < \tfrac{3}{2}c
      \,\big|\,
      Z_1
    \bigr]
  \Bigr]
  +
  \Pr\bigl[|E_{n-1}|>c/2\bigr].
  \label{eq:main-term-S-plus-error}
\end{align}
Here, we used $1-\Pr[\perp] \le 1$.

Next we further decompose the numerical error term
\(\Pr[|E_{n-1}|>c/2]\) into truncation, discretization, and wrap-around
(aliasing) contributions.

Recall that
\[
  S_{n-1}
  :=
  \sum_{i=2}^M l_\varepsilon\bigl(Y_i;x_1,x_1',\mathcal R_{\mathrm{ref}}\bigr)
\]
is the true sum, and that \(\widetilde S_{n-1}\) denotes the truncated,
discretized, and periodized version of \(S_{n-1}\) used in the FFT
computation. As before we write
\[
  \widetilde S_{n-1}
  = S_{n-1} + E_{n-1},
  \qquad
  E_{n-1} := \widetilde S_{n-1} - S_{n-1}.
\]

We view the construction of \(\widetilde S_{n-1}\) as a sequence of three
approximations:

\begin{itemize}
  \item \emph{Truncation:}
  \[
    S_{n-1}^{\mathrm{tr}}
    :=
    \sum_{i=2}^M Z_i^{\mathrm{tr}},
  \]
  where $Z_i^{\mathrm{tr}}$ is coupled with $l(Y_i)$ such that $Z_i^{\mathrm{tr}} = l(Y_i)$ whenever $l(Y_i) \in [s, s+w^{\mathrm{in}}]$.

  \item \emph{Discretization:}
  \[
    S_{n-1}^{\mathrm{di}}
    :=
    \sum_{i=2}^M Z_i^{\mathrm{di}},
  \]
  where \(Z_i^{\mathrm{di}}\) is the discretized version of
  \(Z_i^{\mathrm{tr}}\) on the grid \(\mathcal G = \{x_j=jh:j\in\mathbb Z\}\)
  with \(|Z_i^{\mathrm{tr}}-Z_i^{\mathrm{di}}|\le h/2\) almost surely.

  \item \emph{Wrap-around (aliasing):}
  \(\widetilde S_{n-1}\) is obtained from \(S_{n-1}^{\mathrm{di}}\) by
  periodic extension outside the outer window
  \([-w^{\mathrm{out}}/2,w^{\mathrm{out}}/2]\).
\end{itemize}

Accordingly, we decompose the total error into three components:
\begin{align*}
  E_{\mathrm{Mtrunc}}
  &:= S_{n-1}^{\mathrm{tr}} - S_{n-1},
  \\
  E_{\mathrm{Mdisc}}
  &:= S_{n-1}^{\mathrm{di}} - S_{n-1}^{\mathrm{tr}},
  \\
  E_{\mathrm{Malias}}
  &:= \widetilde S_{n-1} - S_{n-1}^{\mathrm{di}},
\end{align*}
so that
\[
  E_{n-1}
  = E_{\mathrm{Mtrunc}}
    + E_{\mathrm{Mdisc}}
    + E_{\mathrm{Malias}}.
\]

By the triangle inequality we have
\[
  |E_{n-1}|
  \le
  |E_{\mathrm{Mtrunc}}|
  + |E_{\mathrm{Mdisc}}|
  + |E_{\mathrm{Malias}}|.
\]
Hence, for any \(c>0\),
\[
  \Bigl\{|E_{n-1}|>\tfrac{c}{2}\Bigr\}
  \subset
  \Bigl\{|E_{\mathrm{Mtrunc}}|>\tfrac{c}{6}\Bigr\}
  \,\cup\,
  \Bigl\{|E_{\mathrm{Mdisc}}|>\tfrac{c}{6}\Bigr\}
  \,\cup\,
  \Bigl\{|E_{\mathrm{Malias}}|>\tfrac{c}{6}\Bigr\},
\]
and by the union bound this yields
\begin{equation}
  \label{eq:E-decomposition}
  \Pr\bigl[|E_{n-1}|>c/2\bigr]
  \le
  \Pr\bigl[|E_{\mathrm{Mtrunc}}|>c/6\bigr]
  +
  \Pr\bigl[|E_{\mathrm{Mdisc}}|>c/6\bigr]
  +
  \Pr\bigl[|E_{\mathrm{Malias}}|>c/6\bigr].
\end{equation}
This decomposes the tail probability of the total numerical error into
three contributions corresponding to truncation, discretization, and
wrap-around, respectively.
In particular, since the error terms
\(\delta_{\mathrm{err}}^{\mathrm{low}}\) and
\(\delta_{\mathrm{err}}^{\mathrm{up}}\) in
Theorem~\ref{thm:fft-blanket-error} are also built from these three errors, we evaluate each of them separately in the sequel.

Recall that the error terms in Theorem~\ref{thm:fft-blanket-error} are
given by
\[
  \delta_{\mathrm{err}}^{\mathrm{low}}=
  \delta_{\mathrm{err}}^{\mathrm{up}}:=
  E_{\mathrm{alias}}(x_1)
  + E_{\mathrm{disc}}(x_1)
  + 2E_{\mathrm{trunc}}(x_1)
  + \mathrm e^\varepsilon\bigl(
    E_{\mathrm{alias}}(x_1')
    + E_{\mathrm{disc}}(x_1')
    + 2E_{\mathrm{trunc}}(x_1')
  \bigr),
\]
where, for brevity, we suppress the dependence of
\(E_{\mathrm{alias}},E_{\mathrm{disc}},E_{\mathrm{trunc}}\) on
\((c,h,s,w^{\mathrm{in}},w^{\mathrm{out}})\).
Thus \(\delta_{\mathrm{err}}^{\mathrm{low}}=\delta_{\mathrm{err}}^{\mathrm{up}}\), and
\begin{align*}
  \delta_{\mathrm{trunc}}
  &:=
  2E_{\mathrm{trunc}}(x_1)
  + 2\mathrm e^\varepsilon E_{\mathrm{trunc}}(x_1'),
  \\
  \delta_{\mathrm{disc}}
  &:=
  E_{\mathrm{disc}}(x_1)
  + \mathrm e^\varepsilon E_{\mathrm{disc}}(x_1'),
  \\
  \delta_{\mathrm{alias}}
  &:=
  E_{\mathrm{alias}}(x_1)
  + \mathrm e^\varepsilon E_{\mathrm{alias}}(x_1').
\end{align*}

Combining this with \eqref{eq:main-term-S-plus-error} and
\eqref{eq:midpoint-error}, and using that the bound on
\(\Pr[|E_{n-1}|>c/2]\) does not depend on \(x\), we arrive at
\begin{align}
  \bigl|\widehat D_n - \mathcal{D}^{\mathrm{blanket}}\bigr|
  &\le
  \frac{1}{2}\,
  \mathbb E_{Y_1\sim\mathcal R_{x_1}}\Bigl[
    \Pr\bigl[
      |S_{n-1} + Z_1 + \Delta_{\mathrm{ES}}|
      < \tfrac{3}{2}c
      \,\big|\,
      Z_1
    \bigr]
  \Bigr]
  \notag\\
  &\quad
  +
  \frac{\mathrm e^\varepsilon}{2}\,
  \mathbb E_{Y_1\sim\mathcal R_{x_1'}}\Bigl[
    \Pr\bigl[
      |S_{n-1} + Z_1 + \Delta_{\mathrm{ES}}|
      < \tfrac{3}{2}c
      \,\big|\,
      Z_1
    \bigr]
  \Bigr]
  \notag\\
  &\quad
  + \frac{1+\mathrm e^\varepsilon}{2}\,
    \Pr\bigl[|E_{\mathrm{Mtrunc}}|>c/6\bigr]
  + \delta_{\mathrm{trunc}}
  \notag\\
  &\quad
  + \frac{1+\mathrm e^\varepsilon}{2}\,
    \Pr\bigl[|E_{\mathrm{Mdisc}}|>c/6\bigr]
  + \delta_{\mathrm{disc}}
  \notag\\
  &\quad
  + \frac{1+\mathrm e^\varepsilon}{2}\,
    \Pr\bigl[|E_{\mathrm{Malias}}|>c/6\bigr]
  + \delta_{\mathrm{alias}}.
  \label{eq:full-error-by-type}
\end{align}

In other words, the total approximation error splits into a \emph{main
band} term (the first two lines of \eqref{eq:full-error-by-type}) plus
three groups of numerical error:
\begin{itemize}
  \item truncation error:
    \(\delta_{\mathrm{trunc}}\) and
    \(\Pr[|E_{\mathrm{Mtrunc}}|>c/6]\),
  \item discretization error:
    \(\delta_{\mathrm{disc}}\) and
    \(\Pr[|E_{\mathrm{Mdisc}}|>c/6]\),
  \item aliasing (wrap-around) error:
    \(\delta_{\mathrm{alias}}\) and
    \(\Pr[|E_{\mathrm{Malias}}|>c/6]\).
\end{itemize}
Thus each of the three numerical effects (truncation, discretization,
and wrap-around) contributes via a \(\delta\)-part (coming from
Theorem~\ref{thm:fft-blanket-error}) and an \(E_{\mathrm{M}}\)-part
coming from the replacement of \(\widetilde S_{n-1}\) by the true sum
\(S_{n-1}\) in the main term.

\textbf{Step 1: Main-term error and choice of $c_n$.}
We first evaluate the main term using the following lemma.

\begin{restatable}[Band probability]{lemma}{bandprobability}
\label{lem:band-prob-fixed-x-bigO}
Assume the setting of Lemma~\ref{lem:blanket-asymptotics-general}, and suppose that the local
randomizer $\mathcal R$ satisfies Assumption~\ref{assump:regularity}.

Let $(c_n)_{n\ge1}$ be a sequence with $c_n>0$ such that
\[
  \frac{c_n}{\sigma_{\varepsilon_n}\sqrt{n-1}} \longrightarrow 0,
\]
and suppose that
$
  \Delta_{\mathrm{ES}}=O\left(\frac{\sigma_{\varepsilon_n}\sqrt{n}}{\sqrt{\log n}}\right).
$
Write
$
  t_n := -\frac{\mu_n\sqrt{n-1}}{\sigma_{\varepsilon_n}}.
$
Then there exists a constant $C<\infty$, depending only on the
constants in Assumption~\ref{assump:regularity}, such that for all
sufficiently large $n$ and for every $x\in\{x_1,x_1'\}$,
\begin{equation}
  \label{eq:band-prob-bigO}
  \mathbb E_{Z_1\sim\mathcal R_x}\Bigl[
    \Pr\bigl[
      |S_{n-1} + Z_1 + \Delta_{\mathrm{ES}}|
      < \tfrac{3}{2}c_n
      \,\big|\,
      Z_1
    \bigr]
  \Bigr]
  \;\le\;
  C \,\frac{c_n}{\sigma_{\varepsilon_n}\sqrt{n-1}}\,
  \varphi(t_n).
\end{equation}
\end{restatable}

We prove this lemma in Section~\ref{app:proof-band-prob-fixed-x}.
We now evaluate the main-term error.
Recall from Lemma~\ref{lem:band-prob-fixed-x-bigO} that for any sequence
$c_n>0$ satisfying $c_n / (\sigma_{\varepsilon_n}\sqrt{n-1}) \to 0$ we have
for each $x \in \{x_1,x_1'\}$,
\begin{equation}
  \label{eq:band-prob-recall}
  \mathbb E_{Z_1\sim\mathcal R_x}\Bigl[
    \Pr\bigl[
      |S_{n-1} + Z_1 + \Delta_{\mathrm{ES}}|
      < \tfrac{3}{2}c_n
      \,\big|\,
      Z_1
    \bigr]
  \Bigr]
  \;\le\;
  C \,\frac{c_n}{\sigma_{\varepsilon_n}\sqrt{n-1}}\,
  \varphi(t_n),
\end{equation}
for some constant $C<\infty$ independent of $n$.
Inserting this bound into the midpoint error decomposition
\eqref{eq:full-error-by-type} yields
\begin{equation}
  \label{eq:main-term-abs-error}
  \bigl|\text{(main term in \eqref{eq:full-error-by-type})}\bigr|
  \;=\;
  O\!\left(
    \frac{c_n}{\sigma_{\varepsilon_n}\sqrt{n}}\,
    \varphi(t_n)
  \right).
\end{equation}

On the other hand, by the asymptotic expansion of the blanket divergence,
we have
\begin{equation}
  \label{eq:blanket-asymp}
  \mathcal{D}^{\mathrm{blanket}}
  \;=\;
  \varphi\bigl(t_n\bigr)\,
  \left(
    \frac{1}{\chi^3}\,
    \frac{1}{\varepsilon_n^2 n^{3/2}}
  \right)\bigl(1+o(1)\bigr),
\end{equation}
where $\chi>0$ is a constant.

Dividing \eqref{eq:main-term-abs-error} by \eqref{eq:blanket-asymp}, and
using that $\sigma_{\varepsilon_n}\to\sigma>0$ as $n\to\infty$, we obtain
\begin{align}
  \frac{\bigl|\text{(main term)}\bigr|}{\mathcal{D}^{\mathrm{blanket}}}
  &=
  O\!\left(
    \frac{\dfrac{c_n}{\sigma_{\varepsilon_n}\sqrt{n}}\varphi(t_n)}
         {\varphi(t_n)\,\dfrac{1}{\chi^3}\dfrac{1}{\varepsilon_n^2 n^{3/2}}}
  \right)
  (1+o(1))
  \notag\\[4pt]
  &=
  O\!\left(
    c_n \,\varepsilon_n^2 n \,\frac{\chi^3}{\sigma_{\varepsilon_n}}
  \right)
  (1+o(1))
  \notag\\[4pt]
  &=
  O\!\left(
    c_n \,\varepsilon_n^2 n
  \right),
  \label{eq:main-term-rel-error}
\end{align}
where the last step uses that $\chi^3/\sigma_{\varepsilon_n}$ is bounded
uniformly in $n$.

We would like the contribution of the main term to the overall relative
error to be at most $\eta_{\mathrm{main}}$, that is,
\[
  \frac{\bigl|\text{(main term)}\bigr|}{\mathcal{D}^{\mathrm{blanket}}}
  \;\lesssim\;
  \eta_{\mathrm{main}}.
\]
From \eqref{eq:main-term-rel-error}, a sufficient condition is
\begin{equation}
  \label{eq:cn-condition}
  c_n \,\varepsilon_n^2 n
  \;\lesssim\;
  \eta_{\mathrm{main}}.
\end{equation}
Equivalently, we may choose $c_n$ of order
\begin{equation}
  \label{eq:cn-choice-general}
  c_n
  \;=\;
  \Theta\!\left(
    \frac{\eta_{\mathrm{main}}}{\varepsilon_n^2 n}
  \right).
\end{equation}
With such a choice, the main-term contribution to the relative error is
$O(\eta_{\mathrm{main}})$.

Finally, we check that the technical assumption of
Lemma~\ref{lem:band-prob-fixed-x-bigO},
namely $c_n / (\sigma_{\varepsilon_n}\sqrt{n}) \to 0$, is satisfied.
Using \eqref{eq:cn-choice-general}, we have
\[
  \frac{c_n}{\sigma_{\varepsilon_n}\sqrt{n}}
  =
  \Theta\!\left(
    \frac{\eta_{\mathrm{main}}}{\varepsilon_n^2 n^{3/2}}
  \right).
\]
Under our standing assumption $\varepsilon_n^2 n \to \infty$
(e.g.\ $\varepsilon_n^2 = \Theta(\log n / n)$), the right-hand side tends
to $0$, so the condition of
Lemma~\ref{lem:band-prob-fixed-x-bigO} indeed holds.

In particular, in the canonical regime
$\varepsilon_n^2 = \Theta(\log n/n)$, we obtain the more explicit scaling
\[
  c_n
  \;=\;
  \Theta\!\left(
    \frac{\eta_{\mathrm{main}}}{\log n}
  \right).
\]

\medskip

\noindent

\textbf{Step 2: Truncation error and choice of $(s_n,w_n^{\mathrm{in}})$.}
Recall the truncation event
\[
  \perp
  :=
  \Bigl\{
    \exists i\in\{2,\dots,n\} :
    B_i = 1,\ 
    l_{\varepsilon_n}(Y_i;x_1,x_1',\mathcal R_{\mathrm{ref}})
      \notin [s_n,s_n+w_n^{\mathrm{in}}]
  \Bigr\},
\]
and write
\[
  E_{\mathrm{trunc}}(n;s_n,w_n^{\mathrm{in}})
  :=
  \Pr[\perp].
\]
Moreover, let
\[
  S_{n-1}
  :=
  \sum_{i=2}^M l_{\varepsilon_n}(Y_i;x_1,x_1',\mathcal R_{\mathrm{ref}}),
  \qquad
  S_{n-1}^{\mathrm{tr}}
  :=
  \sum_{i=2}^M Z_i^{\mathrm{tr}},
\]
where $Z_i^{\mathrm{tr}}$ are i.i.d.\ copies of the truncated random
variable used in the FFT construction, and define
\[
  E_{\mathrm{Mtrunc}}
  :=
  S_{n-1}^{\mathrm{tr}} - S_{n-1}.
\]

We control $E_{\mathrm{trunc}}$ and $\Pr[|E_{\mathrm{Mtrunc}}|>c_n/6]$
in terms of the blanket divergence.

\medskip

\textbf{Control of $E_{\mathrm{trunc}}$.}

Denote by
\[
  \tau_n := \max\{|s_n|,\ |s_n+w_n^{\mathrm{in}}|\}
\]
the truncation radius.
Assumption~\ref{assump:regularity} guarantees that $l_{\varepsilon_n}(Y)$
has finite moments of all integer orders.
In particular, for every $p\ge2$ there exists a constant
$C_p<\infty$ such that
\[
  \mathbb E\bigl[|l_{\varepsilon_n}(Y)|^p\bigr] \le C_p
  \qquad\text{for all }n.
\]
By Markov's inequality, this implies a polynomial tail bound
\begin{equation}
  \label{eq:ell-poly-tail}
  \Pr\bigl[|l_{\varepsilon_n}(Y)|>T\bigr]
  \;\le\;
  C_p\,T^{-p},
  \qquad T>0.
\end{equation}

Using the binomial structure of $M-1$ and a union bound, we obtain
\begin{align}
  E_{\mathrm{trunc}}(n;s_n,w_n^{\mathrm{in}})
  &=
  \Pr\bigl[\perp\bigr]
  \notag\\
  &\le
  \Pr\Bigl[
    \bigcup_{i=2}^n
    \bigl\{B_i=1,\ l_{\varepsilon_n}(Y_i)\notin[s_n,s_n+w_n^{\mathrm{in}}]\bigr\}
  \Bigr]
  \notag\\
  &\le
  K_\gamma\,n\,
  \Pr\bigl[|l_{\varepsilon_n}(Y)|>\tau_n\bigr]
  \notag\\
  &\le
  K_\gamma C_p\,n\,\tau_n^{-p},
  \label{eq:Etrunc-poly-bound}
\end{align}
for some constant $K_\gamma<\infty$ depending only on $\gamma$.

Dividing \eqref{eq:Etrunc-poly-bound} by the blanket divergence, we get
\begin{align}
  \frac{E_{\mathrm{trunc}}(n;s_n,w_n^{\mathrm{in}})}{\mathcal{D}^{\mathrm{blanket}}}
  &\le
  K'_{\gamma,p}\,
  n\,\tau_n^{-p}\,
  \frac{\varepsilon_n^2 n^{3/2}}{\varphi(t_n)}
  (1+o(1))
  \notag\\
  &=
  K'_{\gamma,p}\,
  \varepsilon_n^2 n^{5/2}\,
  \frac{\tau_n^{-p}}{\varphi(t_n)}
  (1+o(1)).
  \label{eq:Etrunc-over-Dblanket}
\end{align}
In the regime of interest (in particular when $\varepsilon_n^2 n$ and
$t_n$ grow at most logarithmically in $n$), the prefactor
$\varepsilon_n^2 n^{3/2}/\varphi(t_n)$ grows at most polynomially in $n$.
Hence we can choose $p$ large enough and a polynomially growing sequence
$\tau_n$ such that the right-hand side of
\eqref{eq:Etrunc-over-Dblanket} is arbitrarily small.
For concreteness, fix $p\ge4$ and set
\begin{equation}
  \label{eq:Tn-def-poly}
  \tau_n
  :=
  n^{\beta}(\log n)^{1/p}
\end{equation}
with some $\beta>0$ sufficiently large (depending on $p$ and on the
growth rate of $\varepsilon_n^2 n^{3/2}/\varphi(t_n)$).
Then there exists $n_0$ such that, for all $n\ge n_0$,
\begin{equation}
  \label{eq:Etrunc-vs-Dblanket}
  \frac{E_{\mathrm{trunc}}(n;s_n,w_n^{\mathrm{in}})}{\mathcal{D}^{\mathrm{blanket}}}
  \le
  \frac{\eta_{\mathrm{trunc}}}{4},
\end{equation}
where $\eta_{\mathrm{trunc}}>0$ is the prescribed accuracy knob.


\medskip
\medskip

\textbf{Control of $\Pr\left[|E_{\mathrm{Mtrunc}}|>c_n/6\right]$.}
Recall that $Z_i^{\mathrm{tr}}$ is defined as $l_{\varepsilon_n}(Y_i)$ conditioned on the truncation interval.
Under the natural coupling where $Z_i^{\mathrm{tr}} = l_{\varepsilon_n}(Y_i)$ whenever $l_{\varepsilon_n}(Y_i) \in [s_n, s_n+w_n^{\mathrm{in}}]$, we have
\[
  S_{n-1}^{\mathrm{tr}} = S_{n-1}
  \quad \text{whenever the event } \perp^c \text{ occurs}.
\]
Consequently, on the event $\perp^c$, the error $E_{\mathrm{Mtrunc}} = S_{n-1}^{\mathrm{tr}} - S_{n-1}$ is identically zero.
This implies the inclusion
\[
  \bigl\{|E_{\mathrm{Mtrunc}}| > c_n/6\bigr\}
  \;\subset\;
  \bigl\{E_{\mathrm{Mtrunc}} \neq 0\bigr\}
  \;\subset\;
  \perp.
\]
Therefore, regardless of the value of $c_n$, we have the simple bound
\begin{equation}
  \label{eq:EMtrunc-simple-bound}
  \Pr\bigl[|E_{\mathrm{Mtrunc}}| > c_n/6\bigr]
  \;\le\;
  \Pr[\perp]
  \;=\;
  E_{\mathrm{trunc}}(n;s_n,w_n^{\mathrm{in}}).
\end{equation}
Using the bound \eqref{eq:Etrunc-vs-Dblanket} derived in the previous step, we immediately obtain
\begin{equation}
  \label{eq:EMtrunc-vs-Dblanket}
  \Pr\bigl[|E_{\mathrm{Mtrunc}}| > c_n/6\bigr]
  \;\le\;
  \frac{\eta_{\mathrm{trunc}}}{4}\,\mathcal{D}^{\mathrm{blanket}}.
\end{equation}

\medskip

Combining \eqref{eq:Etrunc-vs-Dblanket} and \eqref{eq:EMtrunc-vs-Dblanket}, and recalling that $\delta_{\mathrm{trunc}}$ in \eqref{eq:full-error-by-type} is a linear combination of $E_{\mathrm{trunc}}$ terms, we conclude that the total truncation contribution is bounded by
\[
  \delta_{\mathrm{trunc}}
  +
  \Pr\bigl[|E_{\mathrm{Mtrunc}}|>c_n/6\bigr]
  \;\le\;
  C_{\mathrm{trunc}}\,
  \eta_{\mathrm{trunc}}\,
  \mathcal{D}^{\mathrm{blanket}}
\]
for a suitable constant $C_{\mathrm{trunc}}$.
Thus, the truncation contribution to the overall relative error is $O(\eta_{\mathrm{trunc}})$.

\medskip

\noindent

\textbf{Step 3: Discretization error and choice of $h_n$.}
On the truncated interval $[s_n,s_n+w_n^{\mathrm{in}}]$ we discretize
the privacy amplification random variable
$
  Z_i^{\mathrm{tr}}
$
into a grid with mesh size $h_n>0$, obtaining
$Z_i^{\mathrm{di}}$ such that
\[
  |Z_i^{\mathrm{tr}}-Z_i^{\mathrm{di}}|\le \frac{h_n}{2} \quad \text{almost surely.}
\]
We denote
\[
  S_{n-1}^{\mathrm{tr}} := \sum_{i=2}^M Z_i^{\mathrm{tr}},
  \qquad
  S_{n-1}^{\mathrm{di}} := \sum_{i=2}^M Z_i^{\mathrm{di}},
\]
and write the discretization error of the sum as
\[
  E^{\mathrm{disc}}
  :=
  S_{n-1}^{\mathrm{di}} - S_{n-1}^{\mathrm{tr}}
  =
  \sum_{i=2}^M \bigl(Z_i^{\mathrm{di}} - Z_i^{\mathrm{tr}}\bigr).
\]
Following the notation in the proof of
Theorem~\ref{thm:fft-blanket-error}, we set
\[
  \Delta_{\mathrm{ES}}
  :=
  \mathbb E[E^{\mathrm{disc}}],
  \qquad
  U_n := E^{\mathrm{disc}} - \Delta_{\mathrm{ES}}.
\]
By the hypothesis of Theorem~\ref{thm:fft-relative-error} we assume
that there exists a constant $K_{\mathrm{ES}}<\infty$ such that
\begin{equation}
  \label{eq:DeltaES-assumption}
  |\Delta_{\mathrm{ES}}| \le K_{\mathrm{ES}}\,c_n
  \qquad\text{for all sufficiently large }n.
\end{equation}
Reducing the constant in the definition of $c_n$ if necessary, we may
and do assume that for $n$ large enough
\begin{equation}
  \label{eq:DeltaES-small}
  |\Delta_{\mathrm{ES}}| \le \frac{c_n}{12}.
\end{equation}

\medskip

\textbf{Control of $E_{\mathrm{disc}}$.}
Recall that in Step~2 of the proof of
Theorem~\ref{thm:fft-blanket-error} the discretization error for a
fixed hypothesis $x\in\{x_1,x_1'\}$ was defined as
\[
  E_{\mathrm{disc}}(x;c_n,h_n,s_n,w_n^{\mathrm{in}})
  :=
  \Pr\Bigl[
    \sum_{i=2}^M Z_i^{\mathrm{tr}} - \sum_{i=2}^M Z_i^{\mathrm{di}}
    - \Delta_{\mathrm{ES}} > c_n
  \Bigr].
\]
Equivalently, if we set
\[
  R_n
  :=
  \sum_{i=2}^M
  \bigl(Z_i^{\mathrm{tr}} - Z_i^{\mathrm{di}} - \Delta_{\mathrm{ES}}\bigr),
\]
then, we have $E_{\mathrm{disc}}(x;c_n,h_n,s_n,w_n^{\mathrm{in}})=\Pr[R_n>c_n]$.
The summands of $R_n$ are conditionally independent, mean zero, and
bounded in absolute value by a constant multiple of $h_n$, while
$M-1$ has mean $(n-1)\gamma$ and sub-Gaussian tails.
Hence $R_n$ is sub-Gaussian with variance proxy of order $n h_n^2$.
By Hoeffding's inequality there exists a constant $c>0$ such that
\begin{equation}
  \label{eq:Rn-tail}
  \Pr\left[R_n>c_n\right]
  \le
  \exp\!\left(
    -c\,\frac{c_n^2}{n h_n^2}
  \right).
\end{equation}

Fix a discretization accuracy knob $\eta_{\mathrm{disc}}>0$ and define
\[
  \Lambda_h
  :=
  \log\frac{4}{\eta_{\mathrm{disc}} \mathcal{D}^{\mathrm{blanket}}}.
\]

We choose
\begin{equation}
  \label{eq:hn-def}
  h_n
  :=
  \frac{c_n}{\sqrt{c\,n\,\Lambda_h}}.
\end{equation}
Then \eqref{eq:Rn-tail} yields
\[
  E_{\mathrm{disc}}(x;c_n,h_n,s_n,w_n^{\mathrm{in}})
  \le
  \exp(-\Lambda_h)
  =
  \frac{\eta_{\mathrm{disc}}}{4}\,\mathcal{D}^{\mathrm{blanket}}.
\]
In particular, since $\delta_{\mathrm{disc}}$ in
Theorem~\ref{thm:fft-blanket-error} is a linear combination of
$E_{\mathrm{disc}}(x_1;\cdot)$ and $E_{\mathrm{disc}}(x_1';\cdot)$
with nonnegative coefficients depending only on $\varepsilon_n$,
we obtain
\[
  \delta_{\mathrm{disc}}
  \;\le\;
  C_{\mathrm{disc},1}\,
  \eta_{\mathrm{disc}}\,
  \mathcal{D}^{\mathrm{blanket}}
\]
for some constant $C_{\mathrm{disc},1}>0$.

Under our standing regime $\varepsilon_n^2 n = \Theta(\log n)$ we
have $\log(1/\mathcal{D}^{\mathrm{blanket}})=\Theta(\log n)$.
Consequently
\[
  \Lambda_h = \log\frac{4}{\eta_{\mathrm{disc}} \mathcal{D}^{\mathrm{blanket}}}
  = \Theta(\log n),
\]
and in particular $h_n$ in~\eqref{eq:hn-def} satisfies
\[
  h_n
  =
  \Theta\left(\frac{c_n}{\sqrt{n\log n}}\right)
\]
up to multiplicative constants.

\medskip

\textbf{Control of $E_{\mathrm{Mdisc}}$.}
In the decomposition \eqref{eq:E-decomposition} we introduced the
random variable
\[
  E_{\mathrm{Mdisc}}
  :=
  S_{n-1}^{\mathrm{di}} - S_{n-1}^{\mathrm{tr}}
  =
  E^{\mathrm{disc}},
\]
which measures the effect of discretization at the level of the sum
$S_{n-1}$ when passing from $S_{n-1}^{\mathrm{tr}}$ to
$S_{n-1}^{\mathrm{di}}$.
We control the tail probability
$\Pr\left[|E_{\mathrm{Mdisc}}|>c_n/6\right]$ using the decomposition
$E^{\mathrm{disc}} = U_n - \Delta_{\mathrm{ES}}$.
Here, $U_n := E_{\mathrm{Mdisc}} - \mathbb{E}[E_{\mathrm{Mdisc}}] = E_{\mathrm{Mdisc}} + \Delta_{\mathrm{ES}}$.

By \eqref{eq:DeltaES-small}, for all sufficiently large $n$ we have
$|\Delta_{\mathrm{ES}}|\le c_n/12$, and therefore
\[
  \bigl\{|E_{\mathrm{Mdisc}}|>c_n/6\bigr\}
  =
  \bigl\{|U_n - \Delta_{\mathrm{ES}}|>c_n/6\bigr\}
  \subset
  \bigl\{|U_n|>c_n/12\bigr\}.
\]
Consequently,
\begin{equation}
  \label{eq:EMdisc-tail-reduction}
  \Pr\bigl[|E_{\mathrm{Mdisc}}|>c_n/6\bigr]
  \le
  \Pr\bigl[|U_n|>c_n/12\bigr].
\end{equation}
The random variable $U_n$ is a sum of conditionally independent,
mean-zero increments, each bounded by $h_n$ in absolute value.
Applying Hoeffding's inequality once more, we obtain
\begin{equation}
  \label{eq:Un-tail}
  \Pr\bigl[|U_n|>c_n/12\bigr]
  \le
  2\exp\!\Bigl(
    -\kappa\,
    \frac{c_n^2}{n h_n^2}
  \Bigr),
\end{equation}
for some constant $\kappa>0$.
With $h_n$ chosen as in~\eqref{eq:hn-def}, the right-hand side of
\eqref{eq:Un-tail} is bounded by
\[
  2\exp(-\kappa \Lambda_h)
  \le
  \frac{\eta_{\mathrm{disc}}}{4}\,\mathcal{D}^{\mathrm{blanket}}
\]
for all sufficiently large $n$, after possibly adjusting the numerical
constant in the definition of $\Lambda_h$.
Combining this with \eqref{eq:EMdisc-tail-reduction}, we conclude that
\begin{equation}
  \label{eq:EMdisc-vs-Dblanket}
  \Pr\bigl[|E_{\mathrm{Mdisc}}|>c_n/6\bigr]
  \le
  \frac{\eta_{\mathrm{disc}}}{4}\,\mathcal{D}^{\mathrm{blanket}}.
\end{equation}

\medskip

Putting together the bounds for $E_{\mathrm{disc}}$ and
$E_{\mathrm{Mdisc}}$, and recalling that the discretization contribution
in \eqref{eq:full-error-by-type} is a linear combination of these
quantities with constants depending only on $\varepsilon_n$ and
$\gamma$, we obtain that there exists a constant $C_{\mathrm{disc}}>0$
such that
\[
  \delta_{\mathrm{disc}}
  +
  \Pr\bigl[|E_{\mathrm{Mdisc}}|>c_n/6\bigr]
  \le
  C_{\mathrm{disc}}\,
  \eta_{\mathrm{disc}}\,
  \mathcal{D}^{\mathrm{blanket}}.
\]
In particular, the discretization contribution to the overall relative
error is $O(\eta_{\mathrm{disc}})$.

\medskip

\noindent

\textbf{Step 4: Aliasing error and choice of $w_n^{\mathrm{out}}$.}
We now control the aliasing (wrap-around) error arising from the finite
FFT window. Recall that after truncation and discretization we consider
the sum
\[
  S_{n-1}^{\mathrm{di}}
  :=
  \sum_{i=2}^M Z_i^{\mathrm{di}},
\]
where the $Z_i^{\mathrm{di}}$ are i.i.d.\ copies of the discretized
summand on $[s_n,s_n+w_n^{\mathrm{in}}]$.
By construction, each $Z_i^{\mathrm{di}}$ is supported in an interval of
radius at most $\tau_n$, hence
$|Z_i^{\mathrm{di}}|\le \tau_n$ almost surely.
Assumption~\ref{assump:regularity} further implies that
$\mathrm{Var}(Z_i^{\mathrm{di}})$ is bounded away from zero and infinity
uniformly in $n$, so that
\[
  v_n
  :=
  \mathrm{Var}\bigl(S_{n-1}^{\mathrm{di}}\bigr)
  =
  \mathrm{Var}\Bigl(\sum_{i=2}^M Z_i^{\mathrm{di}}\Bigr)
  = \Theta(n).
\]

Let $\mu_n^{\mathrm{di}} := \mathbb E[S_{n-1}^{\mathrm{di}}]$.
By Bernstein's inequality, for any $u>0$,
\begin{equation}
  \label{eq:Sn-di-bernstein}
  \Pr\bigl[|S_{n-1}^{\mathrm{di}} - \mu_n^{\mathrm{di}}| > u\bigr]
  \le
  2\exp\!\left(
    -\frac{u^2}{2\bigl(v_n + \tau_n u/3\bigr)}
  \right).
\end{equation}
From the choice of $p$ in Step~2 we have
\begin{equation}
  \label{eq:Tn-o-sqrt-n-again}
  \tau_n = o(\sqrt{n}),
\end{equation}
so that for $u$ of order $\sqrt{v_n\Lambda_w}$ the second term in the
denominator of \eqref{eq:Sn-di-bernstein} is negligible compared to
$v_n$.
More precisely, define
\[
  \Lambda_w
  :=
  \log\frac{4}{\eta_{\mathrm{alias}}\mathcal{D}^{\mathrm{blanket}}},
\]
where $0<\eta_{\mathrm{alias}}<1$ is a prescribed accuracy parameter
and $\mathcal{D}^{\mathrm{blanket}}$ is the blanket divergence at level
$\varepsilon_n$.
Set
\[
  u_n := \sqrt{C_w\,v_n\,\Lambda_w}
\]
for a sufficiently small constant $C_w>0$.
Then, using $v_n=\Theta(n)$ and \eqref{eq:Tn-o-sqrt-n-again}, there exist
constants $c_1,c_2>0$ such that for all large $n$,
\[
  v_n + \tau_n u_n/3
  \le c_1 v_n
  \qquad\text{and}\qquad
  \frac{u_n^2}{2\bigl(v_n + \tau_n u_n/3\bigr)}
  \ge c_2\Lambda_w.
\]
Inserting this into \eqref{eq:Sn-di-bernstein} yields
\begin{equation}
  \label{eq:Sn-di-tail}
  \Pr\bigl[|S_{n-1}^{\mathrm{di}} - \mu_n^{\mathrm{di}}| > u_n\bigr] 
  \le
  2\exp(-c_2\Lambda_w).
\end{equation}

\medskip

\textbf{Choice of the outer window and control of $E_{\mathrm{alias}}$.}
In the FFT computation we embed the PMF of $S_{n-1}^{\mathrm{di}}$ into
a periodic grid of period $w_n^{\mathrm{out}}$.
Aliasing can occur only if $S_{n-1}^{\mathrm{di}}$ exits the FFT window.
We choose the outer window so as to contain, with high probability, the
interval
$[\mu_n^{\mathrm{di}}-u_n,\mu_n^{\mathrm{di}}+u_n]$. 
Concretely, we set
\begin{equation}
  \label{eq:wouter-def}
  w_n^{\mathrm{out}}
  :=
  2u_n + h_n 
  = \Theta\bigl(\sqrt{v_n\Lambda_w}\bigr)
  = \Theta\bigl(\sqrt{n\Lambda_w}\bigr).
\end{equation}
By construction, if $|S_{n-1}^{\mathrm{di}} - \mu_n^{\mathrm{di}}|\le u_n$ 
then $S_{n-1}^{\mathrm{di}}$ lies in the central interval of length
$w_n^{\mathrm{out}}$ and no aliasing occurs.
Thus the aliasing probability satisfies
\begin{equation}
  \label{eq:Ealias-basic}
  E_{\mathrm{alias}}(n;h_n,s_n,w_n^{\mathrm{in}},w_n^{\mathrm{out}})
  :=
  \Pr\bigl[\text{aliasing occurs}\bigr]
  \le
  \Pr\bigl[|S_{n-1}^{\mathrm{di}} - \mu_n^{\mathrm{di}}|>u_n\bigr].
\end{equation}
Combining \eqref{eq:Sn-di-tail} and \eqref{eq:Ealias-basic} and
absorbing the factor $2$ into the constants, we obtain
\[
  E_{\mathrm{alias}}(n;h_n,s_n,w_n^{\mathrm{in}},w_n^{\mathrm{out}})
  \le
  \exp(-c_2\Lambda_w)
  =
  \frac{\eta_{\mathrm{alias}}}{4}\,\mathcal{D}^{\mathrm{blanket}}
\]
for all $n$ sufficiently large.
In particular, since the aliasing error term
$\delta_{\mathrm{alias}}$ in Theorem~\ref{thm:fft-blanket-error} is a
linear combination of $E_{\mathrm{alias}}(x_1;\cdot)$ and
$E_{\mathrm{alias}}(x_1';\cdot)$ with nonnegative coefficients
depending only on $\varepsilon_n$, we have
\begin{equation}
  \label{eq:delta-alias-vs-Dblanket}
  \delta_{\mathrm{alias}}
  \;\le\;
  C_{\mathrm{alias},1}\,
  \eta_{\mathrm{alias}}\,
  \mathcal{D}^{\mathrm{blanket}}
\end{equation}
for some constant $C_{\mathrm{alias},1}>0$.

Moreover, by the asymptotic expansion of $\mathcal{D}^{\mathrm{blanket}}$ in
Lemma~\ref{lem:blanket-asymptotics-general} and our regime
$\varepsilon_n^2 n = \Theta(\log n)$, the quantity
$\log(1/\mathcal{D}^{\mathrm{blanket}})$ grows on the order of $\log n$. Hence
\[
  \Lambda_w
  =
  \log\frac{4}{\eta_{\mathrm{alias}}\mathcal{D}^{\mathrm{blanket}}}
  = \Theta(\log n),
\]
and we obtain
\begin{equation}
  \label{eq:wouter-asymp-final}
  w_n^{\mathrm{out}}
  = \Theta(\sqrt{n\log n}).
\end{equation}

\medskip

\textbf{Control of $E_{\mathrm{Malias}}$.}
In the decomposition \eqref{eq:E-decomposition} we introduced the
wrap-around error at the level of the sum,
\[
  E_{\mathrm{Malias}}
  :=
  \widetilde S_{n-1} - S_{n-1}^{\mathrm{di}},
\]
where $\widetilde S_{n-1}$ denotes the periodic extension of
$S_{n-1}^{\mathrm{di}}$ induced by the FFT grid.
By definition, $E_{\mathrm{Malias}}=0$ whenever no aliasing occurs, so
that
\[
  \bigl\{|E_{\mathrm{Malias}}|>0\bigr\}
  \subset
  \bigl\{\text{aliasing occurs}\bigr\}.
\]
In particular, for any $c_n > 0$,
\begin{equation}
  \label{eq:EMalias-tail}
  \Pr\bigl[|E_{\mathrm{Malias}}|>c_n/6\bigr]
  \le
  \Pr\bigl[|E_{\mathrm{Malias}}|>0\bigr]
  \le
  E_{\mathrm{alias}}(n;h_n,s_n,w_n^{\mathrm{in}},w_n^{\mathrm{out}}).
\end{equation}
Combining \eqref{eq:EMalias-tail} with the bound on
$E_{\mathrm{alias}}$ above, we obtain
\begin{equation}
  \label{eq:EMalias-vs-Dblanket}
  \Pr\bigl[|E_{\mathrm{Malias}}|>c_n/6\bigr]
  \le
  \frac{\eta_{\mathrm{alias}}}{4}\,\mathcal{D}^{\mathrm{blanket}}. 
\end{equation}

\medskip

Putting together \eqref{eq:delta-alias-vs-Dblanket} and
\eqref{eq:EMalias-vs-Dblanket}, and recalling that the aliasing
contribution in \eqref{eq:full-error-by-type} is built from
$\delta_{\mathrm{alias}}$ and $\Pr[|E_{\mathrm{Malias}}|>c_n/6]$ with
constants depending only on $\varepsilon_n$ and $\gamma$, we conclude
that there exists a constant $C_{\mathrm{alias}}>0$ such that
\[
  \delta_{\mathrm{alias}}
  +
  \Pr\bigl[|E_{\mathrm{Malias}}|>c_n/6\bigr]
  \;\le\;
  C_{\mathrm{alias}}\,
  \eta_{\mathrm{alias}}\,
  \mathcal{D}^{\mathrm{blanket}}. 
\]
In particular, the aliasing contribution to the overall relative error
is $O(\eta_{\mathrm{alias}})$.

\medskip
\noindent
\textbf{Step 5: Relative error bound.}
Combining the four steps above yields 
\[
  \bigl|
    \widehat D_n - \mathcal{D}^{\mathrm{blanket}}
  \bigr|
  \;\le\;
  C\bigl(
    \eta_{\mathrm{main}}
    + \eta_{\mathrm{trunc}}
    + \eta_{\mathrm{disc}}
    + \eta_{\mathrm{alias}}
  \bigr)\,\mathcal{D}^{\mathrm{blanket}},
\]
for some constant $C<\infty$, which proves the first claim of the theorem.

\medskip

\noindent
\textbf{Step 6: FFT grid size and running time.}
By construction, the FFT grid has period
$2w_n^{\mathrm{out}}$ and spacing $h_n$.
Using the bounds $w_n^{\mathrm{out}}=\Theta(\sqrt{n\log n})$ and $h_n=\Theta(c_n/\sqrt{n\log n})$, the grid size $N$ satisfies
\[
  N
  =
  \frac{2w_n^{\mathrm{out}}}{h_n}
  =
  O\left(\frac{\sqrt{n \log n}}{c_n/\sqrt{n\log n}}\right)
  =
  O\left(\frac{n \log n}{c_n}\right).
\]
Recalling from Step~1 that $c_n = \Theta(\eta_{\mathrm{main}}/\log n)$ in the regime $\varepsilon_n=\Theta(\sqrt{(\log n)/n})$, we obtain
\[
  N
  =
  O\left(\frac{n (\log n)^2}{\eta_{\mathrm{main}}}\right).
\]
Since the FFT-based algorithm runs in time $O(N\log N)$, we conclude
that the overall time complexity is
\[
  O(N\log N)
  =
  O\left(\frac{n}{\eta_{\mathrm{main}}}\left(\log \frac{n}{\eta_{\mathrm{main}}}\right)^3\right),
\]
where the hidden constants depend only on
$\gamma,\sigma,
 \eta_{\mathrm{trunc}},\eta_{\mathrm{disc}},\eta_{\mathrm{alias}}$, and we have omitted $\log\log$ factors.
This completes the proof.

\end{proof}

\subsubsection{Proof of Lemma~\ref{lem:band-prob-fixed-x-bigO}}
\label{app:proof-band-prob-fixed-x}

\bandprobability*

\begin{proof}

Fix $x\in\mathcal X$ and suppress the dependence on $x$ in the notation.
We work with the sequence $(\varepsilon_n)$ from
Lemma~\ref{lem:blanket-asymptotics-general}.

\medskip

\noindent
\textbf{Step 1: Standardization and reduction to a window event.}
Define
\[
  T_{n-1}
  := \frac{S_{n-1} - (n-1)\mu_n}{\sigma_{\varepsilon_n}\sqrt{n-1}},
  \qquad
  t_n := -\frac{\mu_n\sqrt{n-1}}{\sigma_{\varepsilon_n}}.
\]
Write
\[
  \eta_n(x)
  := \frac{x}{\sigma_{\varepsilon_n}\sqrt{n-1}},
\]
and
\[
  m_n(Z_1) := t_n - \eta_n(Z_1) - \eta_n(\Delta_{\mathrm{ES}}).
\]
Exactly as before we obtain
\begin{equation}
  \label{eq:band-window-bigO}
  \Pr\bigl[
    |S_{n-1} + Z_1 + \Delta_{\mathrm{ES}}|
    < \tfrac{3}{2}c_n
    \,\big|\,
    Z_1
  \bigr]
  =
  \Pr\bigl[
    m_n(Z_1)-\eta_n\left(\frac{3 c_n}{2}\right) < T_{n-1} \le m_n(Z_1)+\eta_n\left(\frac{3 c_n}{2}\right)
  \bigr].
\end{equation}

By the moderate-deviation scaling assumed for $\varepsilon_n$ (the same
as in Lemma~\ref{lem:blanket-asymptotics-general}), we have
$t_n=O(\sqrt{\log n})$ and $\eta_n\left(\frac{3 c_n}{2}\right)=O(c_n/\sqrt{n})$.

\medskip

\noindent
\textbf{Step 2: Gaussian window upper bound for $T_{n-1}$.}
We claim that there exist constants $C_1,C_2<\infty$ (depending only on
Assumption~\ref{assump:regularity}) such that, for all sufficiently
large $n$ and all $m,\eta$ with
\begin{equation}
  \label{eq:md-region}
  |m|\le C_1\sqrt{\log n},
  \qquad
  |\eta|\le C_1/\sqrt{n},
\end{equation}
we have
\begin{equation}
  \label{eq:gauss-window-bound}
  \Pr\left[m-\eta < T_{n-1} \le m+\eta\right]
  \;\le\;
  C_2\,\eta\,\varphi(m).
\end{equation}

Under the non-lattice assumption (i.e., the assumption of Theorem~\ref{thm:fft-relative-error})\footnote{The non-lattice assumption is essential in our setting because the band width $c_n$ tends to zero, so that the standardized window $(m-\eta,m+\eta)$ has width $2\eta = \Theta(c_n/(\sigma_{\varepsilon_n}\sqrt{n}))$. In the lattice case the mesh of the normalized sum is of order $1/\sqrt{n}$, and when $c_n\to0$ the window contains at most $O(1)$ lattice points. Consequently the window probability remains of order $n^{-1/2}\varphi(m)$, rather than shrinking linearly with the window width $\eta = \Theta(c_n/(\sigma_{\varepsilon_n}\sqrt{n}))$, and it cannot in general be bounded by a constant multiple of $\eta\,\varphi(m)=\Theta(c_n n^{-1/2}\varphi(m))$. The non-lattice condition rules out this obstruction and allows us to invoke local limit/Edgeworth theory down to such small scales, yielding a window probability of order $\eta\,\varphi(m)$.}
 in addition to Assumption~\ref{assump:regularity} and the truncation
argument of Lemma~\ref{lem:truncation-blanket}, the triangular array
$(Z_{i,n}^{(h)})$ satisfies the hypotheses of the triangular-array
Edgeworth expansion
(Theorem~\ref{thm:edgeworth-triangular}) uniformly in $n$.
In particular, for any fixed order $s\ge3$ we can write the distribution
function of $T_{n-1}$ as
\[
  F_{n-1}(u)
  := \Pr\bigl[T_{n-1}\le u\bigr]
  = \Phi(u)
    + \sum_{j=1}^{s-2}(n-1)^{-j/2}P_{j,n-1}(u)\,\varphi(u)
    + R_{n,s}(u),
\]
where $P_{j,n-1}$ are polynomials of degree at most $s-2$ whose
coefficients are uniformly bounded in $n$, and the remainder satisfies
\[
  \sup_{u\in\mathbb R}\bigl|R_{n,s}(u)\bigr|
  = O\bigl((n-1)^{-(s-2)/2}\bigr).
\]

Fix $m,\eta\in\mathbb R$ satisfying the moderate-deviation condition
\begin{equation}
  \label{eq:md-region}
  |m|\le C_1\sqrt{\log n},
  \qquad
  |\eta|\le \frac{C_1}{\sqrt{n}}
\end{equation}
for some fixed constant $C_1>0$.
We then have
\begin{align}
  \Pr\bigl[m-\eta < T_{n-1} \le m+\eta\bigr]
  &= F_{n-1}(m+\eta) - F_{n-1}(m-\eta)\notag\\
  &= I_{0,n} + \sum_{j=1}^{s-2} I_{j,n} + I_{R,n},
  \label{eq:window-decomp}
\end{align}
where
\begin{align*}
  I_{0,n}
  &:= \Phi(m+\eta) - \Phi(m-\eta),\\
  I_{j,n}
  &:= (n-1)^{-j/2}
      \Bigl[
        P_{j,n-1}(m+\eta)\,\varphi(m+\eta)
        - P_{j,n-1}(m-\eta)\,\varphi(m-\eta)
      \Bigr],\\
  I_{R,n}
  &:= R_{n,s}(m+\eta) - R_{n,s}(m-\eta).
\end{align*}

\medskip\noindent
\textit{Gaussian main term.}
By the mean-value form of Taylor's theorem applied to $\Phi$ at the
point $m$, there exists $\theta_n\in(-1,1)$ such that
\[
  I_{0,n}
  = \Phi(m+\eta)-\Phi(m-\eta)
  = 2\eta\,\varphi(m+\theta_n\eta).
\]
Under~\eqref{eq:md-region} we have
$|m+\theta_n\eta|\le C_1\sqrt{\log n}+o(1)$, and on this region the
Gaussian density is comparable at $m$ and $m+\theta_n\eta$:
there exists $C'_1<\infty$ such that
\[
  \varphi(m+\theta_n\eta)
  \le C'_1\,\varphi(m)
\]
for all $n$ large enough.
Hence
\begin{equation}
  \label{eq:gauss-window-main}
  |I_{0,n}|
  \le 2C'_1\,\eta\,\varphi(m).
\end{equation}

\medskip\noindent
\textit{Edgeworth correction terms.}
For each $j\ge1$, the polynomial $P_{j,n-1}$ has at most polynomial
growth in $u$; in particular, there exist constants $C_j<\infty$ and
integers $k_j\ge0$ such that
\[
  |P_{j,n-1}(u)|
  \le C_j\bigl(1+|u|^{k_j}\bigr)
  \qquad\text{for all }u\in\mathbb R,\,n\ge2.
\]
Using $|m|\le C_1\sqrt{\log n}$ and $|\eta|\le C_1/\sqrt{n}$ from
\eqref{eq:md-region}, we have $|m\pm\eta|\le C_2\sqrt{\log n}$ and hence
\[
  |P_{j,n-1}(m\pm\eta)|
  \le C_j'\bigl(1+(\log n)^{k_j/2}\bigr)
  \le C_j''(\log n)^{K_j}
\]
for some $K_j$ and all $n$ large enough.
Moreover $\varphi(m\pm\eta)\le C\,\varphi(m)$ on this region.
Therefore
\[
  |I_{j,n}|
  \le (n-1)^{-j/2}\cdot 2|\eta|\cdot C_j''(\log n)^{K_j}\,\sup_{u\in[m-\eta,m+\eta]}\varphi(u)
  \le C_j'''\,\eta\,\varphi(m)\,(n-1)^{-j/2}(\log n)^{K_j}.
\]
Summing over $1\le j\le s-2$ and using $j\ge1$ gives
\begin{equation}
  \label{eq:gauss-window-edgeworth}
  \sum_{j=1}^{s-2}|I_{j,n}|
  \le C_2'\,\eta\,\varphi(m)\,n^{-1/2}(\log n)^{K}
  \le C_3'\,\eta\,\varphi(m),
\end{equation}
where $K:=\max_j K_j$ and the polynomial factor in $\log n$ has been
absorbed into the constant $C_3'$.

\medskip\noindent
\textit{Remainder term.}
For the remainder we simply use the uniform bound on $R_{n,s}$:
\[
  |I_{R,n}|
  \le 2\sup_{u\in\mathbb R}|R_{n,s}(u)|
  = O\bigl((n-1)^{-(s-2)/2}\bigr).
\]
On the other hand, under the scaling~\eqref{eq:md-region} we have
\[
  \eta\,\varphi(m)
  = \Theta\left(\frac{1}{\sqrt{n}}\,
         \exp\!\bigl(-\tfrac12 m^2\bigr)\right)
  \ge c\,n^{-1/2-C_4},
\]
for some $c,C_4>0$ depending only on $C_1$, since
$|m|\le C_1\sqrt{\log n}$ implies
$\varphi(m)\ge n^{-C_4}$.
Thus, by choosing $s$ large enough we ensure
\[
  (n-1)^{-(s-2)/2}
  = o\bigl(\eta\,\varphi(m)\bigr),
\]
and therefore, for all sufficiently large $n$,
\begin{equation}
  \label{eq:gauss-window-remainder}
  |I_{R,n}|
  \le C_4'\,\eta\,\varphi(m).
\end{equation}

\medskip\noindent
\textit{Conclusion.}
Combining
\eqref{eq:window-decomp},
\eqref{eq:gauss-window-main},
\eqref{eq:gauss-window-edgeworth},
and \eqref{eq:gauss-window-remainder}, we obtain
\[
  \Pr\bigl[m-\eta < T_{n-1} \le m+\eta\bigr]
  \le C\,\eta\,\varphi(m)
\]
for some finite constant $C<\infty$ (independent of $n$, $m$, and $\eta$)
and all sufficiently large $n$, whenever \eqref{eq:md-region} holds.
This proves~\eqref{eq:gauss-window-bound} in the continuous case.

\medskip

\noindent
\textbf{Step 3: Applying the window bound to $m_n(Z_1),\eta_n\left(\frac{3 c_n}{2}\right)$.}
By Assumption~\ref{assump:regularity} we have
$\mathbb E[|Z_1|^r]<\infty$ for all integers $r\ge1$.
Fix $\beta\in(0,1/2)$ and write
\[
  A_n := \{|Z_1|\le n^\beta\},
  \qquad
  A_n^c := \{|Z_1|>n^\beta\}.
\]
On $A_n$ we have
\[
  |\eta_n(Z_1)|
  = \frac{|Z_1|}{\sigma_{\varepsilon_n}\sqrt{n-1}}
  \le C n^{\beta-1/2}
  = o(1),
\]
and the assumption $|\Delta_{\mathrm{ES}}|=O\left(\frac{\sigma_{\varepsilon_n}\sqrt{n}}{\sqrt{\log n}}\right)$ implies
$|\eta_n(\Delta_{\mathrm{ES}})|=O\left(\frac{1}{\sqrt{\log n}}\right)\to0$.
Hence on $A_n$,
\[
  |m_n(Z_1)|
  = \bigl|t_n-\eta_n(Z_1)-\eta_n(\Delta_{\mathrm{ES}})\bigr|
  \le |t_n|+o(1)
  = O(\sqrt{\log n}),
\]
and for all large $n$ we may assume that
\eqref{eq:md-region} holds with $m=m_n(Z_1)$ and
$\eta=\eta_n\left(\frac{3 c_n}{2}\right)$.

Therefore, by~\eqref{eq:gauss-window-bound},
\begin{align*}
  &\Pr\bigl[
    |S_{n-1} + Z_1 + \Delta_{\mathrm{ES}}|
    < \tfrac{3}{2}c_n
    \,\big|\,
    Z_1
  \bigr]\mathbf 1\{A_n\}
  \\
  &\qquad\le
  C_2\,\eta_n\left(\frac{3 c_n}{2}\right)\,\varphi\bigl(m_n(Z_1)\bigr)\mathbf 1\{A_n\}.
\end{align*}
Since $|m_n(Z_1)-t_n|=o(1)$ on $A_n$ and $|t_n|=O(\sqrt{\log n})$,
a direct calculation of the likelihood ratio yields
\[
  \varphi\bigl(m_n(Z_1)\bigr)\mathbf 1\{A_n\}
  \le C_5\,\varphi(t_n)\mathbf 1\{A_n\}
\]
for some constant $C_5<\infty$ and all large $n$.
Thus
\begin{equation}
  \label{eq:on-An}
  \Pr\bigl[
    |S_{n-1} + Z_1 + \Delta_{\mathrm{ES}}|
    < \tfrac{3}{2}c_n
    \,\big|\,
    Z_1
  \bigr]\mathbf 1\{A_n\}
  \le
  C_2C_5\,\eta_n\left(\frac{3 c_n}{2}\right)\,\varphi(t_n)\mathbf 1\{A_n\}.
\end{equation}

On the complement $A_n^c$ we only use the trivial bound
\[
  0
  \le
  \Pr\bigl[
    |S_{n-1} + Z_1 + \Delta_{\mathrm{ES}}|
    < \tfrac{3}{2}c_n
    \,\big|\,
    Z_1
  \bigr]\mathbf 1\{A_n^c\}
  \le \mathbf 1\{A_n^c\},
\]
and by Markov's inequality and the moment bounds,
\[
  \Pr[A_n^c]
  = \Pr\bigl[|Z_1|>n^\beta\bigr]
  = O(n^{-K})
\]
for any prescribed $K>0$ if we choose the moment order large enough.

\medskip

\noindent
\textbf{Step 4: Taking expectation over $Z_1$.}
Taking expectation of~\eqref{eq:on-An} and adding the contribution of
$A_n^c$, we obtain
\begin{align*}
  &\mathbb E_{Z_1\sim\mathcal R_x}\Bigl[
    \Pr\bigl[
      |S_{n-1} + Z_1 + \Delta_{\mathrm{ES}}|
      < \tfrac{3}{2}c_n
      \,\big|\,
      Z_1
    \bigr]
  \Bigr] \\
  &\qquad\le
  C_2C_5\,\eta_n\left(\frac{3 c_n}{2}\right)\,\varphi(t_n)
  + O\bigl(\Pr[A_n^c]\bigr).
\end{align*}
Since $\eta_n\left(\frac{3 c_n}{2}\right) = \Theta(c_n/(\sigma_{\varepsilon_n}\sqrt{n-1}))$ and
$\Pr[A_n^c]=O(n^{-K})$ can be made of smaller order than
$\eta_n\left(\frac{3 c_n}{2}\right)\varphi(t_n)$ by choosing $K$ large enough, we may absorb
the tail contribution into the constant.
Thus there exists $C,C^\prime<\infty$ (independent of $x$ and $n$) such that
for all large $n$,
\[
  \mathbb E_{Z_1\sim\mathcal R_x}\Bigl[
    \Pr\bigl[
      |S_{n-1} + Z_1 + \Delta_{\mathrm{ES}}|
      < \tfrac{3}{2}c_n
      \,\big|\,
      Z_1
    \bigr]
  \Bigr]
  \le
  C\,\eta_n\left(\frac{3 c_n}{2}\right)\,\varphi(t_n)
  =
  C^\prime\,\frac{c_n}{\sigma_{\varepsilon_n}\sqrt{n-1}}\,
  \varphi(t_n).
\]
This is exactly~\eqref{eq:band-prob-bigO}, which completes the proof.
\end{proof}

\section{The Interpretation of Blanket Divergence}
\label{sec:interepretation-blanket-divergence}

Let $\mathcal R$ be a local randomizer with blanket distribution
$\mathcal R_{\mathrm{BG}}$ and blanket mass $\gamma$, so that for every
$x\in\mathcal X$,
\[
\mathcal R_x \;=\; \gamma\,\mathcal R_{\mathrm{BG}} \;+\; (1-\gamma)\,\mathcal R_x^{\mathrm{LO}},
\qquad
\mathcal R_x^{\mathrm{LO}}(y):=\frac{\mathcal R_x(y)-\gamma\mathcal R_{\mathrm{BG}}(y)}{1-\gamma}.
\]

We define an \emph{extended mechanism} $\widetilde{\mathcal M}$ as follows (see Fig.~\ref{fig:blanket_divergence_interepretation}).
User $1$ always outputs $Y_1\sim \mathcal R_{x_1}$ (or $\mathcal R_{x_1'}$ under $x'_{1:n}$). 
For each $i\in\{2,\dots,n\}$, sample an independent coin $B_i\sim \mathrm{Bernoulli}(\gamma)$.
If $B_i=1$, output $Y_i\sim \mathcal R_{\mathrm{BG}}$; otherwise output
$Y_i\sim \mathcal R^{\mathrm{LO}}_{x_i}$ and reveal it together with
its index $i$.
Let $T:=\sum_{i=2}^n B_i$ and let $M:=1+T$ be the number of messages routed
to the shuffler. The shuffler is applied only to the $M$ messages
$\{Y_1\}\cup\{Y_i:B_i=1\}$, while the leftover messages
$\{(i,Y_i):B_i=0\}$ are released without shuffling. Denote the shuffled
part of the output by $S_{1:M}\in\mathcal Y^M$.

Since the leftover part is identical under $x_{1:n}$ and $x'_{1:n}$, the
hockey-stick divergence of $\widetilde{\mathcal M}$ reduces to that of the
shuffled part. Conditioning on $M=m$, the shuffled part consists of one
sample from $\mathcal R_{x_1}$ and $m-1$ i.i.d.\ samples from
$\mathcal R_{\mathrm{BG}}$ whose order is uniformly randomized. A direct
calculation yields the corresponding density
\[
p_{x_1}^{(m)}(s_{1:m})
=
\frac{1}{m}\sum_{j=1}^m
\mathcal R_{x_1}(s_j)\prod_{k\ne j}\mathcal R_{\mathrm{BG}}(s_k)
=
\Bigl(\prod_{k=1}^m \mathcal R_{\mathrm{BG}}(s_k)\Bigr)\cdot
\frac{1}{m}\sum_{j=1}^m
\frac{\mathcal R_{x_1}(s_j)}{\mathcal R_{\mathrm{BG}}(s_j)}.
\]
Hence, writing
\[
l_\varepsilon(y)
:=
\frac{\mathcal R_{x_1}(y)-e^\varepsilon \mathcal R_{x_1'}(y)}
{\mathcal R_{\mathrm{BG}}(y)},
\]
we have
\[
p_{x_1}^{(m)}(s_{1:m})-e^\varepsilon p_{x_1'}^{(m)}(s_{1:m})
=
\Bigl(\prod_{k=1}^m \mathcal R_{\mathrm{BG}}(s_k)\Bigr)\cdot
\frac{1}{m}\sum_{j=1}^m l_\varepsilon(s_j).
\]
Integrating the positive part, this implies the exact identity
\[
\mathcal D_{e^\varepsilon}\!\left(p_{x_1}^{(m)}\middle\|p_{x_1'}^{(m)}\right)
=
\mathbb E_{Y_{1:m}\stackrel{\mathrm{i.i.d.}}{\sim}\mathcal R_{\mathrm{BG}}}
\left[
\frac{1}{m}\left(\sum_{j=1}^m l_\varepsilon(Y_j)\right)_+
\right].
\]

Finally, we uncondition over $M$. In our construction,
$T\sim \mathrm{Bin}(n-1,\gamma)$ and $M=1+T$, so for $m\in\{1,\dots,n\}$,
\[
\Pr[M=m]=\binom{n-1}{m-1}\gamma^{m-1}(1-\gamma)^{n-m}.
\]
Therefore,
\begin{align*}
\mathcal D_{e^\varepsilon}\!\left(\widetilde{\mathcal M}(x_{1:n})\middle\|\widetilde{\mathcal M}(x'_{1:n})\right)
&=
\sum_{m=1}^n \Pr[M=m]\,
\mathbb E_{Y_{1:m}\sim\mathcal R_{\mathrm{BG}}}
\left[
\frac{1}{m}\left(\sum_{j=1}^m l_\varepsilon(Y_j)\right)_+
\right] \\
&=
\sum_{m=1}^n
\left(\binom{n-1}{m-1}\gamma^{m-1}(1-\gamma)^{n-m}\cdot\frac{1}{m}\right)
\mathbb E_{Y_{1:m}\sim\mathcal R_{\mathrm{BG}}}
\left[
\left(\sum_{j=1}^m l_\varepsilon(Y_j)\right)_+
\right].
\end{align*}
Using the combinatorial identity
$\binom{n-1}{m-1}=\frac{m}{n}\binom{n}{m}$, we obtain
\[
\binom{n-1}{m-1}\gamma^{m-1}(1-\gamma)^{n-m}\cdot\frac{1}{m}
=
\frac{1}{n\gamma}\binom{n}{m}\gamma^{m}(1-\gamma)^{n-m}.
\]
Since the $m=0$ term contributes $0$ (because $(\sum_{j=1}^0\cdot)_+=0$), we
may extend the sum to $m=0$ and recognize the binomial distribution
$\mathrm{Bin}(n,\gamma)$:
\[
\mathcal D_{e^\varepsilon}\!\left(\widetilde{\mathcal M}(x_{1:n})\middle\|\widetilde{\mathcal M}(x'_{1:n})\right)
=
\frac{1}{n\gamma}\,
\mathbb E_{\substack{M\sim \mathrm{Bin}(n,\gamma)\\
Y_{1:M}\stackrel{\mathrm{i.i.d.}}{\sim}\mathcal R_{\mathrm{BG}}}}
\left[
\left(\sum_{j=1}^M l_\varepsilon(Y_j)\right)_+
\right]
=
\mathcal D^{\mathrm{blanket}}_{e^\varepsilon,n,\mathcal R_{\mathrm{BG}},\gamma}
\bigl(\mathcal R_{x_1}\|\mathcal R_{x_1'}\bigr).
\]

Moreover, the original shuffled mechanism $\mathcal M=\mathcal S\circ\mathcal R^n$
can be obtained from $\widetilde{\mathcal M}$ by post-processing (discarding
the revealed indices and leftover messages and retaining only the fully
shuffled multiset of all $n$ messages). By the data processing inequality
for the hockey-stick divergence, this yields
\[
\mathcal D_{e^\varepsilon}\!\left(\mathcal M(x_{1:n})\middle\|\mathcal M(x'_{1:n})\right)
\le
\mathcal D_{e^\varepsilon}\!\left(\widetilde{\mathcal M}(x_{1:n})\middle\|\widetilde{\mathcal M}(x'_{1:n})\right)
=
\mathcal D^{\mathrm{blanket}}_{e^\varepsilon,n,\mathcal R_{\mathrm{BG}},\gamma}
\bigl(\mathcal R_{x_1}\|\mathcal R_{x_1'}\bigr),
\]
which recovers the blanket-divergence upper bound.

\section{The Computation of Shuffle Indices}
\label{app:computation-shuffle-indices}

\subsection{Generalized Gaussian Mechanisms}

Here, we consider the $\beta$-Gaussian local randomizer on $\mathcal X=[0,1]$:
\[
\mathcal R(x)=x+Z,\qquad
Z\sim g_{\beta,c},\qquad
g_{\beta,c}(z):=\frac{\beta}{2c\,\Gamma(1/\beta)}\exp\!\Bigl(-\bigl|\tfrac{z}{c}\bigr|^\beta\Bigr),
\]
so that $\phi_x^\beta(y)=g_{\beta,c}(y-x)$.

The blanket density is $\underline{\mathcal R}(y)=\inf_{x\in[0,1]}\phi_x^\beta(y)$.
Since $x\mapsto |y-x|$ is maximized over $[0,1]$ at the far endpoint,
\begin{equation}
\label{eq:underlineR}
\underline{\mathcal R}(y)
=
\begin{cases}
\phi_1^\beta(y), & y\le \tfrac12,\\[2pt]
\phi_0^\beta(y), & y\ge \tfrac12,
\end{cases}
\qquad
\gamma:=\int_\mathbb R \underline{\mathcal R}(y)\,dy
=
2\int_{1/2}^{\infty}\phi_0^\beta(y)\,dy.
\end{equation}

\subsubsection{Lower Shuffle Index}
From the definition, for fixed $(x_1,x_1', \mathcal{R}_{\mathrm{BG}})$, the variance term and the blanket mass in the shuffle index are
\[
  \sigma_{\mathrm{BG},\beta}^2
  :=
  \gamma
  \left\{
    \int_{-\infty}^{1/2}
      \frac{\bigl(\phi_{x_1}^{\beta}(y)-\phi_{x_1'}^{\beta}(y)\bigr)^2}
           {\phi_1^{\beta}(y)}\,dy
    \;+\;
    \int_{1/2}^{\infty}
      \frac{\bigl(\phi_{x_1}^{\beta}(y)-\phi_{x_1'}^{\beta}(y)\bigr)^2}
           {\phi_0^{\beta}(y)}\,dy
  \right\},
\]
where \(\Gamma(s,x)\) denotes the upper incomplete gamma function.
For fixed $\beta>0$, the quantity $\sigma_{\mathrm{BG},\beta}^2$ (and hence $\chi_{\mathrm{lo}}$) can be evaluated efficiently by standard one-dimensional numerical quadrature. 
As a consequence, searching over $(x_1,x_1')$ for the lower shuffle index is computationally inexpensive\footnote{In all our numerical experiments for generalized Gaussian
mechanisms (including the cases $\beta=1$ and $\beta=2$), the pair
$(x_1,x_1')=(0,1)$ consistently attains (or appears to attain) the infimum of the lower shuffle index $\chi_{\mathrm{lo}}(x_1,x_1')$. 
This is also intuitively plausible, since $(0,1)$ maximizes the separation between two inputs in the unit interval. However, a general theoretical proof that $(0,1)$ is the unique global minimizer of $\chi_{\mathrm{lo}}$ for all admissible parameters would require a delicate analysis beyond the scope of this work, and we therefore do not claim such a statement as a theorem.}.

Here, we adopt the working assumption that the worst-case pair is attained at $(x_1,x_1')=(0,1)$, and under this assumption, $\chi_{\mathrm{lo}}$ admits closed-form expressions for the Laplace ($\beta=1$) and Gaussian ($\beta=2$) mechanisms when $\mathrm{Var}(Z)=\sigma_0^2$. 
For Laplace, we have
\[
\chi_{\mathrm{lo}}^{(\beta=1)}
=
\frac{\sqrt{3}}{2}\cdot
\frac{t}{\sqrt{1-3t^2+2t^3}},
\qquad
t:=\exp\!\left(-\frac{1}{\sqrt{2}\,\sigma_0}\right).
\]
For Gaussian, we have
\[
\chi_{\mathrm{lo}}^{(\beta=2)}
=
\frac{1}{\sqrt{
2\left(
e^{1/\sigma_0^2}\,\Phi\!\left(\frac{3}{2\sigma_0}\right)
-3\,\Phi\!\left(\frac{1}{2\sigma_0}\right)
+1
\right)}},
\]
These formulas make the limiting behavior explicit: $\chi_{\mathrm{lo}}^{(\beta=1)}(0,1)\sim \sigma_0/\sqrt{2}$ and $\chi_{\mathrm{lo}}^{(\beta=2)}(0,1)\sim \sigma_0$ as $\sigma_0\to\infty$, while both decay exponentially fast as $\sigma_0\to 0$.

Fig.~\ref{fig:shuffle_index} illustrates the lower shuffle index $\chi_{\mathrm{lo}}$ for several values of the shape parameter~$\beta$.
Even when the underlying local randomizers are calibrated to have the same variance (and hence the same accuracy for mean estimation), they exhibit different shuffle indices.

\begin{figure}[h]
  \centering
  \begin{subfigure}{0.48\linewidth}
    \centering
    \includegraphics[width=\linewidth]{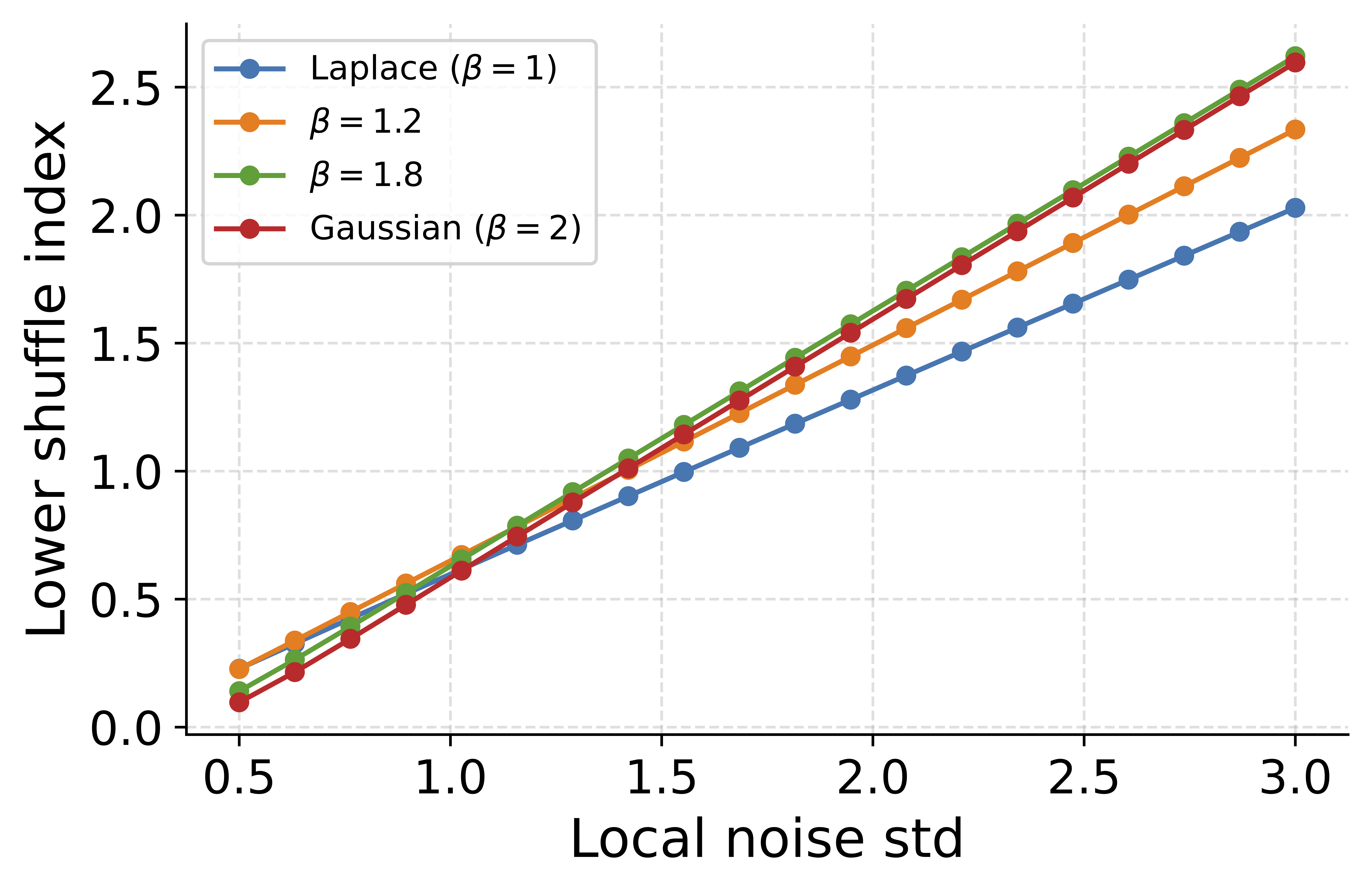}
    \caption{Lower shuffle index for several generalized Gaussian mechanisms.}
    \label{fig:shuffle_index}
  \end{subfigure}
  \hfill
  \begin{subfigure}{0.48\linewidth}
    \centering
    \includegraphics[width=\linewidth]{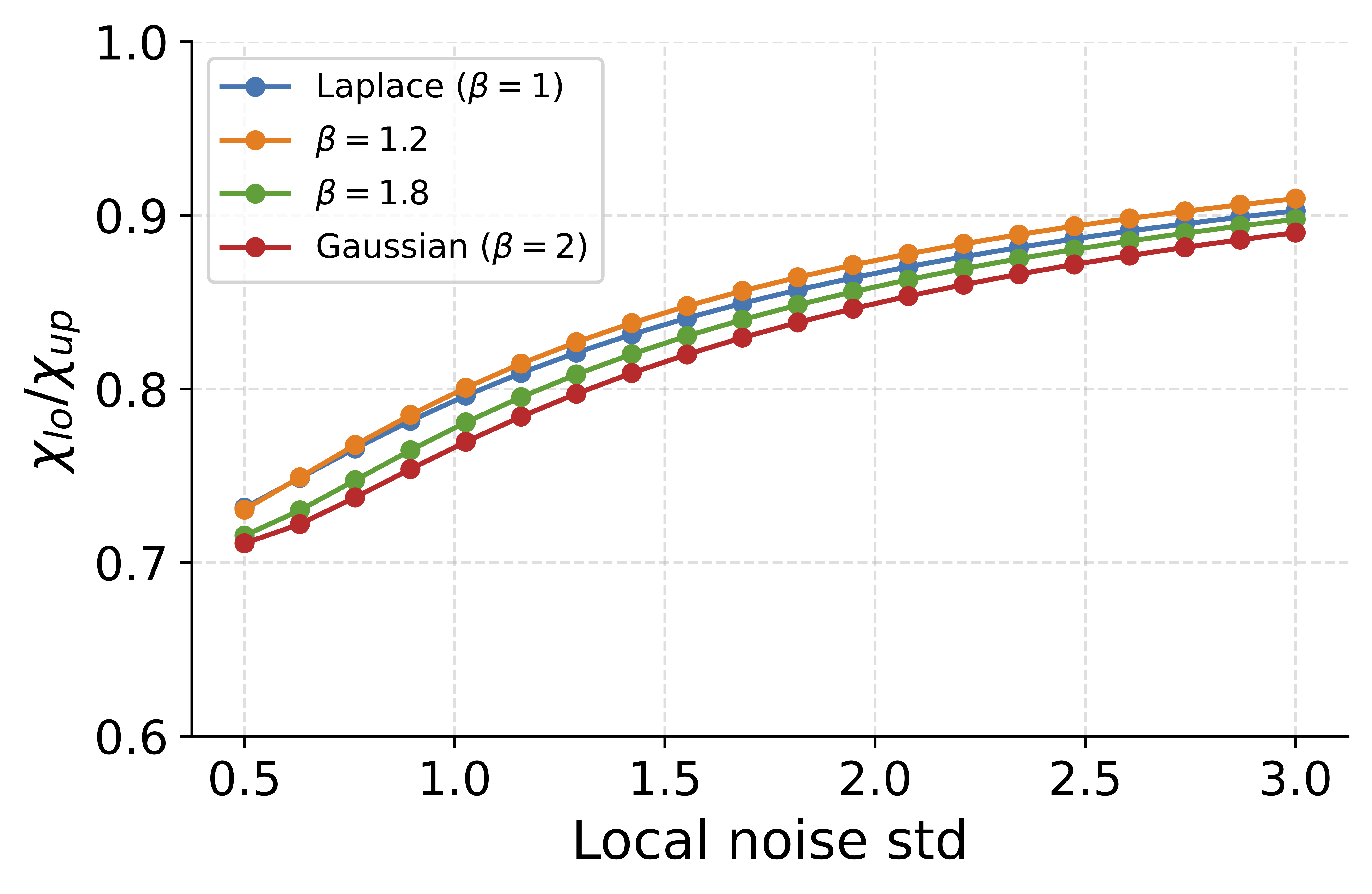}
    \caption{Ratio of lower and upper shuffle indices.}
    \label{fig:shuffle_index_ratio}
  \end{subfigure}

  \caption{
    Shuffle indices for generalized Gaussian mechanisms with different
    shape parameters~$\beta$.
  }
  \label{fig:shuffle_index_all}
\end{figure}

\subsubsection{Upper Shuffle Index}
For the upper shuffle index, we instead take the reference distribution to be one of the local randomizers $\mathcal R_x$, and only the variance is needed in our analysis.
From the definition, the variance is
\begin{equation}
  \label{eq:sigma_of_upper_shuffle_index}
  \sigma_{x,\beta}^2
  :=
  \int_{-\infty}^{\infty}
    \frac{\bigl(\phi_{x_1}^{\beta}(y)-\phi_{x_1'}^{\beta}(y)\bigr)^2}
         {\phi_x^{\beta}(y)}\,dy
  \;=\;
  c \frac{\beta}{2c\,\Gamma(1/\beta)}\,
  F_\beta(a_1,a_1'),
\end{equation}
where
\begin{align}
  \nonumber F_\beta(a_1,a_1')
  :=
    \int_{-\infty}^{\infty}
    \frac{\bigl(e^{-|z-a_1|^\beta}-e^{-|z-a_1'|^\beta}\bigr)^2}
         {e^{-|z|^\beta}}\,dz,
  \qquad
         a_1 := \frac{x_1 - x}{c},
  \qquad
  a_1' := \frac{x_1' - x}{c}
\end{align}
Although the definition of $\chi_{\mathrm{up}}$ involves a
three-point search over $(x_1,x_1',x)$, the fact that
$\{\phi_\mu^\beta\}_{\mu\in\mathbb R}$ forms a location family implies
that $\sigma_{x,\beta}^2$ in
\eqref{eq:sigma_of_upper_shuffle_index} depends on $(x_1,x_1',x)$ only
through the two offsets $(a_1,a_1') = ((x_1-x)/\alpha,(x_1'-x)/\alpha)$.
In particular, the optimization can be reduced to a two-dimensional
search over $(a_1,a_1')$. As in the case of the lower shuffle index,
the variance $\sigma_{x,\beta}^2$ can be efficiently evaluated
by standard one-dimensional numerical quadrature, so this search remains computationally inexpensive.

Under the same worst-case pair assumption $(x_1,x_1')=(0,1)$, the upper shuffle index also admits closed-form expressions for $\beta=1,2$ when the noise is parameterized by $\mathrm{Var}(Z)=\sigma_0^2$. 
For Laplace, taking $\mathcal{R}_\mathrm{ref}=\mathcal{R}_0$ yields
\[
\chi_{\mathrm{up}}^{(\beta=1)}
=
\left(\frac{2}{3}e^{2/\sigma_0^2}-1+\frac{1}{3}e^{-4/\sigma_0^2}\right)^{-1/2}.
\]
For Gaussian, taking $\mathcal{R}_\mathrm{ref}=\mathcal{R}_0$ yields
\[
\chi_{\mathrm{up}}^{(\beta=2)}
=
\left(e^{1/\sigma_0^2}-1\right)^{-1/2}.
\]

In both cases, the tightness ratio $\chi_{\mathrm{lo}}/\chi_{\mathrm{up}}$ exhibits the same qualitative limiting behavior. As $\sigma_0\to\infty$,
\[
\frac{\chi_{\mathrm{lo}}^{(\beta=1)}}{\chi_{\mathrm{up}}^{(\beta=1)}}
=
1-\Theta\!\left(\frac{1}{\sigma_0}\right),
\qquad
\frac{\chi_{\mathrm{lo}}^{(\beta=2)}}{\chi_{\mathrm{up}}^{(\beta=2)}}
=
1-\Theta\!\left(\frac{1}{\sigma_0}\right).
\]
As $\sigma_0\to 0$,
\[
\frac{\chi_{\mathrm{lo}}^{(\beta=1)}}{\chi_{\mathrm{up}}^{(\beta=1)}}
=
\frac{1}{\sqrt2}+o(1),
\qquad
\frac{\chi_{\mathrm{lo}}^{(\beta=2)}}{\chi_{\mathrm{up}}^{(\beta=2)}}
=
\frac{1}{\sqrt2}+o(1).
\]

\subsubsection{Asymptotics of $\chi_{\mathrm{lo}}$.}
Because we already have the asymptotics for $\beta=1$, we focus on the case $\beta>1$.

\noindent\textbf{(i) Large-noise regime $\sigma_0\to \infty$.}
We calibrate $c$ by $\mathrm{Var}(Z)=\sigma_0^2$, i.e.,
\begin{equation}
\label{eq:c-sigma0}
\sigma_0^2=c^2\frac{\Gamma(3/\beta)}{\Gamma(1/\beta)}
\quad\Longleftrightarrow\quad
c=\sigma_0\sqrt{\frac{\Gamma(1/\beta)}{\Gamma(3/\beta)}}.
\end{equation}

For a pair $(x_1,x_1')$, the shuffle index for this pair is
\[
\chi(x_1,x_1')
:=\frac{\sqrt\gamma}{\sigma(x_1,x_1')},
\qquad
\sigma(x_1,x_1')^2
:=\mathrm{Var}_{Y\sim\mathcal R_{\mathrm{BG}}}\!\left(l_0(Y;x_1,x_1',\mathcal R_{\mathrm{BG}})\right).
\]
We have
\begin{equation}
\label{eq:sigma2}
\sigma(x_1,x_1')^2
=
\int_{\mathbb R}\frac{(\phi_{x_1}^\beta(y)-\phi_{x_1'}^\beta(y))^2}{\mathcal R_{\mathrm{BG}}(y)}\,dy
=
\gamma\int_{\mathbb R}\frac{(\phi_{x_1}^\beta(y)-\phi_{x_1'}^\beta(y))^2}{\underline{\mathcal R}(y)}\,dy.
\end{equation}

Fix $\beta>1$ and $x_1,x_1'\in[0,1]$, and set $\Delta:=x_1-x_1'$ and $\eta:=1/c$.
Let
\[
f(z):=\frac{\beta}{2\Gamma(1/\beta)}e^{-|z|^\beta},\qquad z\in\mathbb R,
\]
so that $\int_{\mathbb R}f(z)\,dz=1$ and
\[
\phi_x^\beta(y)=\frac1c\,f\!\left(\frac{y-x}{c}\right).
\]

From the definition of $\gamma$ with $y=cz$, we have
\begin{equation}\label{eq:gamma-eta}
\gamma=\int_{\mathbb R}\underline{\mathcal R}(y)\,dy
=\int_{\mathbb R}\min\{f(z),f(z-\eta)\}\,dz.
\end{equation}
Since $\min\{f(z),f(z-\eta)\}\to f(z)$ pointwise as $\eta\downarrow 0$ and
$0\le \min\{f(z),f(z-\eta)\}\le f(z)\in L^1(\mathbb R)$, dominated convergence implies
\begin{equation}\label{eq:gamma-to-1-eta}
\gamma\to 1 \qquad (c\to\infty),
\quad\text{i.e.}\quad \gamma=1+o(1).
\end{equation}

In what follows, we show that as $c\to\infty$,
\begin{equation}
\label{eq:sigma2-asymp}
\sigma(x_1,x_1')^2
=
\gamma\int_{\mathbb R}\frac{(\phi_{x_1}^\beta(y)-\phi_{x_1'}^\beta(y))^2}{\underline{\mathcal R}(y)}\,dy
=
\Delta^2
\int_{\mathbb R}
\frac{\bigl(\partial_y\phi_0^\beta(y)\bigr)^2}{\phi_0^\beta(y)}
\,dy\;
\bigl(1+o(1)\bigr).
\end{equation}

Since $\beta>1$, we have $f\in C^1(\mathbb R)$ and hence $\phi_0^\beta\in C^1(\mathbb R)$.
Moreover,
\[
\partial_y\phi_0^\beta(y)=\frac1{c^2}f'\!\left(\frac{y}{c}\right).
\]
Therefore, by the change of variables $z=y/c$,
\begin{equation}\label{eq:fisher-scale-eta}
\int_{\mathbb R}\frac{(\partial_y\phi_0^\beta(y))^2}{\phi_0^\beta(y)}\,dy
=\int_{\mathbb R}\frac{c^{-4}f'(y/c)^2}{c^{-1}f(y/c)}\,dy
=\frac1{c^2}\int_{\mathbb R}\frac{f'(z)^2}{f(z)}\,dz.
\end{equation}
Since $\beta>1$, the integrand is bounded near $0$ and decays exponentially as $|z|\to\infty$, hence
\begin{equation}\label{eq:Ibeta-eta}
\mathcal I_\beta:=\int_{\mathbb R}\frac{f'(z)^2}{f(z)}\,dz<\infty.
\end{equation}

A direct change of variables $y=cz$ gives the exact identity
\begin{equation}\label{eq:chi2-scale-eta}
\int_{\mathbb R}\frac{\bigl(\phi_{x_1}^\beta(y)-\phi_{x_1'}^\beta(y)\bigr)^2}{\underline{\mathcal R}(y)}\,dy
=\eta^2\int_{\mathbb R} g_\eta(z)\,dz.
\end{equation}
Here, we defined
\[
m_\eta(z):=\min\{f(z),f(z-\eta)\},\qquad
d_\eta(z):=\frac{f(z-x_1\eta)-f(z-x_1'\eta)}{\eta},
\qquad
g_\eta(z):=\frac{d_\eta(z)^2}{m_\eta(z)}.
\]

Since $f\in C^1(\mathbb R)$, for every $z\in\mathbb R$,
\[
\lim_{\eta\downarrow 0} d_\eta(z)
=\lim_{\eta\downarrow 0}\frac{f(z-x_1\eta)-f(z-x_1'\eta)}{\eta}
=-(x_1-x_1')f'(z)=-\Delta f'(z).
\]
Also $m_\eta(z)\to f(z)>0$. Hence
\begin{equation}\label{eq:g-pointwise-eta}
g_\eta(z)\to \Delta^2\frac{f'(z)^2}{f(z)}
\qquad(\eta\downarrow 0)\quad\text{for all }z\in\mathbb R.
\end{equation}
Therefore, if the limit and the integral can be interchanged, we would obtain
\[
\lim_{\eta \downarrow 0} \int_{\mathbb R} g_\eta(z) \, dz
= \int_{\mathbb R} \left( \lim_{\eta \downarrow 0} g_\eta(z) \right) \, dz
= \Delta^2 \int_{\mathbb R} \frac{f'(z)^2}{f(z)} \, dz
= \Delta^2 \mathcal{I}_\beta,
\]
from which the desired asymptotic formula \eqref{eq:sigma2-asymp} follows immediately.
Thus, the remainder of the proof is devoted to justifying this exchange by applying Lebesgue's dominated convergence theorem to $g_\eta$. Specifically, we will construct an integrable dominating function $H \in L^1(\mathbb R)$ such that $0 \le g_\eta(z) \le H(z)$ for all $z \in \mathbb R$ and all sufficiently small $\eta > 0$.

Fix $\eta_0\in(0,1)$ and restrict to $\eta\in(0,\eta_0)$.
By the fundamental theorem of calculus,
\[
f(z-x_1\eta)-f(z-x_1'\eta)
=-\eta\int_{x_1'}^{x_1} f'(z-t\eta)\,dt,
\]
hence
\[
d_\eta(z)=-\int_{x_1'}^{x_1} f'(z-t\eta)\,dt.
\]
Let $a:=\min\{x_1,x_1'\}$ and $b:=\max\{x_1,x_1'\}$ so that $b-a=|\Delta|$.
Then, by Cauchy--Schwarz,
\begin{equation}\label{eq:CS-eta}
d_\eta(z)^2
\le (b-a)\int_a^b f'(z-t\eta)^2\,dt
\le |\Delta| \int_0^1 f'(z-t\eta)^2\,dt.
\end{equation}

Since $f$ is even and strictly decreasing on $[0,\infty)$, and since
$\max\{|z|,|z-\eta|\}\le |z|+\eta$, we have
\begin{equation}\label{eq:min-lower-eta}
m_\eta(z)=\min\{f(z),f(z-\eta)\}\ge f(|z|+\eta)\ge f(|z|+\eta_0).
\end{equation}

Fix $R\ge 2$. Since $f$ is continuous and strictly positive, let
\[
m_R:=\inf_{|u|\le R+1} f(u)>0,
\qquad
L_R:=\sup_{|u|\le R+1}|f'(u)|<\infty.
\]
For $|z|\le R$ and $\eta\le\eta_0<1$, we have $|z-t\eta|\le R+1$, so
\[
|d_\eta(z)|\le \int_0^1 |f'(z-t\eta)|\,dt \le L_R,
\qquad
m_\eta(z)\ge m_R,
\]
and thus
\begin{equation}\label{eq:compact-dom-eta}
0\le g_\eta(z)\le \frac{L_R^2}{m_R}=:C_R
\qquad\text{for all }|z|\le R,\ \eta\in(0,\eta_0).
\end{equation}

For $|z|\ge R$ and $\eta\le\eta_0$, combine \eqref{eq:CS-eta}--\eqref{eq:min-lower-eta} to get
\begin{equation}\label{eq:tail-start-eta}
g_\eta(z)
\le \frac{|\Delta|\int_0^1 f'(z-t\eta)^2\,dt}{f(|z|+\eta_0)}
\le \frac{|\Delta|\,\sup_{t\in[0,1]} f'(z-t\eta)^2}{f(|z|+\eta_0)}.
\end{equation}
Write $u=z-t\eta$. Then $||u|-|z||\le \eta\le \eta_0$, hence for $|z|\ge R\ge 2$,
\[
|z|-\eta_0 \le |u|\le |z|+\eta_0.
\]
Using $f'(u)^2=\beta^2|u|^{2\beta-2}f(u)^2$ and $f(v)=Ce^{-|v|^\beta}$ with
$C=\beta/(2\Gamma(1/\beta))$, we obtain for all such $u$,
\[
\frac{f'(u)^2}{f(|z|+\eta_0)}
=\beta^2|u|^{2\beta-2}\,C\,\exp\!\Bigl((|z|+\eta_0)^\beta-2|u|^\beta\Bigr)
\le K_1(1+|z|^{2\beta-2})\exp\!\Bigl((|z|+\eta_0)^\beta-2(|z|-\eta_0)^\beta\Bigr),
\]
for a constant $K_1$ depending only on $\beta$ and $\eta_0$.

We claim that there exists $R_0=R_0(\beta,\eta_0)$ such that for all $r\ge R_0$,
\begin{equation}\label{eq:exp-neg-eta}
(r+\eta_0)^\beta-2(r-\eta_0)^\beta \le -\tfrac12 r^\beta.
\end{equation}
Indeed, set $h(r):=(r+\eta_0)^\beta-2(r-\eta_0)^\beta+\tfrac12 r^\beta$.
Since $\beta>1$, we have $h(r)/r^\beta\to -\tfrac12$ as $r\to\infty$, hence $h(r)\le 0$ for all $r\ge R_0$.

Taking $R\ge R_0$ and applying \eqref{eq:exp-neg-eta} with $r=|z|$ in \eqref{eq:tail-start-eta} gives
\begin{equation}\label{eq:tail-dom-eta}
0\le g_\eta(z)\le K_2(1+|z|^{2\beta-2})e^{-|z|^\beta/2}
=:H_{\mathrm{tail}}(z),
\qquad |z|\ge R,\ \eta\in(0,\eta_0),
\end{equation}
for some constant $K_2$.
Since $(1+|z|^{2\beta-2})e^{-|z|^\beta/2}\in L^1(\{|z|\ge R\})$, we have $H_{\mathrm{tail}}\in L^1(\{|z|\ge R\})$.

Define
\[
H(z):=C_R\,\mathbf 1_{\{|z|\le R\}}+H_{\mathrm{tail}}(z)\,\mathbf 1_{\{|z|\ge R\}}.
\]
Then $H\in L^1(\mathbb R)$ and, by \eqref{eq:compact-dom-eta} and \eqref{eq:tail-dom-eta},
\begin{equation}\label{eq:dominate-eta}
0\le g_\eta(z)\le H(z)
\qquad\text{for all }z\in\mathbb R,\ \eta\in(0,\eta_0).
\end{equation}

By \eqref{eq:g-pointwise-eta} and \eqref{eq:dominate-eta}, dominated convergence yields
\begin{equation}\label{eq:g-limit-eta}
\int_{\mathbb R} g_\eta(z)\,dz
\longrightarrow
\Delta^2\int_{\mathbb R}\frac{f'(z)^2}{f(z)}\,dz
=\Delta^2\mathcal I_\beta
\qquad(\eta\downarrow 0).
\end{equation}

Combining \eqref{eq:chi2-scale-eta} and \eqref{eq:g-limit-eta} gives
\[
\int_{\mathbb R}\frac{\bigl(\phi_{x_1}^\beta(y)-\phi_{x_1'}^\beta(y)\bigr)^2}{\underline{\mathcal R}(y)}\,dy
=\eta^2\bigl(\Delta^2\mathcal I_\beta+o(1)\bigr)
=\frac{\Delta^2\mathcal I_\beta}{c^2}\,(1+o(1)).
\]
Multiplying by $\gamma=1+o(1)$ from \eqref{eq:gamma-to-1-eta} yields
\[
\sigma(x_1,x_1')^2
=\gamma\int_{\mathbb R}\frac{\bigl(\phi_{x_1}^\beta(y)-\phi_{x_1'}^\beta(y)\bigr)^2}{\underline{\mathcal R}(y)}\,dy
=\frac{\Delta^2\mathcal I_\beta}{c^2}\,(1+o(1)).
\]
Finally, by \eqref{eq:fisher-scale-eta} and \eqref{eq:Ibeta-eta},
\[
\frac{\mathcal I_\beta}{c^2}
=\int_{\mathbb R}\frac{(\partial_y\phi_0^\beta(y))^2}{\phi_0^\beta(y)}\,dy,
\]
and therefore
\[
\sigma(x_1,x_1')^2
=
\Delta^2\int_{\mathbb R}\frac{(\partial_y\phi_0^\beta(y))^2}{\phi_0^\beta(y)}\,dy\,
(1+o(1)),
\qquad c\to\infty,
\]
which is the desired asymptotic.

Since $f'(u)=-\beta\,\mathrm{sign}(u)\,|u|^{\beta-1}f(u)$,
\[
\int \frac{(f')^2}{f}
=
\beta^2\int |u|^{2\beta-2}f(u)\,du
=
\beta^2\frac{\beta}{2\Gamma(1/\beta)}\cdot 2\int_0^\infty u^{2\beta-2}e^{-u^\beta}\,du
=
\beta^2\frac{\Gamma(2-1/\beta)}{\Gamma(1/\beta)}.
\]
That is,
\begin{align}
\int_{\mathbb R}\frac{(\partial_y \phi_0^\beta(y))^2}{\phi_0^\beta(y)}\,dy
&=
\frac{1}{c^2}\int_{\mathbb R}\frac{(f'(z))^2}{f(z)}\,dz
\label{eq:fisher}\\
&=
\frac{1}{c^2}\,\beta^2\,\frac{\Gamma(2-1/\beta)}{\Gamma(1/\beta)}.
\nonumber
\end{align}
Therefore, for $|\Delta|=1$,
\[
\sigma(0,1)
=
\frac{\beta}{c}\sqrt{\frac{\Gamma(2-1/\beta)}{\Gamma(1/\beta)}}\,(1+o(1)),
\qquad
\chi_{\mathrm{lo}}(\mathcal R)
=
\frac{\sqrt\gamma}{\sigma(0,1)}
=
\frac{c}{\beta}\sqrt{\frac{\Gamma(1/\beta)}{\Gamma(2-1/\beta)}}\,(1+o(1)).
\]
Substituting \eqref{eq:c-sigma0}, we obtain the linear growth
\begin{equation}
\label{eq:chi-large-noise}
\chi_{\mathrm{lo}}(\mathcal R)
=
c_\beta\,\sigma_0\,(1+o(1)),
\qquad
c_\beta
:=
\frac{\Gamma(1/\beta)}{\beta\sqrt{\Gamma(3/\beta)\Gamma(2-1/\beta)}}.
\end{equation}

\noindent\textbf{(ii) Small-noise regime $\sigma_0\to 0$.}
Let
\[
K:=\frac{\beta}{2\Gamma(1/\beta)},\qquad
\phi_x(y)=\frac{K}{c}\exp\!\Bigl(-\bigl|\tfrac{y-x}{c}\bigr|^\beta\Bigr).
\]
For the pair $(0,1)$, the variance term in the lower shuffle index satisfies
\begin{equation}\label{eq:sigma2-lower}
\sigma^2
=
\mathrm{Var}_{Y\sim \mathcal R_{\mathrm{BG}}}\!\left(\frac{\phi_0(Y)-\phi_1(Y)}{\mathcal R_{\mathrm{BG}}(Y)}\right)
=
\gamma\left[
\int_{-\infty}^{1/2}\frac{(\phi_0-\phi_1)^2}{\phi_1}\,dy
+
\int_{1/2}^{\infty}\frac{(\phi_0-\phi_1)^2}{\phi_0}\,dy
\right].
\end{equation}
Hence,
\[
\sigma^2 \ge \gamma\int_0^c \frac{(\phi_0-\phi_1)^2}{\phi_1}\,dy.
\]
For $y\in[0,c]$,
\[
\frac{\phi_1(y)}{\phi_0(y)}
=
\exp\!\left(-\frac{(1-y)^\beta-y^\beta}{c^\beta}\right).
\]
Choose $c_0\in(0,1/4)$ small enough so that for all $c\in(0,c_0)$ and all $y\in[0,c]$,
\[
(1-y)^\beta-y^\beta \ge (1-c)^\beta-c^\beta \ge \frac12.
\]
Then for such $c$ and $y$ we have $\phi_1(y)/\phi_0(y)\le e^{-1/(2c^\beta)}\le 1/2$, and therefore
\[
(\phi_0(y)-\phi_1(y))^2 \ge \frac14\,\phi_0(y)^2.
\]
It follows that, for all $c\in(0,c_0)$,
\begin{equation}\label{eq:sigma2-lower2}
\sigma^2
\ge
\frac{\gamma}{4}\int_0^c \frac{\phi_0(y)^2}{\phi_1(y)}\,dy.
\end{equation}
Next, for $y\in[0,c]$,
\[
\frac{\phi_0(y)^2}{\phi_1(y)}
=
\frac{K}{c}\exp\!\left(\frac{(1-y)^\beta-2y^\beta}{c^\beta}\right)
\ge
\frac{K}{c}\exp\!\left(\frac{(1-c)^\beta}{c^\beta}-2\right).
\]
Possibly shrinking $c_0$ further, we may ensure $(1-c)^\beta\ge 1/2$ for all $c\in(0,c_0)$, so that
\[
\frac{\phi_0(y)^2}{\phi_1(y)}
\ge
\frac{K}{c}\exp\!\left(\frac{1}{2c^\beta}-2\right)
\qquad \text{for all }y\in[0,c],\ c\in(0,c_0).
\]
Plugging this into \eqref{eq:sigma2-lower2} yields
\[
\sigma^2
\ge
\frac{\gamma}{4}\int_0^c \frac{K}{c}\exp\!\left(\frac{1}{2c^\beta}-2\right)\,dy
=
\gamma\cdot \frac{K}{4}\exp(-2)\cdot \exp\!\left(\frac{1}{2c^\beta}\right).
\]
Therefore,
\[
\chi_{\mathrm{lo}}=\frac{\sqrt{\gamma}}{\sigma}
\le
\frac{\sqrt{\gamma}}{\sqrt{\gamma}\,\sqrt{(K/4)e^{-2}}\,\exp\!\left(\frac{1}{4c^\beta}\right)}
=
C\,\exp\!\left(-\frac{1}{4c^\beta}\right)
\]
for some constant $C>0$ and all $c\in(0,c_0)$.
Hence $1/c^\beta = \Theta(1/\sigma_0^\beta)$, so there exist $A>0$ and $\sigma_\star>0$ such that
\[
\chi_{\mathrm{lo}} \le \exp\!\Bigl(-\frac{A}{\sigma_0^\beta}\Bigr)
\qquad\text{for all }\sigma_0\in(0,\sigma_\star).
\]

\subsubsection{The Dimensional Dependence}
\label{app:dim-dependence-gamma}
We consider the $\beta=2$-Gaussian local randomizer on the $d$-dimensional
$\ell_2$-ball $\mathcal X=\{x\in\mathbb R^d:\|x\|_2\le r\}$ with isotropic noise
$Z\sim\mathcal N(0,\sigma_0^2 I_d)$:
\[
\mathcal R(x)=x+Z,\qquad 
\phi_x(y)=\frac{1}{(2\pi\sigma_0^2)^{d/2}}
\exp\!\left(-\frac{\|y-x\|_2^2}{2\sigma_0^2}\right).
\]
The blanket density is $\underline{\mathcal R}(y)=\inf_{x\in\mathcal X}\phi_x(y)$.
Since $\phi_x(y)$ decreases with $\|y-x\|_2$, the infimum is attained at a point
$x^\star(y)\in\mathcal X$ that maximizes $\|y-x\|_2$. Over the Euclidean ball,
the farthest point from $y$ is the boundary point opposite to $y$, hence
\[
\max_{\|x\|_2\le r}\|y-x\|_2=\|y\|_2+r,
\qquad
x^\star(y)=-r\,\frac{y}{\|y\|_2}\ \ (y\neq 0).
\]
Therefore,
\[
\underline{\mathcal R}(y)
=\frac{1}{(2\pi\sigma_0^2)^{d/2}}
\exp\!\left(-\frac{(\|y\|_2+r)^2}{2\sigma_0^2}\right).
\]
Let $\phi_0$ denote the $\mathcal N(0,\sigma_0^2 I_d)$ density. Expanding
$(\|y\|_2+r)^2=\|y\|_2^2+2r\|y\|_2+r^2$ gives
\[
\underline{\mathcal R}(y)
=\phi_0(y)\exp\!\left(-\frac{r}{\sigma_0^2}\|y\|_2-\frac{r^2}{2\sigma_0^2}\right).
\]
Hence the blanket mass
\[
\gamma:=\int_{\mathbb R^d}\underline{\mathcal R}(y)\,dy
=\exp\!\left(-\frac{r^2}{2\sigma_0^2}\right)\,
\mathbb E_{Y\sim\mathcal N(0,\sigma_0^2 I_d)}
\Bigl[\exp\!\bigl(-\tfrac{r}{\sigma_0^2}\|Y\|_2\bigr)\Bigr].
\]
Writing $Y=\sigma_0 G$ with $G\sim\mathcal N(0,I_d)$ and setting
$a:=r/\sigma_0>0$ yields
\begin{equation}
\label{eq:gamma-laplace-norm}
\gamma
=\exp\!\left(-\frac{r^2}{2\sigma_0^2}\right)\,
\mathbb E\bigl[e^{-a\|G\|_2}\bigr].
\end{equation}

We now derive a simple upper bound showing $\gamma=\exp(-\Omega(\sqrt d))$ for
fixed $(r,\sigma_0)$. Let $t:=\frac12\sqrt d$. Splitting on the event
$\{\|G\|_2>t\}$ gives
\[
\mathbb E[e^{-a\|G\|_2}]
=\mathbb E\!\left[e^{-a\|G\|_2}\mathbf 1_{\{\|G\|_2>t\}}\right]
+\mathbb E\!\left[e^{-a\|G\|_2}\mathbf 1_{\{\|G\|_2\le t\}}\right]
\le e^{-at}+\Pr(\|G\|_2\le t),
\]
since $e^{-a\|G\|_2}\le e^{-at}$ on $\{\|G\|_2>t\}$ and
$e^{-a\|G\|_2}\le 1$ everywhere. It remains to bound the left-tail probability.
Because the map $g\mapsto \|g\|_2$ is $1$-Lipschitz, Gaussian concentration
implies that for all $s\ge 0$,
\begin{equation}
\label{eq:gauss-conc-norm}
\Pr\bigl(\|G\|_2\le \mathbb E\|G\|_2-s\bigr)\le e^{-s^2/2}.
\end{equation}
Moreover, $\mathbb E\|G\|_2\ge \sqrt{d-\tfrac12}$. For $d\ge 8$ this yields
\[
s:=\mathbb E\|G\|_2-t
\ge \sqrt{d-\tfrac12}-\tfrac12\sqrt d
\ge \tfrac14\sqrt d,
\]
and therefore, by~\eqref{eq:gauss-conc-norm},
\[
\Pr(\|G\|_2\le t)\le \exp\!\left(-\frac{s^2}{2}\right)
\le \exp\!\left(-\frac{d}{32}\right).
\]
Combining the bounds, we obtain for $d\ge 8$,
\[
\mathbb E[e^{-a\|G\|_2}]
\le \exp\!\left(-\frac{a}{2}\sqrt d\right)+\exp\!\left(-\frac{d}{32}\right)
\le 2\exp\!\left(-\frac{a}{2}\sqrt d\right),
\]
where the last inequality holds for all sufficiently large $d$ since
$e^{-d/32}\le e^{-(a/2)\sqrt d}$. Substituting this into~\eqref{eq:gamma-laplace-norm}
gives
\[
\gamma
\le 2\exp\!\left(-\frac{r^2}{2\sigma_0^2}\right)
\exp\!\left(-\frac{r}{2\sigma_0}\sqrt d\right)
= C(r,\sigma_0)\exp\!\left(-\frac{r}{2\sigma_0}\sqrt d\right),
\]
for a constant $C(r,\sigma_0)=2\exp(-r^2/(2\sigma_0^2))$ independent of $d$.
In particular, for fixed $(r,\sigma_0)$ we have $\gamma=\exp(-\Omega(\sqrt d))$.

\subsection{$k$-RR}

We next specialize the shuffle indices to the $k$-RR. For an input $x\in[k]:=\{1,\dots,k\}$, the output distribution is
\[
  \phi_x(y)
  =
  \Pr[\mathcal R(x)=y]
  =
  \begin{cases}
    p, & y=x,\\[2pt]
    q, & y\neq x,
  \end{cases}
  \qquad
  p=\dfrac{\mathrm e^{\varepsilon_0}}{\mathrm e^{\varepsilon_0}+k-1},\quad
  q=\dfrac{1}{\mathrm e^{\varepsilon_0}+k-1}.
\]
The blanket distribution is uniform on $[k]$:
$
  \phi_{\mathrm{BG}}(y)
  = \frac{1}{k},
  \gamma=\sum_{y}\min_x\phi_x(y)
  = kq.
$

\subsubsection{Lower shuffle index.}
For any distinct pair $x_1\neq x_1'\in[k]$, the privacy amplification
random variable with respect to the blanket distribution is
\[
  l_0(y;x_1,x_1',\mathcal R_{\mathrm{BG}})
  = \frac{\phi_{x_1}(y)-\phi_{x_1'}(y)}{\phi_{\mathrm{BG}}(y)}
  = k\bigl(\phi_{x_1}(y)-\phi_{x_1'}(y)\bigr).
\]
A simple case distinction shows that, for $Y\sim\mathcal R_{\mathrm{BG}}$,
\[
  l_0(Y;x_1,x_1',\mathcal R_{\mathrm{BG}})
  =
  \begin{cases}
    k(p-q), & Y=x_1,\\
    -k(p-q), & Y=x_1',\\
    0, & \text{otherwise},
  \end{cases}
\]
so that $\mathbb E[l_0(Y)]=0$ and
$
  \sigma_{\mathrm{BG}}^2
  := \mathrm{Var}\bigl(l_0(Y;x_1,x_1',\mathcal R_{\mathrm{BG}})\bigr)
  = 2k(p-q)^2.
$
By definition,
\[
  \chi
  = \sqrt{\frac{\gamma}{\sigma_{\mathrm{BG}}^2}}
  = \sqrt{\frac{q}{2(p-q)^2}}
  = \sqrt{\frac{\mathrm e^{\varepsilon_0}+k-1}{2\bigl(\mathrm e^{\varepsilon_0}-1\bigr)^2}}.
\]
In particular, $\chi$ does not depend on the
particular choice of the distinct pair $(x_1,x_1')$, and hence
\[
  \chi_{\mathrm{lo}}
  = \inf_{x_1\neq x_1'} \chi
  = \sqrt{\frac{q}{2(p-q)^2}}.
\]

\subsubsection{Upper shuffle index.}
For the upper shuffle index we take the reference distribution to be one of
the local randomizers $\mathcal R_x$. Fix $x_1\neq x_1'$, and consider
\[
  l_0(y;x_1,x_1',\mathcal R_x)
  = \frac{\phi_{x_1}(y)-\phi_{x_1'}(y)}{\phi_x(y)},
  \qquad Y\sim\mathcal R_x.
\]
By symmetry, the variance only depends on how $x$ is positioned relative to
$\{x_1,x_1'\}$.

\begin{itemize}
  \item If $x=x_1$ (or symmetrically $x=x_1'$), then
  \[
    l_0(Y;x_1,x_1',\mathcal R_{x_1})
    =
    \begin{cases}
      \dfrac{p-q}{p}, & Y=x_1,\\[4pt]
      -\dfrac{p-q}{q}, & Y=x_1',\\[4pt]
      0, & \text{otherwise},
    \end{cases}
  \]
  so that
  $
    \mathrm{Var}\bigl(l_0(Y;x_1,x_1',\mathcal R_{x_1})\bigr)
    = (p-q)^2\Bigl(\frac{1}{p}+\frac{1}{q}\Bigr).
  $

  \item If $k\ge3$ and $x\notin\{x_1,x_1'\}$, then
  $
    \mathrm{Var}\bigl(l_0(Y;x_1,x_1',\mathcal R_x)\bigr)
    = 2\,\frac{(p-q)^2}{q}.
  $
\end{itemize}
Since for $p>q$ we have
$
  2\,\frac{(p-q)^2}{q}
  \;>\;
  (p-q)^2\Bigl(\frac{1}{p}+\frac{1}{q}\Bigr),
$
the variance is maximized (and hence $\chi$ is minimized) by choosing
$x\notin\{x_1,x_1'\}$ whenever $k\ge3$. Because $\gamma=1$ for
$\mathcal R_x$, we obtain, for $k\ge3$,
\[
  \chi_{\mathrm{up}}
  = \inf_{x\in[k]} \sqrt{\frac{1}{\mathrm{Var}(l_0(Y;x_1,x_1',\mathcal R_x))}}
  = \sqrt{\frac{q}{2(p-q)^2}}
\]
Thus, $\chi_{\mathrm{up}}$ also does not depend on the particular
choice of the distinct pair $(x_1,x_1')$, and hence, for $k$-randomized response with $k\ge3$, the lower and upper
shuffle indices coincide:
\[
  \chi_{\mathrm{up}} = \chi_{\mathrm{lo}}
  = \sqrt{\frac{q}{2(p-q)^2}}
  = \sqrt{\frac{\mathrm e^{\varepsilon_0}+k-1}{2\bigl(\mathrm e^{\varepsilon_0}-1\bigr)^2}}.
\]

\subsubsection{Comparison to the Clone Paradigm}
\label{app:krr_comparison_to_clone}

Here, we consider Theorem~7.4 of Feldman~\cite{feldmanStrongerPrivacyAmplification2023a}: the specific clone reduction of $k$RR.
Let the categorical output alphabet be $\mathcal A:=\{0,1,2,3\}$.
Fix parameters $p=\frac{1}{e^{\varepsilon_0}+k-1}$ and $q=\frac{k-2}{e^{\varepsilon_0}+k-1}$.
Define the \emph{reference} categorical distribution $Q_{\mathrm{ref}}$ on $\mathcal A$ by
\[
Q_{\mathrm{ref}}(0)=p,\qquad
Q_{\mathrm{ref}}(1)=p,\qquad
Q_{\mathrm{ref}}(2)=q,\qquad
Q_{\mathrm{ref}}(3)=1-2p-q.
\]
Define the two \emph{target} categorical distributions $Q_0,Q_1$ on $\mathcal A$:
\[
Q_{0}(0)=e^{\varepsilon_0}p,\qquad
Q_{0}(1)=p,\qquad
Q_{0}(2)=1-(1+e^{\varepsilon_0})p,\qquad
Q_0(3)=0,
\]
\[
Q_{1}(0)=p,\qquad
Q_{1}(1)=e^{\varepsilon_0}p,\qquad
Q_{1}(2)=1-(1+e^{\varepsilon_0})p,\qquad
Q_1(3)=0.
\]
The reduced shuffled mechanism is then
\[
\widehat{\mathcal M}_b
\;:=\;
\mathcal S\bigl(Z_1^{(b)},Z_2,\dots,Z_n\bigr),
\qquad
Z_1^{(b)}\sim Q_b,\ \ Z_2,\dots,Z_n\stackrel{\mathrm{i.i.d.}}{\sim}Q_{\mathrm{ref}},
\qquad b\in\{0,1\}.
\]
That is, one \emph{target} user draws a single categorical message from a distribution $Q_b$,
while each of the remaining $n-1$ users draws i.i.d.\ messages from a \emph{reference} categorical
distribution $Q_{\mathrm{ref}}$; all $n$ messages are then passed through an ideal shuffler.
From the discussion in Section~\ref{remark:clone_paradigm}, the hockey-stick divergence between the two shuffled outputs can be expressed as the blanket divergence with $\gamma=1$.
That is, we can interpret that Theorem~7.4 states that the blanket divergence upper boundeds the hockey-stick divergence of shuffled $k$RR.
Under this interpretation, we may analyze the reduced categorical instance with $\gamma=1$ and identify the corresponding shuffle index.

Writing $a:=e^{\varepsilon_0}$, we obtain the pointwise values of the privacy amplification random variable $l_0(y)$:
\[
l_{0}(0)=\frac{ap-p}{p}=a-1,\qquad
l_{0}(1)=\frac{p-ap}{p}=1-a=-(a-1),\qquad
l_{0}(2)=0,\qquad
l_{0}(3)=0,
\]
since $Q_0(2)=Q_1(2)$ and $Q_0(3)=Q_1(3)=0$.
Therefore, under $Y\sim Q_{\mathrm{ref}}$,
\[
l_0(Y)=
\begin{cases}
+(a-1) & \text{w.p. }p,\\
-(a-1) & \text{w.p. }p,\\
0 & \text{w.p. }1-2p,
\end{cases}
\]
which implies
\[
\sigma^2=\mathrm{Var}(l_0(Y))=\mathbb E[l_0(Y)^2]=2p\,(a-1)^2.
\]
Consequently, the shuffle index for this $\gamma=1$ categorical instance is
\[
\chi \;=\; \frac{1}{\sqrt{2p}\,(a-1)}.
\]
Finally, one has $p=\frac{1}{a+k-1}$ (equivalently, $p=\frac{1}{e^{\varepsilon_0}+k-1}$), hence
\[
\chi \;=\;
\frac{\sqrt{e^{\varepsilon_0}+k-1}}{\sqrt{2}\,\bigl(e^{\varepsilon_0}-1\bigr)}.
\]
This corresponds to $\chi_\mathrm{lo}$ of $k$RR computed in the previous section.

\section{Implementations of FFT algorithms}
\label{sec:fft-implementation-details}

\subsection{Parameters}

\label{sec:fft-parameter-selection}
The error terms in Theorem~\ref{thm:fft-blanket-error} admit explicit and
rigorous upper bounds.
Let 
\[
  M \sim 1 + \mathrm{Bin}(n-1,\gamma)
  \quad\text{as}\quad
  M = 1 + \sum_{i=2}^n B_i,
\]
where $B_2,\dots,B_n$ are i.i.d.\ $\mathrm{Bernoulli}(\gamma)$ variables,
independent of the blanket samples. In particular,
\[
  \sum_{i=2}^M Z_i
  \;\stackrel{d}{=}\;
  \sum_{i=2}^n B_i Z_i,
\]
so every error event can be expressed as a tail event of a sum of independent variables of the form $B_i Z_i$.

\smallskip

\emph{(i) Truncation error $E_{\mathrm{trunc}}$.}
Let $q := \Pr[l(Y)\notin[s,s+w^{\mathrm{in}}]] =
F_{\mathrm {ref}}(s) + 1 - F_{\mathrm {ref}}(s+w^{\mathrm{in}})$, where
$F_{\mathrm {ref}}$ is the blanket CDF. Then the truncation error is
\[
  E_{\mathrm{trunc}}
  :=
  \Pr\bigl[\perp\bigr]
  =
  1 - \mathbb E\bigl[(1-q)^{M-1}\bigr]
  =
  1 - (1-\gamma q)^{n-1},
\]
which is computable exactly from $F_{\mathrm {ref}}$ and the numerical
parameters $(n,\gamma,s,w^{\mathrm{in}})$. In particular,
$E_{\mathrm{trunc}}\le (n-1)\gamma q$ by the union bound.

\smallskip

\emph{(ii) Discretization error $E_{\mathrm{disc}}$.}
Let $D_i := Z_i^{\mathrm{tr}} - Z_i^{\mathrm{di}}$ and observe that,
by construction, $|D_i|\le h/2$ almost surely. Using the selector
representation we may write
\[
  \sum_{i=2}^M D_i
  \;\stackrel{d}{=}\;
  \sum_{i=2}^n B_i D_i,
\]
where $B_2,\dots,B_n$ are i.i.d.\ $\mathrm{Bernoulli}(\gamma)$
variables independent of $(D_i)_{i\ge2}$. Define
\[
  X_i := B_i D_i - \mathbb E[B_i D_i],\qquad
  S_{\mathrm{disc}} := \sum_{i=2}^n X_i.
\]
Then the variables $X_i$ are independent and bounded as
\[
  |X_i|
  \;\le\;
  K_{\mathrm{disc}} := \frac{h}{2}
  \quad\text{a.s.}
\]
Moreover, since $|D_i|\le h/2$ we have
\[
  \mathrm{Var}(X_i)
  = \mathrm{Var}(B_i D_i)
  \le \gamma\,\mathbb E[D_i^2]
  \le \gamma\,\frac{h^2}{4}
  =: v_{\mathrm{disc}}^2.
\]
Hence Bernstein's inequality yields, for any $a>0$,
\[
  \Pr\bigl[|S_{\mathrm{disc}}|>a\bigr]
  \;\le\;
  2\exp\!\left(
    -\frac{a^2}{
      2(n-1)v_{\mathrm{disc}}^2
      + \frac{2}{3}K_{\mathrm{disc}} a
    }
  \right)
  =
  2\exp\!\left(
    -\frac{a^2}{
      2(n-1)\gamma h^2/4
      + \frac{2}{3}(h/2)a
    }
  \right).
\]
Since
\[
  E_{\mathrm{disc}}(x;c,h,s,w^{\mathrm{in}})
  :=
  \Pr\Bigl[
    S_{\mathrm{disc}} > c
  \Bigr]
  \leq
  \Pr\bigl[|S_{\mathrm{disc}}|>c\bigr],
\]
we obtain the explicit upper bound.

\smallskip

\emph{(iii) Aliasing error $E_{\mathrm{alias}}$.}
Similarly, the aliasing error corresponds to the event that the sum of
the discretized blanket variables leaves the outer FFT window. Using
the selector representation,
\[
  \sum_{i=2}^M Z_i^{\mathrm{di}}
  \;\stackrel{d}{=}\;
  \sum_{i=2}^n B_i Z_i^{\mathrm{di}},
\]
where the truncated discretized variables $Z_i^{\mathrm{di}}$ lie in the
bounded interval $[s,s+w^{\mathrm{in}}]$. Setting
$Y_i := B_i Z_i^{\mathrm{di}} - \mathbb E[B_i Z_i^{\mathrm{di}}]$ we
obtain independent, bounded variables with $|Y_i|\le K_{\mathrm{alias}}$
and variance at most $v_{\mathrm{alias}}^2$, both computable from the
truncated blanket distribution and $\gamma$. Bernstein's inequality
applied to $S_{\mathrm{alias}} := \sum_{i=2}^n Y_i$ then yields an
explicit exponentially small upper bound of the form
\[
  E_{\mathrm{alias}}
  :=
  \Pr\Bigl[
    \sum_{i=2}^n B_i Z_i^{\mathrm{di}}
    \notin
    [\,-w^{\mathrm{out}}/2,\, w^{\mathrm{out}}/2\,]
  \Bigr]
  \;\le\;
  2\exp\!\left(
    -\frac{(w^{\mathrm{out}}/2)^2}{
      2(n-1)v_{\mathrm{alias}}^2
      + \frac{2}{3}K_{\mathrm{alias}} w^{\mathrm{out}}/2
    }
  \right).
\]
That is, the aliasing error can also be explicitly controlled.

In particular, all three error contributions
$E_{\mathrm{trunc}},E_{\mathrm{disc}},E_{\mathrm{alias}}$ admit explicit,
computable upper bounds in terms of the CDFs
$F_{\mathrm {ref}}$ and the numerical parameters
$(n,\gamma,s,w^{\mathrm{in}},h,w^{\mathrm{out}})$. Hence the total
error terms $\delta_{\mathrm{err}}^{\mathrm{low}}$ and
$\delta_{\mathrm{err}}^{\mathrm{up}}$ in
Theorem~\ref{thm:fft-blanket-error} can be made as small as desired by
a suitable choice of these knobs.

Here, we outline a systematic procedure for selecting the FFT
parameters given $\alpha,\eta_{\mathrm{trunc}},\eta_{\mathrm{disc}},\eta_{\mathrm{alias}},\eta_{\mathrm{main}}$.
By Lemma~\ref{lem:blanket-asymptotics-general}, we may write the blanket
divergence in the form\footnote{We use \eqref{eq:A-n-expansion} to be more precise.}
\[
  D^{\mathrm{blanket}}_{\mathrm e^\varepsilon,n,\mathcal R_{\mathrm{ref}},\gamma}
  =
  D_n(\varepsilon)\,\bigl(1+o(1)\bigr)
  \qquad\text{as } n\to\infty.
\]
Set
\[
  \varepsilon_n
  =
  \frac{\alpha}{\chi}\sqrt{\frac{\log n}{n}}.
\]
Moreover, Lemma~\ref{lem:band-prob-fixed-x-bigO} yields
\[
  E_{\mathrm{main}}(c_n) = \eta_{\mathrm{main}}\,D_n(\varepsilon_n)
\]
for the choice
\[
  c_n
  :=
  \frac{\sigma^2}{4\alpha^2}\,\frac{\eta_{\mathrm{main}}}{\log n},
\]
which is obtained by solving for $c_n$ from the leading term in
\eqref{eq:band-window-bigO}.
Note that $E_{\mathrm{main}}$ is not a numerical approximation error term
(in particular, it is not included in
$\delta_{\mathrm{err}}^{\mathrm{up}}$ or
$\delta_{\mathrm{err}}^{\mathrm{low}}$), so we do not need to enforce a
strict upper bound on $E_{\mathrm{main}}$; it suffices to match its
leading term to the target scale $D_n(\varepsilon_n)$.

We then choose the remaining FFT parameters in the following order.

\begin{enumerate}
  \item \textbf{Truncation window.}
  First we fix truncation parameters $s_n$ and $w_n^{\mathrm{in}}$ and let
  \[
    q_n
    :=
    \Pr\bigl[
      l_\varepsilon(Y;x_1,x_1',\mathcal R_{\mathrm{ref}})
      \notin [s_n,s_n+w_n^{\mathrm{in}}]
    \bigr].
  \]
  Using the exact expression
  \[
    E_{\mathrm{trunc}}
    = 1 - (1-\gamma q_n)^{n-1},
  \]
  we choose $(s_n,w_n^{\mathrm{in}})$ so that
  \[
    E_{\mathrm{trunc}}
    \le
    \eta_{\mathrm{trunc}}\,D_n(\varepsilon_n)
  \]
  for a prescribed relative budget $\eta_{\mathrm{trunc}}>0$.

  \item \textbf{Discretization step size.}
  Next we fix a grid width $h_n>0$. By the Bernstein-type bound above,
  the discretization error satisfies
  \[
    E_{\mathrm{disc}}
    \le
    2\exp\!\left(
      -\frac{c_n^2}{
        2(n-1)v_{\mathrm{disc}}^2
        + \frac{2}{3}K_{\mathrm{disc}} c_n
      }
    \right),
  \]
  with $v_{\mathrm{disc}}^2\le\gamma h_n^2/4$ and $K_{\mathrm{disc}}=h_n/2$.
  Hence
  \[
    E_{\mathrm{disc}}
    \le
    2\exp\!\left(
      -\frac{c_n^2}{
        \frac{\gamma}{2}\,n h_n^2 + \frac{1}{3}\,c_n h_n
      }
    \right),
  \]
  We then choose $h_n$ so that this upper bound satisfies
  \[
    E_{\mathrm{disc}}
    \le
    \eta_{\mathrm{disc}}\,D_n(\varepsilon_n)
  \]
  for a prescribed $\eta_{\mathrm{disc}}>0$.

  \item \textbf{Outer FFT window.}
  Finally, for the aliasing (wrap-around) error we have the bound
  \[
    E_{\mathrm{alias}}
    \le
    2\exp\!\left(
      -\frac{(w_n^{\mathrm{out}/2})^2}{
        2(n-1)v_{\mathrm{alias}}^2
        + \frac{2}{3}K_{\mathrm{alias}} w_n^{\mathrm{out}}/2
      }
    \right),
  \]
  where $v_{\mathrm{alias}}^2$ and $K_{\mathrm{alias}}$ depend only on the
  truncated blanket distribution and $\gamma$. Once $h_n$ and
  $w_n^{\mathrm{in}}$ are fixed as above, $v_{\mathrm{alias}}$ and
  $K_{\mathrm{alias}}$ can be computed explicitly from the truncated
  distribution. We then choose $w_n^{\mathrm{out}}$ so that
  \[
    E_{\mathrm{alias}}
    \le
    \eta_{\mathrm{alias}}\,D_n(\varepsilon_n)
  \]
  for a prescribed $\eta_{\mathrm{alias}}>0$.
\end{enumerate}


\subsection{Computation of CDFs}

The algorithms~\ref{alg:main_term_random_m}--\ref{alg:calc_final_prob}
require as input the CDFs of the privacy amplification random variable.
In the discrete case (e.g., $k$-randomized response), these CDFs can be
computed by direct enumeration.

In the generalized Gaussian setting, the map $l_\varepsilon(y;x_1,x_1',\mathcal R_{\mathrm{ref}})$ is strictly decreasing on $\mathbb R$ whenever\footnote{%
  This restriction on the parameter range is not essential.
  When $y\mapsto l_\varepsilon(y;x_1,x_1',\mathcal R_{\mathrm{ref}})$ is not monotone,
  one can still represent the CDF of $l_\varepsilon(Y)$ as a finite sum
  of terms $F_{\mathrm{ref}}(b_j(u))-F_{\mathrm{ref}}(a_j(u))$ over the
  finitely many intervals $[a_j(u),b_j(u)]$ on which
  $l_\varepsilon(y)\le u$, where each endpoint $a_j(u),b_j(u)$ can be
  computed by one-dimensional root finding.
  However, in the moderate-deviation regime $\varepsilon_n\to0$
  considered in this paper, the monotonicity condition is satisfied for
  all sufficiently large $n$, so we impose it in the main text for
  simplicity.
}
\[
  \varepsilon \le \ln 2 - \left(\frac{x_1-x_1^\prime}{2\alpha}\right)^\beta,
\]
Under this condition the privacy amplification map is invertible, and
for any reference distribution $Y\sim\mathcal R_{\mathrm{ref}}$ with CDF
$F_{\mathrm{ref}}$, the CDF of
\[
  X := l_\varepsilon(Y;x_1,x_1',\mathcal R_{\mathrm{ref}})
\]
is given by the simple composition formula
\[
  F_X(x)
  = \Pr[X \le x]
  = \Pr\bigl[Y \ge l_\varepsilon^{-1}(x)\bigr]
  = 1 - F_{\mathrm{ref}}\!\bigl(l_\varepsilon^{-1}(x)\bigr).
\]
Thus the only remaining technical task is to evaluate the inverse map
$l_\varepsilon^{-1}$.

For generalized Gaussian mechanisms with shape parameters
$\beta=1$ (Laplace) and $\beta=2$ (Gaussian), the inverse
$l_\varepsilon^{-1}$ can be expressed in closed form, so that
$F_X(x)$ is available analytically (in terms of $\Phi$ in the Gaussian
case). For general $\beta\in(1,2)$, the inverse is no longer available
in closed form, but it can be computed numerically by solving a
one-dimensional root-finding problem for each $x$, which yields
arbitrarily accurate values of $F_X(x)$ at a negligible cost compared
to the FFT step.



\end{document}